\documentclass[aps,physrev,reprint,superscriptaddress,floatfix]{revtex4-2}

\usepackage{mathtools,amssymb}
\usepackage{amsthm}
\usepackage[caption=false]{subfig}
\usepackage[dvipsnames]{xcolor}
\usepackage{nicematrix}
\usepackage{tikz}
\usetikzlibrary{cd}

\usepackage[linesnumbered, ruled, vlined]{algorithm2e}
\SetKwInput{KwInput}{Input}
\SetKwInput{KwOutput}{Output}

\usepackage{graphicx}

\usepackage{hyperref}
\hypersetup{
    colorlinks=true,
    linkcolor=blue,
    citecolor=magenta,
    filecolor=magenta,
    urlcolor=blue,
    pdftitle={},
    pdfpagemode=FullScreen,
    }
\usepackage[capitalise]{cleveref}

\DeclareMathOperator{\im}{im}
\DeclareMathOperator{\supp}{supp}

\theoremstyle{plain}
\newtheorem{theorem}{Theorem}
\newtheorem{lemma}{Lemma}

\newtheorem{example}{Example}

\theoremstyle{definition}
\newtheorem{definition}{Definition}
\newtheorem{remark}{Remark}

\DeclareMathOperator{\wt}{wt}

\DeclareMathOperator*{\argmin}{arg\,min}

\begin{document}

\title{Fault-Tolerant Logical Measurements via Homological Measurement}
\author{Benjamin Ide}
\thanks{Corresponding author: \href{mailto:ben.ide@xanadu.ai}{ben.ide@xanadu.ai}.}
\affiliation{Xanadu, Toronto, Ontario M5G 2C8, Canada}
\author{Manoj G.\ Gowda}
\affiliation{Xanadu, Toronto, Ontario M5G 2C8, Canada}
\author{Priya J.\ Nadkarni}
\affiliation{Xanadu, Toronto, Ontario M5G 2C8, Canada}
\author{Guillaume Dauphinais}
\affiliation{Xanadu, Toronto, Ontario M5G 2C8, Canada}

\date{\today}

\begin{abstract}
We introduce homological measurement, a framework for measuring the logical Pauli operators encoded in CSS stabilizer codes. The framework is based on the algebraic description of such codes as chain complexes. Protocols such as lattice surgery and some of its recent generalizations are shown to be special cases of homological measurement. Using this framework, we develop a specific protocol called \emph{edge expanded homological measurement} for fault-tolerant measurement of arbitrary logical Pauli operators of general qLDPC codes, requiring a number of ancillary qubits growing only linearly with the weight of the logical operator measured, and guarantee that the distance of the code is preserved. We further benchmark our protocol numerically in a photonic architecture based on GKP qubits, showing that the logical error rates of various codes are on par with other methods requiring more ancilla qubits.
\end{abstract}

\maketitle

\section{Introduction}
\label{sec:intro}

Due to the fragile nature of quantum systems, one of the most promising approaches to build a useful fault-tolerant large scale quantum computer is to encode the information using stabilizer codes~\cite{Shor1995,gottesman1997stabilizer}. Two-dimensional topological codes~\cite{KITAEV2003-2,Bombin2006} have long been studied due to their relatively high resilience to noise as evidenced by their high thresholds~\cite{Dennis2002,Bombin2012} and their planar layout, however they suffer from a low encoding rate~\cite{BPT-2010}. This contributes to a large qubit overhead requirement in practical applications using these codes~\cite{Gidney2021, PhysRevResearch.4.023019}.

More general quantum low density parity check (qLDPC) codes have recently attracted significant interest~\cite{Kovalev2012,Tillich2014,Gottesman2014,Bravyi2014,Leverrier2015,Breuckmann2016,Breuckmann2017,Fawzi2018,dasilva2018,Hastings2020,Evra2020,Panteleev2021,Breuckmann2021Balanced, Panteleev2022,Vuillot2022,dinur2022good,Strikis2023, Roffe2023biastailoredquantum,yang2023spatiallycoupled,feng2023new,miao2024joint,sabo2024weight}. While in general they require some degree of long range connectivity, they do not suffer from the same bounds relating their size, distance, and rate as their local counterparts, allowing for the design of qLDPC codes that require significantly less overhead.

Photonics platforms allow for flexibility of connectivity since components of the quantum processor can be connected via optical fibers~\cite{Bourassa2021blueprintscalable,tzitrin2021staticlinear, bartolucci2021fusionbased,pankovich2023high}, thus they are an ideal match for qLDPC codes. The optical nature of the qubits used also makes measurement-based quantum computing~\cite{Raussendorf2005,RAUSSENDORF2006,Raussendorf2007-1,Raussendorf2007-2,Bolt2016} a natural choice. In some photonic architectures~\cite{tzitrin2021staticlinear, walshe2024}, the effective noise experienced by a qubit scales with the number and weight of stabilizers in which the qubit participates. Further, the dominant source of noise in this architecture is photon loss, and this physical source of noise works nicely with assumptions of fault-tolerance. As such, this work is further motivated by specific needs of this kind of photonics-based architecture to develop fault-tolerant, qLDPC, low space overhead measurement protocols.

A key ingredient to run a useful algorithm in such a platform is the ability to perform logical Pauli measurements. Clifford gates can be implemented using joint logical Pauli measurements alongside the ability to prepare fault-tolerantly encoded Pauli eigenstates~\cite{Litinski2019}. Additionally, having access to some encoded non-Clifford states is sufficient to achieve a universal gate set~\cite{Bravyi2005}. While lattice surgery techniques have been widely studied for topological codes~\cite{Horsman_2012,landahl2014quantum,Litinski2018latticesurgery,PoulsenNautrup2017}, relatively little work has focused on more general qLDPC codes. Homomorphic logical measurements~\cite{Huang2023} and generalized lattice surgery~\cite{Cohen2021} both maintain qLDPC properties of the codes they act on, but require a large additional overhead scaling at least quadratically with the distance, and CSS code surgery~\cite{cowtan2023css} cannot guarantee fault-tolerance, as well as being limited to measuring CSS logical Pauli operators. The overhead of~\cite{Cohen2021} was improved in~\cite{cross2024linearsize}, but it is only improved to linear overhead in cases where the logical operators have a certain structure related to edge expansion.

In this work, we introduce the mathematical framework of \emph{homological measurement} based on the chain complex description of CSS codes. Using this framework, we develop a scheme to measure any multiqubit logical Pauli operator with a number of additional physical qubits scaling linearly with the weight of the logical operator being measured. Given a code with a logical operator to measure, our protocol gives a fault-tolerant method of extending the code while including the logical operator in the stabilizer group.

Homological measurement hinges on the fact that a CSS code $\mathcal C$ can naturally be described by a chain complex of length $2$, which comprises $3$ vector spaces. One of these describes the physical qubits of the system, while the other $2$ spaces describe the $X$ and $Z$ syndrome spaces. Linear maps connecting these spaces, called boundary maps, are given by the $X$ and $Z$ type check matrices.

Since CSS codes can be described by chain complexes, homological algebra~\cite{weibel1994introduction} can be used to describe features of CSS codes and operations on CSS codes. For example, logical operators are given by the homology and cohomology groups of the chain complex. One such operation that is used to great effect in this work is the \emph{mapping cone} applied to maps between chain complexes called \emph{chain maps}.

The main idea behind homological measurement is to create an ancillary quantum code also described via a chain complex $\mathcal A$ and a chain map $f$ from $\mathcal A$ to $\mathcal C$. A new CSS code is then available via the chain complex resulting from the mapping cone of $f$. The logical operators in this new code can be studied via the mathematics of the well-established literature on homological algebra~\cite{weibel1994introduction}. Various known fault-tolerant logical measurement schemes can be described straightforwardly as homological measurements, including lattice surgery~\cite{Horsman_2012} and generalized lattice surgery~\cite{Cohen2021,cross2024linearsize}. The mapping cone construction has also been been used to develop methods to reduce the weight of stabilizer generators of CSS codes~\cite{hastings2016, hastings2023quantum, sabo2024weight, wills2023tradeoff}. By carefully choosing the chain complex $\mathcal A$ and chain map $f$, we can engineer a system to perform fault-tolerant logical Pauli operator measurements for arbitrary CSS qLDPC codes with fewer resources than previous methods. Additionally, by exploiting edge expansion of a graph structure implied by the chain complex $\mathcal A$, we are able to maintain the distance with much smaller ancillary systems compared to previous methods. We call this \emph{edge expanded homological measurement}.

A self-contained introduction to the relevant concepts of graph theory and homological algebra is presented in~\cref{sec:Backround}, followed by the description of the homological measurement framework as well as our protocol using this framework in~\cref{sec:hom-meas}. Numerical simulations benchmarking our protocol for the measurement-based implementation of various quantum error correcting codes using optical Gottesman-Kitaev-Preskill qubits~\cite{GKP} are presented in~\cref{sec:numerics}, followed by concluding remarks in~\cref{sec:discussion}.

\section{Background and Notation} \label{sec:Backround}

\subsection{Graph theory}\label{sec:graph-theory}

The construction in this work relies on the notion of a hypergraph which generalizes graphs so that ``edges'' are arbitrary subsets of a vertex set.

\begin{definition}
  A \emph{hypergraph} $\mathcal G(V, E)$ is given by a vertex set $V$ and an edge set $E \subseteq \mathcal P(V)$ where $\mathcal P(V)$ is the power set of $V$.
\end{definition}

Graphs and hypergraphs will mostly be given by incidence matrices.

\begin{definition}
  An \emph{edge-vertex incidence matrix} $M$ for a (hyper)graph $\mathcal G(V, E)$ is a $|E|\times|V|$ matrix where rows represent edges and columns represent vertices given by
  \begin{equation}
      M_{i,j} = \begin{cases}
          1 & v_j \in e_i \\
          0 & \text{otherwise.}
      \end{cases}
  \end{equation}
\end{definition}

In this manuscript, when the term \emph{incidence matrix} is used without qualification, we mean the \emph{edge-vertex incidence matrix}. We allow edges to be repeated, so the edge set $E$ could be a multiset, and correspondingly, the incidence matrix could have multiple identical rows.

 We use the following generalization to hypergraphs for the Cheeger constant:
\begin{definition}
  Consider the incidence matrix \(M\) of a (hyper)graph \(\mathcal G(V, E)\). The \emph{boundary} of \(S \subset V\) is
  \begin{equation}
    \partial S = \{e \in E\ \big\vert\ |e \cap S|\text{ is odd}\}.
  \end{equation}
  The \emph{Cheeger constant} of \(\mathcal G\) is given by
  \begin{equation}
  \begin{aligned}
    h(\mathcal G) &= \min \left\{\frac{|\partial S|}{|S|}\ \bigg\vert\ S \subset V,\ |S| \leq |V|/2\right\} \\
    &= \min \left\{\frac{\wt(Mv)}{\wt(v)}\ \bigg\vert\ \wt(v) \leq |V|/2 \right\}.
  \end{aligned}
  \end{equation}
\end{definition}
This definition matches the usual definition for graphs. Further, if the same matrix $M$ is instead interpreted as a biadjacency matrix for a Tanner graph where rows represent checks and columns represents variables, this matches the Cheeger constant definition in \cite{cross2024linearsize}.

\begin{definition} \label{def:cycle-basis}
    If $M$ is an incidence matrix, then a \emph{cycle basis} is given by the rows of any full-rank matrix $N$ such that $\ker N \cong \im M$. The matrix $N$ can be seen as a cycle-edge incidence matrix where the edge given by the $i$-th row of $M$ is given by the $i$-th column of $N$, and a row of $N$ defines the edges of a cycle in the cycle basis.
\end{definition}

\subsection{Homological algebra}\label{app:mappingcone}

The following is designed to be mostly self-contained, but interested readers should consult standard references to learn more~\cite{rotman2009introduction, weibel1994introduction}.

A \emph{chain complex} \(\mathcal C\) is an ordered sequence of finite dimensional vector spaces $\{C_i\}$ over $\mathbb{F}_2$ with linear transformations between consecutive spaces, \(\partial_{i+1}:\ C_{i+1} \to C_i\), such that $\partial_i \circ \partial_{i + 1} = 0$:
\begin{equation*}
    \mathcal C: \quad \cdots \to C_{i + 1} \xrightarrow{\partial_{i + 1}} C_i \xrightarrow{\partial_i} C_{i - 1} \to \cdots .
\end{equation*}
We are only interested in chain complexes with a finite number of vector spaces, i.e., for indices $m, m+1, \dots, n$:
\begin{equation*}
    \mathcal C: \quad C_{n} \xrightarrow{\partial_{n}} \cdots \xrightarrow{\partial_{m + 2}} C_{m + 1} \xrightarrow{\partial_{m + 1}} C_{m} .
\end{equation*}
A chain complex with $\ell$ spaces is called an $\ell$-term chain complex. Since $\partial_i \circ \partial_{i + 1} = 0$, we have $\im\left(\partial_{i + 1}\right) \subseteq \ker\left(\partial_i\right)$. The $i$th \emph{homology group} is defined as $H_i(\mathcal C) = \ker \left(\partial_i\right) / \im\left(\partial_{i + 1}\right)$. For a chain with a finite number of terms, to calculate the homology groups assume that any unspecified spaces are the trivial space $0$ with trivial maps. Note that throughout this work, the symbol $H$ is also used for parity-check and stabilizer matrices but the usage is clear from context. The dual of a chain complex \(\mathcal C\) is a the \emph{cochain complex} \(\mathcal C^{*}\):
\begin{equation*}
    \mathcal C^{*}: \quad \cdots \leftarrow C^*_{i + 1} \xleftarrow{\delta_{i}} C^*_i \xleftarrow{\delta_{i-1}} C^*_{i - 1} \leftarrow \cdots 
\end{equation*}
where $C_i^*$ is the space of all linear functionals on $C_i$. Note that $C_i^* \cong C_i$ since we only consider finite dimensional vector spaces. The corresponding dual of the homology groups are the \emph{cohomology groups} ${H^i(\mathcal C) = \ker \delta_{i} / \im \delta_{i-1}}$. We assume that vector spaces and their duals are equipped with the standard basis so that, when viewed as matrices, $\delta_{i-1}^{\mathrm T} = \partial_i$.

A chain is said to be \emph{exact} at $C_i$ if $\im\left(\partial_{i + 1}\right) = \ker\left(\partial_i\right)$ and is said to be an \emph{exact sequence} if it is exact at each $C_i$. An exact sequence of the form $0 \to X \to Y \to Z \to 0$ is called a \emph{short exact sequences} (SES). Any exact sequence with more than three nontrivial spaces is called a \emph{long exact sequence} (LES).

Suppose that ${0 \to V_1 \to \cdots \to V_r \to 0}$ is an exact sequence of finite dimensional vector spaces. Then
\begin{equation}\label{eq:rank-nullity-gen}
    \sum_{i=1}^r (-1)^i \dim V_i = 0.
\end{equation}
This is a generalization of the rank-nullity theorem.

Consider two chain complexes $\mathcal{A}$ and $\mathcal{C}$ and maps $f_i: A_i \to C_i$ linking them:
\begin{equation}\label{eq:chainmap}
    \begin{tikzcd}
        {\mathcal{A}:} & {\cdots} & {A_i} & {A_{i - 1}} & \cdots \\
        {\mathcal{C}:} & {\cdots} & {C_i} & {C_{i - 1}} & \cdots
        \arrow[from=1-2, to=1-3]
        \arrow[from=2-2, to=2-3]
        \arrow["{\partial^{\mathcal A}_i}", from=1-3, to=1-4]
        \arrow["{\partial^{\mathcal C}_i}", from=2-3, to=2-4]
        \arrow["{f_i}", from=1-3, to=2-3]
        \arrow["{f_{i - 1}}", from=1-4, to=2-4]
        \arrow[from=1-4, to=1-5]
        \arrow[from=2-4, to=2-5]
        \arrow["{f}"', from=1-1, to=2-1]
    \end{tikzcd}
\end{equation}

When $\partial^{\mathcal C}_{i}(f_{i}(a)) = f_{i-1}\left(\partial^{\mathcal A}_{i}(a)\right)$ for all $a \in A_{i}$, the maps $f_i$ are collectively called a \emph{chain map} \(f: \mathcal A \to \mathcal C\). This condition is often written more compactly as ${\partial^{\mathcal C}f = f\partial^{\mathcal A}}$. The \emph{mapping cone} of \(f\) is defined to be the chain complex with spaces $\mathrm{cone}(f)_i = C_i \oplus A_{i-1}$ and maps $\partial_i^{\mathrm{cone}(f)}: \mathrm{cone}(f)_i \to \mathrm{cone}(f)_{i - 1}$
\begin{equation}\label{eq:cone}
    \begin{gathered}
    \begin{tikzcd}
        {A_{i-1}} & {A_{i-2}}  \\
        {C_i} & {C_{i-1}}  \\
        {\mathrm{cone}(f)_i} & {\mathrm{cone}(f)_{i-1}}
        \arrow["{f_{i-1}}", from=1-1, to=2-2]
        \arrow["{\partial^{\mathcal C}_i}"', from=2-1, to=2-2]
        \arrow["{\partial^{\mathcal A}_{i-1}}", from=1-1, to=1-2]
        \arrow["{\partial_i^{\mathrm{cone}(f)}}", from=3-1, to=3-2]
        \arrow["\oplus", phantom, from=1-1, to=2-1]
        \arrow["\oplus", phantom, from=1-2, to=2-2]
    \end{tikzcd}\\
    \partial_i^{\mathrm{cone}(f)} = \, \begin{pNiceArray}{cc}[first-row, first-col]
        & C_i & A_{i-1} \\
        C_{i-1} & \partial^{\mathcal C}_i & -f_{i-1} \\
        A_{i-2} & 0 & -\partial^{\mathcal A}_{i-1}
    \end{pNiceArray}.
    \end{gathered}
\end{equation}
In this paper we restrict to $\mathbb F_2$, so the minus signs in the map above will be omitted.

The mapping cone gives the short exact sequence
\begin{equation}\label{SES}
    0 \to C_k \xrightarrow{\iota} \mathrm{cone}(f)_k \xrightarrow{\pi} A_{k-1} \to 0
\end{equation}
for all indices $k$ where \(\iota\) is inclusion and \(\pi\) is projection. Applying the Snake Lemma~\cite{rotman2009introduction, weibel1994introduction} with the chain map $f$ as the connecting map, we obtain the following long exact sequence on homology where $f_*$, $\iota_*$, and $\pi_*$ are maps derived from $f$, $\iota$, and $\pi$:
\begin{equation}\label{LES}
    \begin{tikzcd}
        {\cdots} & {H_{k + 1}(\mathrm{cone}(f))} & {H_k(\mathcal{A})} \\
        {H_k(\mathcal{C})} & {H_k(\mathrm{cone}(f))} & {\cdots}
        \arrow["{\iota_*}", from=1-1, to=1-2]
        \arrow["{\pi_*}", from=1-2, to=1-3]
        \arrow["{f_*}"', from=1-3, to=2-1]
        \arrow["{\iota_*}"', from=2-1, to=2-2]
        \arrow["{\pi_*}"', from=2-2, to=2-3]
    \end{tikzcd}
\end{equation}
This exact sequence allows us to study the logical operators that remain after logical measurement is conducted via the mapping cone.

\subsection{Codes as chain complexes}\label{sec:codes-as-chains}

There is a natural correspondence between codes and chain complexes. Let $H \in \mathbb{F}^{n - k \times n}_2$ be a parity-check matrix of an $[n, k, d]$ (classical) linear code, where $n$ is the number of variables, $k$ is the dimension, and $d$ is the distance. Then we can express this code as
\begin{equation}\label{classicalchain}
    \mathbb{F}^n_2 \xrightarrow{H} \mathbb{F}^{n - k}_2,
\end{equation}
which vacuously satisfies the definition of a chain complex. Conversely, any chain complex with only one map gives a classical linear code.

Now consider an $[\![n, k, d]\!]$ CSS code. Define $d_X$ to be the minimum weight of nontrivial $X$-logical operators and $d_Z$ to be the minimum weight of nontrivial $Z$-logical operators. Then $n$ is the number of physical qubits, $k$ is the number of logical qubits, and $d = \min\{d_X, d_Z\}$ is the distance. A CSS code $\mathrm{CSS}(H_X, H_Z)$ is generated by $n_X$ $X$-stabilizer generators, not necessarily independent, given by rows of the matrix $H_X$ and $n_Z$ $Z$-stabilizer generators given by rows of $H_Z$. Since $H_Z H_X^{\mathrm T} = 0$, we can treat this as the chain complex
\begin{equation}\label{CSSchain}
    \mathbb{F}^{n_X}_2 \xrightarrow{H_X^{\mathrm T}} \mathbb{F}^n_2 \xrightarrow{H_Z} \mathbb{F}^{n_Z}_2
\end{equation}
with the qubits in the center. Conversely, we can derive a CSS code from any two consecutive boundary maps $\partial_{i+1}$ and $\partial_i$ by setting $H_X^{\mathrm T} = \partial_{i + 1}$ and $H_Z = \partial_i$. The $X$-logical operators commute with the $Z$ stabilizers (i.e., they are in $\ker H_Z$) and are not $X$ stabilizers (i.e., they are not in $\im H_{X}^{\mathrm T}$), which make them nontrivial elements of the first homology group $H_1$. The dual gives the cochain
\begin{equation}\label{CSScochain}
    \mathbb{F}^{n_X}_2 \xleftarrow{H_X} \mathbb{F}^n_2 \xleftarrow{H^{\mathrm T}_Z} \mathbb{F}^{n_Z}_2,
\end{equation}
which implies that the $Z$ logicals are given by the first cohomology group $H^1$. Since these descriptions of the CSS code are duals of each other, one can also view \cref{CSScochain} as the chain with \cref{CSSchain} as its cochain, which is common in the literature. We choose to consider \cref{CSSchain} as the chain complex because this is natural for measurement of $X$-logical operators.

Since CSS codes and 3-term chain complexes are equivalent, we can create new codes using mapping cones. Consider a chain map ${f: \mathcal A \to \mathcal C}$, where $\mathcal C$ is the chain complex ${C_2 \xrightarrow{H^{\mathrm T}_X} C_1 \xrightarrow{H_Z} C_0}$ given by a CSS code and $\mathcal A$ is a chain complex ${A_1 \xrightarrow{\partial_1} A_0 \xrightarrow{\partial_0} A_{-1}}$. The \emph{cone code} is given by
\begin{equation}\label{conecode}
    \begin{tikzcd}
    & {A_1} & {A_0} & {A_{-1}} \\
	& {C_2} & {C_1} & {C_0}
	\arrow["{\partial_1}", from=1-2, to=1-3]
	\arrow["{\partial_0}", from=1-3, to=1-4]
	\arrow["{H_X^{\mathrm T}}"', from=2-2, to=2-3]
	\arrow["{H_Z}"', from=2-3, to=2-4]
	\arrow["{f_1}", from=1-2, to=2-3]
	\arrow["{f_0}", from=1-3, to=2-4]
    \arrow["\oplus", phantom, from=1-2, to=2-2]
    \arrow["\oplus", phantom, from=1-3, to=2-3]
    \arrow["\oplus", phantom, from=1-4, to=2-4]
    \end{tikzcd}
\end{equation}
which has stabilizer matrices
\begin{equation}\label{conecode-mats}
    \widetilde{H}_X = \begin{pmatrix}
        H_X & \\
        f_1^{\mathrm T} & \partial_1^{\mathrm T}
    \end{pmatrix}, \quad
    \widetilde{H}_Z = \begin{pmatrix}
        H_Z & f_0 \\
            & \partial_0
    \end{pmatrix}.
\end{equation}

\section{Homological measurement} \label{sec:hom-meas}

The goal of this work is to measure logical operators fault-tolerantly while minimizing the number of qubits required. Naively, one would simply measure the desired logical operator directly, removing it from the logical group and adding it to the stabilizer group. However, for codes with large distances, this would result in measuring directly a high-weight operator, a process which is prone to errors. Consequently, protocols have been developed \cite{Horsman_2012,Cohen2021,cross2024linearsize} to infer the value of these higher-weight logical operators by measuring a number of lower-weight operators instead. To accomplish this, one must modify the original code(s) in a way that leaves the logical information left unaffected by the desired measurement. This can be done for any CSS code using the generalized lattice surgery protocol from \cite{Cohen2021}, however, the resource overhead is prohibitively large. The situation is improved in \cite{cross2024linearsize} by exploiting edge expansion that is present in the Tanner graph, but except in some specific cases where the edge expansion is sufficient, this still uses more resources than necessary. The protocol developed in this work uses fewer resources in general by improving the edge expansion directly rather than accepting it as a fixed quantity in a process we call \emph{edge expanded homological measurement}.

Given a CSS code $\mathcal C$ and a logical operator of $\mathcal C$, we design an ancillary code $\mathcal A$ and a chain map \(f: \mathcal A \to \mathcal C\) such that this logical operator is in the stabilizer group of the cone code of $f$. We require that the result of the logical measurement can be inferred correctly despite the presence of errors, and as such, the distance of $\widetilde{\mathcal C}$ should be at least that of $\mathcal C$ and stabilizer weights should be reasonably low. Finally, no other logical operators of $\mathcal C$ should be measured so that the remaining logical information is left intact.

A discussion of the properties of cone codes in general is given in \cref{sec:math-framework-properties}, along with examples of past measurement protocols that fit the cone code framework. The edge expanded homological measurement construction for measuring any $X$- or $Z$-logical operator is given in \cref{sec:single-log} and a protocol for using this construction in a measurement scheme is given in \cref{sec:physical-considerations}. A construction for measuring operators built from a mixture of logical $X$-, $Y$-, and $Z$-Paulis is given in \cref{sec:non-css-meas}.

\subsection{Mathematical properties of cone codes}\label{sec:math-framework-properties}
 
Consider a CSS code given by stabilizer matrices \(H_X\) and \(H_Z\) written as the chain complex
\begin{equation}\label{eqn:code-chain-complex}
    \mathcal C:\quad C_2 \xrightarrow{H^{\mathrm T}_X} C_1 \xrightarrow{H_Z} C_0,
\end{equation}
an ancillary CSS code with $X$-stabilizer matrix $\partial_1^{\mathrm T}$ and $Z$-stabilizer matrix $\partial_0$ written as the chain complex
\begin{equation}\label{eqn:ancilla-chain-complex}
    \mathcal A:\quad A_1 \xrightarrow{\partial_1} A_0 \xrightarrow{\partial_0} A_{-1},
\end{equation}
and a chain map \(f:\,\mathcal A \to \mathcal C\). We will use the notation
\begin{equation}
  \mathrm{cone}(f) = \widetilde{\mathcal C}:\quad \widetilde C_2 \xrightarrow{\widetilde H^{\mathrm T}_X} \widetilde C_1 \xrightarrow{\widetilde H_Z} \widetilde C_0
\end{equation}
for the cone code of \(f\).

Before constructing a specific \(\mathcal A\) and \(f\), we first make a few general observations. We choose to interpret the resulting cone code as a subsystem code, and so $\dim H_1(\widetilde{\mathcal C}) = \widetilde k + \widetilde r$ where $\widetilde k$ is the number of logical qubits and $\widetilde r$ is the number of gauge qubits. Applying the generalization of the rank-nullity theorem given in \cref{eq:rank-nullity-gen} to the long exact sequence of \cref{LES} for $\widetilde{\mathcal C}$, and treating any of the nontrivial operators of the cone code that arise from $\mathcal A$ as gauges, we have the following:
\begin{align}
  &\widetilde k = k + \left( \dim \ker \widetilde H_X^{\mathrm T} - \dim \ker H_X^{\mathrm T} \right) - \dim \ker \partial_1 \label{change-in-k}\\
  &\begin{aligned}[b]
      \widetilde r = &\dim \left( \ker \partial_0 / \im \partial_1 \right) + \left( \dim \ker \widetilde H_Z^{\mathrm T} - \dim \ker \partial_0^{\mathrm T} \right) \\ & - \dim \ker H_Z^{\mathrm T}.
  \end{aligned} \label{eq:new-gauges}
\end{align}
The term $\left( \dim \ker \widetilde H_X^{\mathrm T} - \dim \ker H_X^{\mathrm T} \right)$ measures the amount of redundancy of the $X$-stabilizer generators in $\widetilde{\mathcal C}$ that were not directly copied from $\mathcal C$. \Cref{lem:ker-partial1} below shows that this remaining redundancy is isomorphic to a subspace of $\ker \partial_1$ and that the remaining elements of $\ker \partial_1$ contribute to reduction in $\widetilde k$. \Cref{eq:new-gauges} mirrors \cref{change-in-k} with the roles of $\mathcal A$ and $\mathcal C$ swapped and the roles of $X$ and $Z$ swapped, noting in particular that the number of logical qubits in $\mathcal A$ is $\dim \left( \ker \partial_0 / \im \partial_1 \right)$. When applying a cone code construction to measure logical operators, as it is used in \cref{sec:single-log}, the reduction of $\widetilde k$ corresponds to the measurement of a logical operator.

\begin{lemma}[Characterization of $\ker \partial_1$]\label{lem:ker-partial1}
    Given a CSS code $\mathcal C = \mathrm{CSS}(H_X, H_Z)$ interpreted as a chain complex $\mathcal C:C_2 \xrightarrow{H_X^{\mathrm T}} C_1 \xrightarrow{H_Z} C_0$, a chain complex ${\mathcal A:A_1 \xrightarrow{\partial_1} A_0 \xrightarrow{\partial_0} A_{-1}}$, and a chain map ${f: \mathcal A \to \mathcal C}$,
    \begin{equation}\label{eq:ker-partial1}
        \ker \partial_1 \cong \mathcal S \oplus \mathcal L
    \end{equation}
    where $\mathcal S$ is the space of $X$-stabilizers of $\mathcal C$ that can be measured via the new stabilizer generators of $\widetilde H_X$ and $\mathcal L$ is the space of nontrivial $X$-logical operators of $\mathcal C$ that can be measured via the new stabilizer generators of $\widetilde H_X$.
\end{lemma}
\begin{proof}
    Recall that the mapping cone of $f$ gives a CSS code $\widetilde{\mathcal C}$ with stabilizer matrices
    \begin{equation}
        \widetilde{H}_X = \begin{pmatrix}
            H_X & \\
            f_1^{\mathrm T} & \partial_1^{\mathrm T}
        \end{pmatrix}, \quad
        \widetilde{H}_Z = \begin{pmatrix}
            H_Z & f_0 \\
                & \partial_0
        \end{pmatrix}.
    \end{equation}
    Let $\vec v \in \ker \partial_1$ and define
    \begin{equation}
        S_{\vec v} = \left(\begin{pmatrix} f_1 \\ \partial_1 \end{pmatrix} \vec v \right)^{\mathrm T}.
    \end{equation}
    By construction, $S_{\vec v}$ is a sum of some of the new $X$-type stabilizer generators of $\widetilde{\mathcal C}$ and therefore is an $X$ stabilizer of $\widetilde{\mathcal C}$. Since $\vec v \in \ker \partial_1$, $S_{\vec v}$ has no support on the ancillary qubits.
    
    Since $\widetilde H_Z$ restricted to the original qubits of $\mathcal C$ is simply $H_Z$, we know that $S_{\vec v}|_{\mathcal C}$ represents an $X$ operator that commutes with $H_Z$. Thus, $S_{\vec v}$ is either an $X$ stabilizer or a nontrivial $X$ logical of $\mathcal C$.
\end{proof}

One on the uses of \cref{lem:ker-partial1} is to show that our protocol does not inadvertently measure extra logical operators by enforcing $\dim \ker \partial_1 = 1$.

Even with this generality the \(Z\)-distance of \(\widetilde{\mathcal C}\) is maintained:

\begin{theorem}[$Z$-distance]\label{thm:z-dist}
  If the non-trivial operators of the cone code in \cref{conecode-mats} resulting from the logicals of the ancilla code are treated as gauges as in \cref{change-in-k,eq:new-gauges}, then the dressed \(Z\)-distance \(\widetilde d_Z\) is maintained, i.e., \(\widetilde d_Z \geq d_Z\).
\end{theorem}
\begin{proof}
  The dressed \(Z\)-distance is the same as the \(Z\)-distance of the stabilizer code where a full basis of the \(Z\)-gauges are included as stabilizers. To that end, since all gauges arise from logicals of the ancilla code, replace \(\partial_0\) with a matrix \(G\) such that \(\ker G \cong \im \partial_1\) leaving us to study the following stabilizer code:
  \begin{equation}\label{eq:conecodeZ}
    \widetilde H_X = \begin{pmatrix}
                       H_X & \\
                       f_1^{\mathrm T} & \partial_1^{\mathrm T}
                     \end{pmatrix}, \quad
    \widetilde H_Z' = \begin{pmatrix}
                       H_Z & f_0 \\
                        & G
                     \end{pmatrix}.
  \end{equation}

  Let \(v = \begin{pmatrix}v_C \\ v_A\end{pmatrix} \in \ker \widetilde H_X\) where \(v_C\) is defined on the block of qubits from \(\mathcal C\) and \(v_A\) is on the block of qubits from \(\mathcal A\). This implies
  \begin{equation}\label{eq:zdist1}
    H_Xv_C = 0
  \end{equation}
  and
  \begin{equation}\label{eq:zdist2}
    f_1^{\mathrm T}v_C + \partial_1^{\mathrm T}v_A = 0.
  \end{equation}

  From \eqref{eq:zdist1}, if \(v_C^{\mathrm T}\) is not a \(Z\)-stabilizer of \(\mathcal C\), then it must be a logical of \(\mathcal C\) and so \(\wt(v) \geq \wt(v_C) \geq d_Z\). Suppose that \(v_C^{\mathrm T}\) is a \(Z\)-stabilizer of the original code, therefore there is a vector \(w_1\) such that \(v_C = H_Z^{\mathrm T}w_1\). Since \(f\) is a chain map, we know that \(H_Zf_1 = f_0\partial_1\), or equivalently, \(f_1^{\mathrm T}H_Z^{\mathrm T} = \partial_1^{\mathrm T}f_0^{\mathrm T}\). Using this,
  \begin{equation}
    f_1^{\mathrm T}v_C = f_1^{\mathrm T}H_Z^{\mathrm T}w_1 = \partial_1^{\mathrm T}f_0^{\mathrm T}w_1,
  \end{equation}
  and so from \eqref{eq:zdist2}, we find
  \begin{equation}
    \partial_1^{\mathrm T}\left(f_0^{\mathrm T}w_1 + v_A\right) = 0,
  \end{equation}
  and since \(\ker G \cong \im \partial_1\), we also have \(\im G^{\mathrm T} \cong \ker \partial_1^{\mathrm T}\), so there exists a \(w_2\) such that \(G^{\mathrm T}w_2 = f_0^{\mathrm T}w_1 + v_A\). Now we have:
  \begin{equation}
    \begin{aligned}
      \widetilde H_Z^{\mathrm T} \begin{pmatrix} w_1 \\ w_2 \end{pmatrix} &= \begin{pmatrix} H_Z^{\mathrm T}w_1 \\ f_0^{\mathrm T}w_1 + G^{\mathrm T}w_2 \end{pmatrix} \\ &= \begin{pmatrix} v_C \\ f_0^{\mathrm T}w_1 + f_0^{\mathrm T}w_1 + v_A \end{pmatrix} = \begin{pmatrix} v_C \\ v_A \end{pmatrix} = v, 
    \end{aligned}
  \end{equation}
  therefore in the case where \(v_C^{\mathrm T}\) is a \(Z\)-stabilizer of the original code, \(v^{\mathrm T}\) is a stabilizer of the code in \eqref{eq:conecodeZ}. Thus, the \(Z\)-distance for \eqref{eq:conecodeZ} is at least \(d_Z\), and therefore the dressed \(Z\)-distance of \(\widetilde{\mathcal C}\) is also at least \(d_Z\).
\end{proof}

This framework is sufficiently general to encompass many known schemes, including both lattice surgery \cite{Horsman_2012} and generalizations of lattice surgery \cite{Cohen2021,cross2024linearsize} as we show in the following examples.

\begin{example}[Lattice surgery of surface codes~\cite{Horsman_2012}]
    Consider two surface codes of equal $X$-distance $d$, $\mathcal C_1 = \mathrm{CSS}(H_{X,1}, H_{Z,1})$ and $\mathcal C_2 = \mathrm{CSS}(H_{X,2}, H_{Z,2})$. Let $\bar X_1$ be a weight $d$ logical operator of $\mathcal C_1$ and $\bar X_2$ be a weight $d$ logical operator of $\mathcal C_2$ such that $\mathcal Q_1 = \supp \bar X_1 = \{q_{1,1}, q_{1,2}, \dots, q_{1,d}\}$, and $\mathcal Q_2 = \supp \bar X_2= \{q_{2,1}, q_{2,2}, \dots, q_{2,d}\}$ and suppose we wish to measure $\bar X_1 \bar X_2$.
    
    Define the chain complex $\mathcal C$ in \cref{eqn:code-chain-complex} by $H_X = H_{X, 1} \oplus H_{X, 2}$ and $H_Z = H_{Z, 1} \oplus H_{Z, 2}$. Define $H_Z|_{\mathcal Q_1}^{\mathrm{nz}}$ to be $H_Z$ restricted to the columns given by $\mathcal Q_1$ with zero rows removed (``$\mathrm{nz}$'' refers to no zeroes), and similarly for $\mathcal Q_2$. Define the chain $\mathcal A$ by the maps
    \begin{equation}\label{eqn:lattice-surgery-anc-stabs}
        \partial_1 = \begin{pmatrix}
            1 & 1 & & & \\
            & 1 & 1 & & \\
            &   & \ddots & \ddots & \\
            & & & 1 & 1
        \end{pmatrix}_{(d - 1) \times d}, \qquad
        \partial_0 = 0
    \end{equation}
    where $\partial_1$ is equal to both $H_Z|_{\mathcal Q_1}^{\mathrm{nz}}$ and $H_Z|_{\mathcal Q_2}^{\mathrm{nz}}$ up to permutations, and define the chain map $f: \mathcal A \to \mathcal C$ by
    \begin{equation}
        \left( f_1 \right)_{i,j} = \delta_{i, q_{1,j}} + \delta_{i, q_{2,j}},
    \end{equation}
    and
    \begin{equation}
        \left( f_0 \right)_{i,j} = \begin{cases}
            1 & \text{if } \begin{cases}(H_Z)_{i,q_{1,j}} = (H_Z)_{i,q_{1,j+1}} = 1 \text{ or,} \\ (H_Z)_{i,q_{2,j}} = (H_Z)_{i,q_{2,j+1}} = 1 \end{cases} \\
            0 & \text{otherwise}
        \end{cases},
    \end{equation}
    where $f_1$ has dimensions $n \times d$, $f_0$ has dimensions $n_Z \times (d-1)$. The cone code of $f$ gives the merged system after performing lattice surgery.

    Note that $f_0$ embeds the rows of $\partial_1$ into the corresponding rows of $H_Z$, and $f_1^{\mathrm T}$ has $d$ rows of weight two which sum to $\bar X_1 \bar X_2$. From the structure of $f_1$ and the fact that $\partial_1$ has even weight rows, we see that $\bar X_1 \bar X_2$ is a stabilizer of the resulting code and can therefore be measured.
\end{example}

\begin{example}[Generalized lattice surgery~\cite{Cohen2021}]\label{ex:gen-lattice-surgery}
    Consider measuring an $X$-logical operator $\bar X$ of weight $w$ of the CSS code $\mathcal C = \mathrm{CSS}(H_X, H_Z)$. Let $\mathcal{Q} = \{q_1, \dots, q_w\} = \supp \bar X$, $H_Z|_{\mathcal{Q}}$ be the matrix $H_Z$ restricted to the qubits in $\mathcal{Q}$, and $H_Z|_{\mathcal Q}^{\mathrm{nz}}$ be the restriction to non-zero rows of $H_Z|_{\mathcal Q}$. Consider the following parity-check matrix of dimensions $(r-1)\times r$ for the $[r, 1, r]$ classical repetition code:
    \begin{equation}
        H_r = \begin{pmatrix}
            1 & 1 & & & \\
            & 1 & 1 & & \\
            &   & \ddots & \ddots & \\
            & & & 1 & 1
        \end{pmatrix},
    \end{equation}
    where the choice \(r = d\) with $d$ the distance of $\mathcal C$ is used in general to ensure that distance is maintained \cite{Cohen2021}. The ancillary chain $\mathcal A$ is given by the hypergraph product code~\cite{tillich2014quantumLDPC}
    \begin{equation}
      \begin{aligned}
        \partial_1^{\mathrm T} &= \begin{pmatrix}I_r \otimes \partial_1^{\mathrm T} & H_r^{\mathrm T} \otimes I_w\end{pmatrix}, \\
        \partial_0 &= \begin{pmatrix}H_r \otimes I_{n_Z'} & I_{r-1} \otimes \partial_1\end{pmatrix},
      \end{aligned}
    \end{equation}
    where \(n_Z'\) is the number of rows of $H_Z|_{\mathcal Q}^{\mathrm{nz}}$.

    Let $h_j$ be the index of the \(j\)-th non-zero row of $H_Z|_{\mathcal Q}$, and we define the chain map $f: \mathcal A \to \mathcal C$ by
    \begin{equation}
        \left( f_1 \right)_{i,\, j} = \begin{cases}
            \delta_{i,\, q_j} & j \leq w \\
            0 & \text{otherwise,}
        \end{cases}
    \end{equation}
    and
    \begin{equation}
        \left( f_0 \right)_{i,\, j} = \begin{cases}
            \delta_{i,\, h_j} & j \leq n_Z' \\
            0 & \text{otherwise,}
        \end{cases}
    \end{equation}
    where $f_1$ is an $n \times rw$ matrix and $f_0$ is an ${n_Z \times (rn_Z' + (r-1)w)}$ matrix. The cone code \eqref{conecode-mats} then gives the generalized lattice surgery construction for measurement of $\bar X$.
\end{example}

\begin{example}[Generalized lattice surgery with reduced \(r\)~\cite{cross2024linearsize}]\label{ex:IBM-lattice-surgery}
  The choice of \(r = d\) in the previous \cref{ex:gen-lattice-surgery} is not necessarily optimal, but due to the potential presence of gauges, it is difficult to find an analysis of when a reduced choice of \(r\) is admissable. To address this, \cite{cross2024linearsize} includes all \(Z\)-logical operators of the ancilla system from \cref{ex:gen-lattice-surgery} as stabilizers in \(\partial_0\) so that the resulting cone code has no gauges. It is then proven (following a similar result in weight reduction \cite[Lemma 8(8)]{hastings2023quantum}) that the edge expansion of the Tanner graph of \(H_Z\) restricted to the variable nodes \(Q\) can be used to find an upper bound on \(r\) that will maintain distance.

  The new rows of \(\partial_0\) in this scheme are not necessarily low weight, and indeed there does not always exist a choice that could be considered LDPC. For example, the measurement of a weight \(d\) logical operator of a distance \(d\) toric code would give a row in \(\partial_0\) of weight \(d\). This shortcoming is addressed in our scheme.

  Note that in \cite{cross2024linearsize}, a parameter \(L\) is used instead of \(r\) with the relationship \(2r - 1 = L\) for odd \(L\). Even values of \(L\) give a slightly different construction that, while still falling neatly within the cone code framework, does not provide a useful measurement technique and will be addressed in \cref{app:cylinder}.
\end{example}

\Cref{ex:gen-lattice-surgery,ex:IBM-lattice-surgery} (after \cite{Cohen2021,cross2024linearsize}) along with another recent work \cite{zhang2024timeefficient} all have similar schemes for measuring individual logicals, but note that joint measurements are conducted differently across these three works. These examples show that the homological measurement framework elegantly includes many past schemes. In the remainder of the section, we work within the cone code framework to develop a more optimized measurement scheme.

\subsection{Measuring an \texorpdfstring{\(X\)}{X}-logical operator}\label{sec:single-log}

Let \(\mathcal C = \mathrm{CSS}(H_X, H_Z)\) be an \([\![n, k, d = \min\{d_X, d_Z\}]\!]\) CSS code and \(w_X\) and \(q_X\) (\emph{resp.} \(w_Z\) and  \(q_Z\)) be the maximum row and column weights of \(H_X\) (\emph{resp.} \(H_Z\)). Given an \(X\)-logical operator \(\bar X\) to measure, define \({\mathcal Q = \supp \bar X = \{q_1,\dots,q_w\}}\), where \(w = |\mathcal Q|\) is the weight of \(\bar X\) and $q_i < q_{i + 1}$ for all $i$ so that the qubit labels are ordered. This works identically for \(Z\)-logical operators by swapping the roles of \(X\) and \(Z\). To perform a joint measurement across multiple codes $\mathcal C_i = \mathrm{CSS}(H_{X,i}, H_{Z,i})$, let $\mathcal C = \mathrm{CSS}(\bigoplus_{i} H_{X,i}, \bigoplus_{i} H_{Z,i})$ and proceed as usual.

We begin by choosing $f_1$ to be the $n \times w$ matrix
\begin{equation} \label{eqn:f1_def}
    (f_1)_{i, j} = \delta_{i, q_j}.
\end{equation}
If we choose a $\partial_1$ with even weight rows, then the sum of the new $X$-stabilizers given by the rows $\begin{pmatrix}f_1^{\mathrm T} & \partial_1^{\mathrm T}\end{pmatrix}$ from \cref{conecode-mats} is $\begin{pmatrix}\bar X & 0\end{pmatrix}$ as desired. This choice of $f_1$ requires $f_0$ and $\partial_1$ to be chosen such that
\begin{equation} \label{eq:f0-f1-relationship}
    f_0\partial_1 = H_Zf_1 = H_Z|_{\mathcal Q},
\end{equation}
so that $f$ is a valid chain map (equivalently, to enforce commutativity of the stabilizer matrices in \cref{conecode-mats}). The notation $H_Z|_{\mathcal Q}$ means $H_Z$ restricted to the columns given by the qubits $Q$. From \cref{eq:f0-f1-relationship}, we see that all rows of $H_Z|_{\mathcal Q}$ must be in the row space of $\partial_1$, and in order to have a sparse choice of $f_0$, it should require very few rows of $\partial_1$ to generate rows of $H_Z|_{\mathcal Q}$, and rows of $\partial_1$ should not be re-used many times to represent different rows of $H_Z|_{\mathcal Q}$. We will construct a stabilizer code which often requires a nonzero choice for $\partial_0$ to ensure that there are no gauges, so the sparsity of $\partial_0$ must also be considered. Given this setup, the following \cref{thm:x-dist} gives a sufficient condition to ensure that the $X$-distance is maintained in the cone code.

\begin{theorem}[$X$-distance]\label{thm:x-dist}
    Given a CSS code defined by ${\mathcal C:C_2 \xrightarrow{H_X^{\mathrm T}} C_1 \xrightarrow{H_Z} C_0}$, a logical operator $\bar X$ we wish to measure from $\mathcal C$, a chain complex ${\mathcal A:A_1 \xrightarrow{\partial_1} A_0 \xrightarrow{\partial_0} A_{-1}}$ such that ${\ker \partial_0 \cong \im \partial_1}$, and a chain map ${f: \mathcal A \to \mathcal C}$ with $f_1$ as in \cref{eqn:f1_def}, the cone code of $f$ maintains the $X$-distance of $\mathcal C$ when the Cheeger constant of the (hyper)graph defined by interpreting $\partial_1$ as an edge-vertex incidence matrix is at least 1.
\end{theorem}
\begin{proof}
    Since no gauges are created and the original $X$-logical operators are still $X$-logical operators, we must consider the weight of an unmeasured logical operator $\bar X_2$ from the original code, say of weight $w_2$, multiplied by some collection of the new stabilizer generators from the rows of $\begin{pmatrix} f_1^{\mathrm T} & \partial_1^{\mathrm T} \end{pmatrix}$. Label the corresponding $w$ stabilizers $S_{q_1},\dots,S_{q_w}$ where $S_{q_i}$ has support in the qubits of the original code on only $q_i$. To ensure the $X$-distance is maintained, it is sufficient to show that for any subset $R \subseteq \mathcal Q$,
    \begin{equation}
        \wt\left(\bar X_2\prod_{q \in R}S_q\right) \geq d.
    \end{equation}

    Fix a subset $R \subseteq \mathcal Q$ and define the logical representative $L = \bar X_2\prod_{q \in R}S_q$. Let $R_{\text{int}} = R \cap \supp \bar X_2$ and $R_{\text{int}}^c = R \setminus \supp \bar X_2$. 
    Notice that $R$ is the disjoint union of $R_{\text{int}}$ and $R_{\text{int}}^c$. Suppose the size of the overlap between $\bar X$ and $\bar X_2$ is \(i = \left| \supp \bar X \cap \supp \bar X_2 \right|\). Because the weight of \(\bar X \bar X_2\) must be at least \(d\), we have
    \begin{align}
        w + w_2 - 2i \geq d &\implies |R_{\text{int}}| \leq i \leq \frac{w + w_2 - d}{2} \\
        &\implies 2|R_{\text{int}}| \leq w + w_2 - d
    \end{align}
    Let $h$ be the Cheeger constant of the graph with $w$ vertices defined by the incidence matrix $\partial_1$, recalling that $h \geq 1$ by supposition. The weight of $L$ on the qubits of $\mathcal C$ is $w_2 + |R_{\text{int}}^c| - |R_{\text{int}}|$ by definition, and on the qubits of $\mathcal A$ it is lower bounded by $h \cdot \min\{|R|, w - |R|\}$ by the definition of the Cheeger constant. If $|R| \leq w/2$, then
    \begin{align}
        \wt(L) &\geq w_2 + |R_{\text{int}}^c| - |R_{\text{int}}| + |R|h \\
        &\geq w_2 + |R_{\text{int}}^c| - |R_{\text{int}}| + |R| \\
        &\geq w_2 \geq d,
    \end{align}
    and if $|R| > w/2$, then noting that $|R_{\text{int}}^c| - |R| = -|R_{\text{int}}|$, we have
    \begin{align}
        \wt(L) &\geq w_2 + |R_{\text{int}}^c| - |R_{\text{int}}| + (w - |R|)h \\
        &\geq w_2 + |R_{\text{int}}^c| - |R_{\text{int}}| + w - |R| \\
        &= w + w_2 - 2|R_{\text{int}}| \\
        &\geq w + w_2 - (w + w_2 - d) = d
    \end{align}
    And so, we find that $\wt(L)\geq d$, and thus the $X$-distance of the cone code of $f$ is at least $d$.
\end{proof}

Following choices made in \cite{hastings2023quantum,Cohen2021,cross2024linearsize,zhang2024timeefficient}, we start by defining the matrix \(\partial_1^*\) to be the \(n_Z' \times w\) matrix obtained by restricting \(H_Z\) to the columns \(\mathcal Q\) followed by removing zero rows where \(n_Z'\) is the number of non-zero rows:
\begin{equation} \label{eqn:p1*_def}
    \partial_1^* = H_Z|_{\mathcal Q}^{\mathrm{nz}}
\end{equation}
The superscript-$*$ is to emphasize that this will not be our final choice, but is instead a step towards that choice. Let \(h_j\) be the row index of the \(j\)-th non-zero row of \(H_Z|_{\mathcal Q}\) and define the \(n_Z \times n_Z'\) matrix \(f_0^*\) by
\begin{equation} \label{eqn:f0*_def}
    (f_0^*)_{i, j} = \delta_{i, h_j}.
\end{equation}

\begin{remark}
    If we were to use $f_0^*$ and $\partial_1^*$ as in \cref{eqn:p1*_def,eqn:f0*_def} and choose $\partial_0 = 0$, the cone code would give us the $r=1$ construction of \cite{Cohen2021}. See also \cref{ex:gen-lattice-surgery}.
\end{remark}

\begin{remark}
An alternative way to obtain the same maps is to choose $\mathcal Q$ as above, and consider the following operations:
\begin{equation}\label{derive-f-partial}
    \begin{tikzcd}
        {\begin{pmatrix}H_Z & I_{n_Z} \\ I_n & \end{pmatrix}} \\
        {\begin{pmatrix}H_Z|_{\mathcal Q} & I_{n_Z} \\ f_1 & \end{pmatrix}} \\
        {\begin{pmatrix} \partial_1^* & {f_0^*}^{\mathrm T} \\ f_1 & \end{pmatrix}}
        \arrow["{\text{delete columns } \{1, \dots, n\} \setminus \mathcal Q}", from=1-1, to=2-1]
        \arrow["{\text{delete zero rows of $H_Z|_{\mathcal Q}$}}", from=2-1, to=3-1]
    \end{tikzcd}
\end{equation}
where we delete the columns of \(\begin{pmatrix} H_Z \\ I_n \end{pmatrix}\) corresponding to qubits \(\{1,\dots,n\} \setminus \mathcal Q\) followed by deleting the rows of \(\begin{pmatrix} H_Z|_{\mathcal Q} & I_{n_Z} \end{pmatrix}\) whose corresponding rows of \(H_Z|_{\mathcal Q}\) are 0.
\end{remark}

Notice that since $\bar X$ must commute with $Z$-stabilizers, rows of $\partial_1^*$ must have even weight. All modifications we make to $\partial_1^*$ to obtain the final $\partial_1$ will retain this feature so that the sum of the rows $\begin{pmatrix} f_1^{\mathrm T} & \partial_1^{\mathrm T} \end{pmatrix}$ gives $\begin{pmatrix}\bar X & 0\end{pmatrix}$. Further, $f_0^*$ is sparse and the row and column weights of $\partial_1^*$ are bounded by $w_Z$ and $q_Z$. The remaining design problems are finding a sparse choice of $\partial_0$ and modifying $\partial_1^*$ so that the graph it defines has a Cheeger constant of 1 so that the $X$-distance is maintained. To solve the former problem, we will use a combination of replacing hyperedges with usual edges followed by \emph{cellulation}. The latter problem is solved via application of \cref{alg:increase-cheeger} where new edges of the graph correspond to new rows of $\partial_1^*$ along with corresponding new zero-columns in $f_0^*$.

\begin{remark}
    A positive Cheeger constant implies that there is only a single connected component, which necessarily gives $\dim \ker \partial_1 = 1$, so by \cref{lem:ker-partial1}, only the desired operator is measured. Therefore joint measurements do not require any extra work after applying \cref{alg:increase-cheeger}.
\end{remark}

\begin{algorithm}[ht]
    \DontPrintSemicolon
    \SetAlgoLined
    \SetKwComment{Comment}{// }{}
    \KwInput{Hypergraph \(A = \mathcal G(V, E)\) with Cheeger constant \(h(A) < 1\).}
    \KwOutput{A hypergraph \(B\) with Cheeger constant \(h(B) = 1\), the same vertices as $A$, and a superset of the edges of $A$.}
    \BlankLine
    $E^* \gets E$ and $B \gets \mathcal G(V, E^*)$ \;
    \BlankLine
    \While{\(h(B) < 1\)}{
        \BlankLine
        \Comment*[l]{Find the sparsest cut:}
        \(S \gets \displaystyle \argmin_{S \subset V,\ |S| \leq |V|/2}\left\{\frac{|\partial S|}{|S|}\right\}\) where $\partial S$ is calculated using the edges $E^*$.
        \BlankLine
        \Comment*[l]{Add an appropriate edge:}
        $h^* \gets -\infty$ and initialize an edge $e$ that will be overwritten. \;
        \For{$v_1$ a vertex of minimum degree in $S$}{
            \For{$v_2$ a vertex of minimum degree in $V \setminus S$}{
                \If{$h(\mathcal G(V, E^* \cup \{(v_1, v_2)\})) > h^*$}{
                    $h^* \gets h(\mathcal G(V, E^* \cup \{(v_1, v_2)\}))$ \;
                    $e \gets (v_1, v_2)$ \;
                }
            }
        }
        $E^* \gets E^* \cup \{e\}$ and $B \gets \mathcal G(V, E^*)$ \;
    }
    \BlankLine
    \KwRet{\(B\)}
    \caption{Greedy algorithm to add edges to a graph to obtain a Cheeger constant of one.}
    \label{alg:increase-cheeger}
\end{algorithm}

\begin{example}\label{ex:6vert-cheeger}
    Consider a graph given by vertices
    \begin{equation*}
        V = \{v_1,\dots,v_6\},
    \end{equation*}
    and edges
    \begin{equation*}
            E = \{(v_1,v_2), (v_2,v_3), (v_4,v_5), (v_5,v_6)\},
    \end{equation*}
    which can be seen in \cref{fig:example4-a}. This can be obtained from $\partial_1^*$ in \cref{eqn:p1*_def} when considering measurement of the product of the $X$-logical operators across two distance $3$ surface codes. Applying \cref{alg:increase-cheeger}, three edges are added. This can be seen in \cref{fig:example4-b,fig:example4-c,fig:example4-d}.
    
    For concreteness, we show the incidence matrix $\partial_1$ at each stage with a superscript representing which subfigure of \cref{fig:example4} is described by the matrix. The labeling of the vertices and edges is consistent with the order of rows and columns. For the original graph as shown in \cref{fig:example4-a}.
    \begin{equation*}\begin{gathered}
        \partial_1^{(a)} = \begin{pmatrix}
            1 & 1 & & & & \\
             & 1 & 1 & & & \\
             & & & 1 & 1 & \\
             & & & & 1 & 1
        \end{pmatrix}, \quad \partial_1^{(b)} = \begin{pmatrix}
            1 & 1 & & & & \\
             & 1 & 1 & & & \\
             & & & 1 & 1 & \\
             & & & & 1 & 1 \\
             & & 1 & 1 & &
        \end{pmatrix}, \\
        \partial_1^{(c)} = \begin{pmatrix}
            1 & 1 & & & & \\
             & 1 & 1 & & & \\
             & & & 1 & 1 & \\
             & & & & 1 & 1 \\
             & & 1 & 1 & & \\
            1 & & & & & 1
        \end{pmatrix}, \quad \partial_1^{(d)} = \begin{pmatrix}
            1 & 1 & & & & \\
             & 1 & 1 & & & \\
             & & & 1 & 1 & \\
             & & & & 1 & 1 \\
             & & 1 & 1 & & \\
            1 & & & & & 1 \\
            1 & & & 1 & &
        \end{pmatrix}.
    \end{gathered}\end{equation*}
    As a consequence of this example, we see that a joint measurement of the $X$-logical operators across two distance 3 surface codes requires 7 ancillas.
\end{example}

\begin{figure}[ht]
    \centering
    \subfloat[\label{fig:example4-a}]{\includegraphics{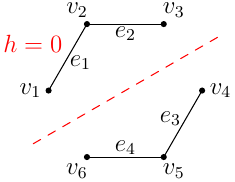}\ }
    \subfloat[\label{fig:example4-b}]{\ \includegraphics{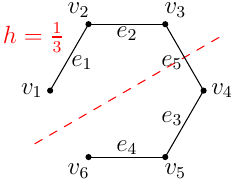}}\\
    \subfloat[\label{fig:example4-c}]{\includegraphics{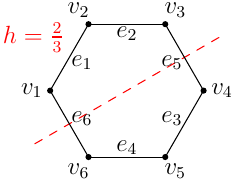}\ }
    \subfloat[\label{fig:example4-d}]{\ \includegraphics{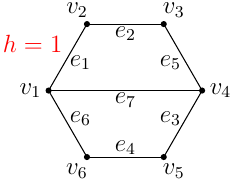}}
    \caption{Visualization of the application of \cref{alg:increase-cheeger} discussed in \cref{ex:6vert-cheeger}. The red dashed lines give a sparsest cut, while the value $h$ in red gives the Cheeger constant. Edges are added one at a time until a Cheeger constant of 1 is obtained leading from (a) to (d).}\label{fig:example4}
\end{figure}

After applying \cref{alg:increase-cheeger} to increase the Cheeger constant to 1, one could choose the most sparse option for $\partial_0$ among all matrices that satisfy $\ker \partial_0 \cong \im \partial_1$, in other words, $\partial_0$ is chosen to represent a cycle basis for the hypergraph as defined in \cref{def:cycle-basis}. This choice is analogous to the gauges that are promoted to stabilizers in \cite{cross2024linearsize}. However, when there are redundancies in the $Z$-stabilizer generators $H_Z$, this can often be improved. Consider the space of row vectors $V = \{v^T f_0 \ \vert \ v \in \ker H_Z^{\mathrm T}\}$ which are cycles that will be represented in $\widetilde H_Z$ regardless of the choice of $\partial_0$. With these in mind, choose $\partial_0$ such that the rows represent a basis for a subspace of all cycles $W$ such that $V \oplus W$ is the entire cycle space of the graph defined by $\partial_1$. A method for randomly searching such choices of $\partial_0$ for low weight options is given in \cref{alg:low-wt-p0}. If this is sparse enough, then the construction is complete.

\begin{algorithm}[ht]
    \DontPrintSemicolon
    \SetAlgoLined
    \SetKwComment{Comment}{// }{}
    \KwInput{$\partial_1$, $H_Z$, $f_0$, number of random samples $n$}
    \KwOutput{$\partial_0$}
    \BlankLine
    Define $V$ to be any matrix whose rows form a basis of $\left\{v^T f_0 \ \vert \ v \in \ker H_Z^{\mathrm T}\right\}$. \;
    Put $V$ in reduced row echelon form. \;
    Define $W$ to be any matrix such that $\ker W \cong \im \partial_1$. \;
    Add rows of $V$ to rows of $W$ to zero out the pivot columns of $V$ in $W$. \;
    Put $W$ in reduced row echelon form with zero-rows removed. \;
    Initialize $\partial_0 \gets W$
    \BlankLine
    \For{$i \in \{1,\dots,n\}$}{
        Let $A$ be a random, invertible matrix \cite{randall1993efficient}. \;
        Let $B$ be a random (not necessarily invertible) matrix. \;
        \If{the maximum row weight of $AW + BV$ is less than that of $\partial_0$}{
            $\partial_0 \gets AW + BV$. \;
        }
        \Comment*[l]{Frequently, not adding rows of $V$ gives lower weight. Check for this:}
        \If{the maximum row weight of $AW$ is less than that of $\partial_0$}{
            $\partial_0 \gets AW$. \;
        }
    }
    \BlankLine
    \KwRet{$\partial_0$}
    \caption{Random search for low weight $\partial_0$.}
    \label{alg:low-wt-p0}
\end{algorithm}

Unfortunately, there is no guarantee that $\partial_0$ will be sparse even for a qLDPC input code, in which case more work must be done. In order to improve the sparsity of $\partial_0$, take a step back and start with the incidence matrix $\partial_1^*$ as defined in \cref{eqn:p1*_def} before \cref{alg:increase-cheeger} has been applied, take all edges of the hypergraph that are defined on more than two vertices and replace them with any collection of (usual) edges defined only on two vertices each such that all vertices of the original edge are represented. For an edge defined on $v$ vertices, it suffices to choose $v/2$ edges for the replacement. The choice of these $v/2$ edges should be made such that the Cheeger constant is kept as high as possible. Representing all vertices from the original edges is necessary so that the row space of $\partial_1^*$ still includes $H_Z|_{\mathcal Q}$ as required by \cref{eq:f0-f1-relationship}, and $f_0$ must be modified to satisfy \cref{eq:f0-f1-relationship} by replacing each column corresponding to an expanded hyperedge on $v$ vertices with $v/2$ copies of that column. At this point, \cref{alg:increase-cheeger} should be applied to obtain a $\partial_1^*$ corresponding to a graph with only weight-two edges and a Cheeger constant of 1.

Next, apply \cref{alg:low-wt-p0} to obtain $\partial_0$. Now that we have a usual graph rather than a hypergraph, we can understand the cycles defined by the rows of $\partial_0$ more easily. If there is no cycle basis where all cycles are low weight, we can add edges across vertices inside of large cycles to create multiple smaller cycles, thus reducing the weight of $\partial_0$, a process called \emph{cellulation}. Here, the construction is complete and we obtain the final $f_1$, $f_0$, $\partial_1$, and $\partial_0$. The full construction is summarized in \cref{alg:main-construction}.

Parallel measurement refers to independent measurements of multiple logical operators, as opposed to a single joint measurement. Parallel measurements can be performed via iterative application of \cref{alg:main-construction}. When a measurement of an $X$-type operator is performed, the remaining $X$-logical operators are unchanged, so the next application of \cref{alg:main-construction} is straightforward. However, to measure an $X$-type and $Z$-type operator in parallel, after constructing the cone code for measurement of the $X$-type operator, the $Z$-type operator will have its support extended to some of the ancillas and this must be taken into account.

\begin{remark}[Bounds on connectivity in edge expanded homological measurement]
    The decongestion lemma~\cite{freedman2021building} can be used to build a cycle basis to define $\partial_0$ such that the connectivity of the cone code in edge expanded homological measurement is increased at most polylogarithmically in the weight of the measured logical operator as compared to the input code. We find that \Cref{alg:low-wt-p0} tends to result in less connectivity than the decongestion lemma, therefore we suggest the use of \cref{alg:low-wt-p0} for practical purposes.
\end{remark}

\begin{algorithm}[hbt!]
    \DontPrintSemicolon
    \SetAlgoLined
    \SetKwComment{Comment}{// }{}
    \KwInput{A code $\mathcal C = \mathrm{CSS}(H_X, H_Z)$ and $X$-logical operator $\bar X$}
    \KwOutput{A code $\widetilde{\mathcal C}$ with at least the distance of $\mathcal C$, with $\bar X$ in the stabilizer group, and with all other logical operators unharmed.}
    \BlankLine
    Define $f_1$, $\partial_1^*$, and $f_0^*$ as in \cref{eqn:f1_def,eqn:p1*_def,eqn:f0*_def}. \;
    Apply \cref{alg:increase-cheeger} to the incidence matrix $\partial_1^*$ to obtain a new incidence matrix $\partial_1$. \;
    Add zero columns to $f_0^*$ corresponding to the new edges from the previous step obtaining $f_0$ \;
    Apply \cref{alg:low-wt-p0} to $\partial_1$, $H_Z$, and $f_0$ to obtain $\partial_0$ \;
    \BlankLine
    \If{the sparsity of $\partial_0$ is deemed unacceptable}{
        $\partial_1 \gets \partial_1^*$ \;
        $f_1 \gets f_1^*$ \;
        \BlankLine
        \Comment*[l]{Expand hyperedges to weight-two edges}
        \For{each row $e$ of $\partial_1$ with $\wt e > 2$}{
            Replace $e$ with $(\wt e) / 2$ weight-two rows that sum to $e$ that keep the Cheeger constant as high as possible. \;
            Replace the column of $f_0$ that corresponds to $e$ with $(\wt e) / 2$ copies of that column.
        } 
        Apply \cref{alg:increase-cheeger} to the incidence matrix $\partial_1$ which adds new edges to $\partial_1$. \;
        Add zero columns to $f_0$ corresponding to the new edges from the previous step. \;
        \BlankLine
        \Comment*[l]{find a cycle basis}
        Apply \cref{alg:low-wt-p0} to $\partial_1$, $H_Z$, and $f_0$ to obtain $\partial_0$ \;
        \BlankLine
        \Comment*[l]{cellulate large cycles}
        \For{each row $c$ of $\partial_0$ with $\wt c$ higher than desired}{
            Add new edges (rows of $\partial_1$, along with corresponding zero-columns of $f_0$) within the cycle defined by $c$ to break it into smaller cycles. This results in replacing the high weight row $c$ with multiple lower weight rows corresponding to the new cycles. \;
        }
    }
    \BlankLine
    Define $\widetilde C$ to be the mapping cone of $f$ as in \cref{conecode}
    \BlankLine
    \KwRet{$\widetilde{\mathcal C}$}
    \caption{Main construction for the edge expanded homological measurement}
    \label{alg:main-construction}
\end{algorithm}

\begin{example}\label{ex:toric-d8}
    Cellulation can be useful for controlling the weights of $\partial_0$. Suppose we wish to measure a weight-8 $X$-logical operator, and the initial $\partial_1^*$ obtained from \cref{eqn:p1*_def} gives a cycle graph with 8 vertices seen in \cref{fig:example5-a} which has Cheeger constant $1/2$. The construction in \cite{cross2024linearsize} would give 24 ancilla qubits and a weight-8 $Z$-stabilizer corresponding to the weight of the single cycle in this graph. If the distance of the code is $d=8$, the construction in \cite{Cohen2021} does not create any new stabilizers above weight 4, but requires 120 ancilla qubits. Our edge expanded homological measurement construction results in 10 ancilla qubits without cellulation and 12 with cellulation with maximum stabilizer weights of 5, as seen below.
    
    In our construction, application of \cref{alg:increase-cheeger} adds two edges resulting in \cref{fig:example5-b} which increases the Cheeger constant to $1$. There is a cycle basis with 3 cycles of weight 5 which is easily seen due to the graph being planar. These three cycles result in 3 stabilizers of weight 5. For concreteness, given a particular ordering of qubits and stabilizers of the original code, this would result in the following $\partial_1$ and $\partial_0$:
    \begin{gather}
        \partial_1 = \begin{pmatrix}
            1 & 1 &   &   &   &   &   &   \\
              & 1 & 1 &   &   &   &   &   \\
              &   & 1 & 1 &   &   &   &   \\
              &   &   & 1 & 1 &   &   &   \\
              &   &   &   & 1 & 1 &   &   \\
              &   &   &   &   & 1 & 1 &   \\
              &   &   &   &   &   & 1 & 1 \\
            1 &   &   &   &   &   &   & 1 \\
            1 &   &   &   & 1 &   &   &   \\
              &   & 1 &   &   &   & 1 &
        \end{pmatrix}, \\ \partial_0 = \begin{pmatrix}
            1 & 1 & 1 & 1 &   &   &   &   & 1 &  \\
              &   & 1 & 1 & 1 & 1 &   &   &   & 1\\
              &   &   &   & 1 & 1 & 1 & 1 & 1 &  
        \end{pmatrix}.
    \end{gather}
    We can improve the weights of $\partial_0$ via cellulation. Since we are measuring an $X$ operator, each new $X$-stabilizer is given by a vertex of the graph, and the weight of the new $X$-stabilizer associated with a vertex is one more than the vertex degree. Therefore, when choosing a cellulation to try to reduce weight 5 operators, we avoid any choices that increase the vertex weights beyond 3. This tradeoff between the weights in $\partial_1$ and $\partial_0$ is a general feature of cellulation. Given this, the most reasonable cellulation choice is given in \cref{fig:example5-c} where we have reduced this to a single weight 5 cycle. For the same particular ordering of qubits and stabilizers in the above matrices, this gives us the following modified $\partial_1$ and $\partial_0$:
    \begin{gather}
        \partial_1 = \begin{pmatrix}
            1 & 1 &   &   &   &   &   &   \\
              & 1 & 1 &   &   &   &   &   \\
              &   & 1 & 1 &   &   &   &   \\
              &   &   & 1 & 1 &   &   &   \\
              &   &   &   & 1 & 1 &   &   \\
              &   &   &   &   & 1 & 1 &   \\
              &   &   &   &   &   & 1 & 1 \\
            1 &   &   &   &   &   &   & 1 \\
            1 &   &   &   & 1 &   &   &   \\
              &   & 1 &   &   &   & 1 &   \\
              & 1 &   & 1 &   &   &   &   \\
              &   &   &   &   & 1 &   & 1
        \end{pmatrix}, \\ \partial_0 = \begin{pmatrix}
            1 &   &   & 1 &   &   &   &   & 1 &   & 1 &   \\
              & 1 & 1 &   &   &   &   &   &   &   & 1 &   \\
              &   & 1 &   &   & 1 &   &   &   & 1 &   & 1 \\
              &   &   & 1 & 1 &   &   &   &   &   &   & 1 \\
              &   &   &   & 1 & 1 & 1 & 1 & 1 &   &   &
        \end{pmatrix}.
    \end{gather}
\end{example}

\begin{figure}
    \centering
    \subfloat[\label{fig:example5-a}]{\includegraphics{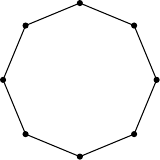}\qquad}
    \subfloat[\label{fig:example5-b}]{\qquad\includegraphics{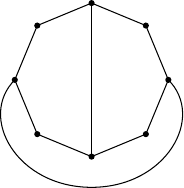}}\\
    \subfloat[\label{fig:example5-c}]{\includegraphics{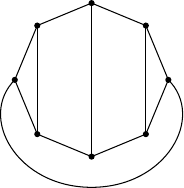}}
    \caption{The graph in (a) has a Cheeger constant of $1/2$. (b) is the result of applying \cref{alg:increase-cheeger} to (a) and has a Cheeger constant of 1. In (c), two of the three weight-5 cycles from (b) have been cellulated. See \cref{ex:toric-d8}.}
    \label{fig:example5}
\end{figure}

\begin{example}
    Consider the $[\![7,1,3]\!]$ Steane code given by the following stabilizer matrices:
    \begin{equation}\label{eq:steane-pcm}
        H_X = H_Z = \begin{pmatrix}
              &   &   & 1 & 1 & 1 & 1 \\
              & 1 & 1 &   &   & 1 & 1 \\
            1 &   & 1 &   & 1 &   & 1
        \end{pmatrix}.
    \end{equation}
    Measuring the logical operator $\bar X = \begin{pmatrix} 1&1&1&0&0&0&0 \end{pmatrix}$ gives the following stabilizer matrices:
    \begin{align}
        \widetilde H_X = \begin{pmatrix}
            H_X & \\ f_1^{\mathrm T} & \partial_1^{\mathrm T}
        \end{pmatrix} &= \begin{pNiceArray}{ccccccc|cc}
              &   &   & 1 & 1 & 1 & 1 &   &  \\
              & 1 & 1 &   &   & 1 & 1 &   &  \\
            1 &   & 1 &   & 1 &   & 1 &   &  \\\hline
            1 &   &   &   &   &   &   &   & 1 \\
              & 1 &   &   &   &   &   & 1 &   \\
              &   & 1 &   &   &   &   & 1 & 1
        \end{pNiceArray} \\
        \widetilde H_Z = \begin{pmatrix}
            H_Z & f_0
        \end{pmatrix} &= \begin{pNiceArray}{ccccccc|cc}
              &   &   & 1 & 1 & 1 & 1 &   & \\
              & 1 & 1 &   &   & 1 & 1 & 1 & \\
            1 &   & 1 &   & 1 &   & 1 &   & 1
        \end{pNiceArray}.
    \end{align}
    This is defines a $[\![ 9, 0 ]\!]$ code. Notice that the last three rows of $\widetilde H_X$ sum to $\begin{pmatrix} \bar X & 0 \end{pmatrix}$. Since $k$ drops to zero, there is no minimum distance to consider. See \cref{fig:f-partial-tanner} for a Tanner graph representation of $\partial_1$, $f_1$, and $f_0$ for this example. Since there are no cycles in the graph given by $\partial_1$, $\partial_0$ does not appear in this example.
\end{example}

\begin{figure}[t]
    \centering
    \subfloat[]{\includegraphics{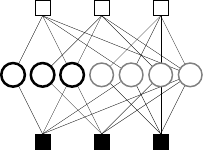}\quad}
    \subfloat[]{\quad\includegraphics{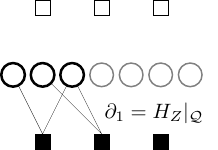}}\\
    \subfloat[]{\includegraphics{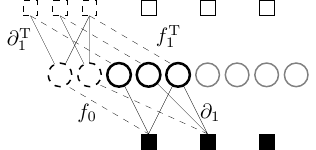}}
    \caption{A Tanner graph depicting measurement in the Steane code of the logical $X$ operator supported on the first three qubits shown with a thick border. The variable nodes with a dashed border represent ancilla. The original code is shown in (a), the restriction of $H_Z$ to the support of the logical operator to obtain $\partial_1$ is shown in (b), and $f_1$, $f_0$, and $\partial_1$ are shown in (c).}
    \label{fig:f-partial-tanner}
\end{figure}

\begin{example}
    To see when a low Cheeger constant can cause a decrease in distance, consider the $[\![ 15, 7, 3 ]\!]$ quantum Hamming code given by the stabilizer matrices
    \begin{equation}
        H_X = H_Z = \begin{pmatrix}
              &   &   &   &   &   &   & 1 & 1 & 1 & 1 & 1 & 1 & 1 & 1 \\
              &   &   & 1 & 1 & 1 & 1 &   &   &   &   & 1 & 1 & 1 & 1 \\
              & 1 & 1 &   &   & 1 & 1 &   &   & 1 & 1 &   &   & 1 & 1 \\
            1 &   & 1 &   & 1 &   & 1 &   & 1 &   & 1 &   & 1 &   & 1
        \end{pmatrix}
    \end{equation}
    and the logical-$X$ operator $\bar X = X_{3}X_{4}X_{5}X_{12}X_{14}$. Without applying \cref{alg:increase-cheeger} to improve the Cheeger constant of the graph given by $\partial_1$, we obtain the $[\![ 19, 6, 2 ]\!]$ code with stabilizer matrices:
    \begin{align}
        \widetilde H_X &= \begin{pmatrix}
              &   &   &   &   &   &   & 1 & 1 & 1 & 1 & 1 & 1 & 1 & 1 &   &   &   &   \\
              &   &   & 1 & 1 & 1 & 1 &   &   &   &   & 1 & 1 & 1 & 1 &   &   &   &   \\
              & 1 & 1 &   &   & 1 & 1 &   &   & 1 & 1 &   &   & 1 & 1 &   &   &   &   \\
            1 &   & 1 &   & 1 &   & 1 &   & 1 &   & 1 &   & 1 &   & 1 &   &   &   &   \\
              &   & 1 &   &   &   &   &   &   &   &   &   &   &   &   &   &   & 1 & 1 \\
              &   &   & 1 &   &   &   &   &   &   &   &   &   &   &   &   & 1 &   &   \\
              &   &   &   & 1 &   &   &   &   &   &   &   &   &   &   &   & 1 &   & 1 \\
              &   &   &   &   &   &   &   &   &   &   & 1 &   &   &   & 1 & 1 &   &   \\
              &   &   &   &   &   &   &   &   &   &   &   &   & 1 &   & 1 & 1 & 1 &  
        \end{pmatrix} \\
        \widetilde H_Z &= \begin{pmatrix}
              &   &   &   &   &   &   & 1 & 1 & 1 & 1 & 1 & 1 & 1 & 1 & 1 &   &   &   \\
              &   &   & 1 & 1 & 1 & 1 &   &   &   &   & 1 & 1 & 1 & 1 &   & 1 &   &   \\
              & 1 & 1 &   &   & 1 & 1 &   &   & 1 & 1 &   &   & 1 & 1 &   &   & 1 &   \\
            1 &   & 1 &   & 1 &   & 1 &   & 1 &   & 1 &   & 1 &   & 1 &   &   &   & 1
        \end{pmatrix}.
    \end{align}
    The weight-two logical operators of the new code are given by $X_{1}X_{19}$, $X_{2}X_{18}$, and $X_{8}X_{16}$. The Cheeger constant of the hypergraph given by $\partial_1$ is $0.5$. Only two additional edges (and therefore two additional ancilla), which can be found by applying \cref{alg:increase-cheeger}, are necessary to maintain the distance.
\end{example}

\subsection{Fault tolerant protocol for homological measurements}\label{sec:physical-considerations}

 The mathematical framework we have described is a good tool for describing the final state of the system, but does not prescribe a specific protocol. In this section, we describe the steps starting from the input code described by check matrices $H_X$ and $H_Z$ and leading to the final code described by $\widetilde{H}_X$ and $\widetilde{H}_Z$ given in~\cref{conecode}. It closely follows the steps used to perform lattice surgery introduced in \cite{Horsman_2012}. This protocol applies to any homological measurement, not just edge expanded homological measurement.

\begin{itemize}
    \item[1)] \textbf{Initialize ancilla:} Prepare the ancilla qubits $A_0$ in the $\lvert 0 \rangle^{\otimes (\dim A_0)}$ state for an $X$ measurement (or $\lvert + \rangle^{\otimes (\dim A_0)}$ for a $Z$ measurement). At this point, the physical system is stabilized by
    \begin{equation}
        H_X' = \begin{pmatrix}H_X & 0\end{pmatrix}, \quad
        H_Z' = \begin{pmatrix}H_Z & \\ & I\end{pmatrix}.
    \end{equation}

    \item[2)] \textbf{Measure the new $X$ stabilizers:} Measure each of the new $X$ stabilizers from the rows $\begin{pmatrix} f_1^{\mathrm T} & \partial_1^{\mathrm T} \end{pmatrix}$ of $\widetilde H_X$. The measurement results are uniformly random and must be recorded. Since
    \begin{equation}
        \mathrm{rowspace}\begin{pmatrix}H_Z & \\ & I\end{pmatrix} = \mathrm{rowspace}\begin{pmatrix}H_Z & f_0 \\ & I\end{pmatrix},
    \end{equation}
    and the rows $\begin{pmatrix}H_Z & f_0\end{pmatrix}$ commute with the measurements by construction, only the rows $\begin{pmatrix} 0 & I\end{pmatrix}$ of the latter matrix anticommute and are replaced with $\begin{pmatrix} 0 & G\end{pmatrix}$, where the rows of $G$ generate the subspace that commutes with $\partial_1^{\mathrm T}$, i.e., $G$ is a matrix such that $\im \partial_1 = \ker G \subseteq \ker \partial_0$. The result is a physical system that is stabilized by
    \begin{equation}
        H_X' = \begin{pmatrix}H_X & \\ f_1^{\mathrm T} & \partial_1^{\mathrm T}\end{pmatrix}, \quad
        H_Z' = \begin{pmatrix}H_Z & f_0 \\ & G\end{pmatrix},
    \end{equation}
    where the $Z$-stabilizers $\begin{pmatrix}0 & \partial_0\end{pmatrix}$ are included in the row space of $\begin{pmatrix}0 & G\end{pmatrix}$.

    \item[3)] \textbf{Measure the $Z$ stabilizers:} Measure the stabilizers $\begin{pmatrix}H_Z & f_0\end{pmatrix}$. If $\mathrm{rowspace}(\partial_0) \subsetneq \mathrm{rowspace}(G)$, treat the stabilizers in $\mathrm{rowspace}(G) \setminus \mathrm{rowspace}(\partial_0)$ as fixed gauges and do not measure them to avoid obtaining potentially high weight stabilizer generators.

    \item[4)] \textbf{Measure the stabilizers of $\widetilde C$ multiple times and perform error correction:} While measuring the new stabilizers in $\widetilde H_X$ gives random results since they anti commute with some existing $Z$-type stabilizers, their product gives the logical operator being measured. In the ideal case, that would be sufficient to infer the results of the logical measurement with certainty, but in the presence of errors the measurement of the stabilizers can be repeated over time a sufficient amount of time, typically a number of times equal to the distance of the initial code, and the results fed into a decoder in order to infer with high confidence the outcome of the logical measurement. Since in general there are other logical qubits remaining in the code, this is also required to maintain their protection.
    
    \item[5)] \textbf{Measure out ancillas to return to original code:} Measure the ancillas in the $Z$ basis. Use the resulting values to determine the resulting code space. Note that this must be done in a fault-tolerant way, since the measurements can be prone to errors.
\end{itemize}

\subsection{Jointly measuring a mixture of \texorpdfstring{\(X\)}{X}-, \texorpdfstring{\(Y\)}{Y}-, and \texorpdfstring{\(Z\)}{Z}-logical operators}\label{sec:non-css-meas}

To measure a logical operator $L$ that is a mixture of $X$-, $Y$-, and $Z$-logical operators, write the operator as a product of $X$ and $Z$-logical operators (including phases if required). Note that one needs to keep track of an overall sign in order to interpret the logical measurement correctly. First, apply \cref{alg:increase-cheeger} separately for $\bar X$ and $\bar Z$ obtaining $f_1^{(X)}$, $f_0^{(X)}$, and $\partial_1^{(X)}$ for $\bar X$ and $f_1^{(Z)}$, $f_0^{(Z)}$, and $\partial_1^{(Z)}$ for $\bar Z$. Note that the roles of $H_X$ and $H_Z$ are swapped when measuring $\bar Z$. We do not find a cycle basis yet at this point. Consider the following matrices which are \emph{not} the final product:
\begin{align}
    \widetilde H_X^{(\text{first step})} &= \begin{pmatrix}
        H_X & & f_0^{(Z)} \\
        {f_1^{(X)}}^{\mathrm T} & {\partial_1^{(X)}}^{\mathrm T} &
    \end{pmatrix} \label{eq:Y-firststep-HX} \\
    \widetilde H_Z^{(\text{first step})} &= \begin{pmatrix}
        H_Z & f_0^{(X)} & \\
        {f_1^{(Z)}}^{\mathrm T} & & {\partial_1^{(Z)}}^{\mathrm T}
    \end{pmatrix} \label{eq:Y-firststep-HZ}.
\end{align}

There are two possibilities after this first step. Either the two matrices do not commute, or if they do, the code they implement would measure $\bar X$ and $\bar Z$ separately. To fix both problems, we will merge some pairs of stabilizers, i.e., replace stabilizer generators with their product. If the matrices anticommute, then this will be due to stabilizer generators with a single overlapping qubit in the blocks ${f_1^{(X)}}^{\mathrm T}$ and ${f_1^{(Z)}}^{\mathrm T}$, corresponding to the physical Pauli-$Y$ operators that are being measured. Merge each anticommuting pair of stabilizer generators. If the matrices commute, merge any $X$ stabilizer from the rows starting with ${f_1^{(X)}}^{\mathrm T}$ with any $Z$ stabilizer from the rows starting with ${f_1^{(Z)}}^{\mathrm T}$. Any choice will have the desired outcome, so choose the lowest weight stabilizers. This results in a valid non-CSS stabilizer code given by the following stabilizer matrix in symplectic form where $f_1^{(X \setminus Z)}$, $f_1^{(Z \setminus X)}$, $\partial_1^{(X \setminus Z)}$, and $\partial_1^{(Z \setminus X)}$ are what remains after merging, and $f_1^{(Y|_X)}$, $f_1^{(Y|_Z)}$, $\partial_1^{(Y|_X)}$, and $\partial_1^{(Y|_Z)}$ are created according to the merging:
\begin{equation}
    \begin{pNiceArray}{ccc|ccc}
        H_X & & f_0^{(Z)} & & & \\
        {f_1^{(X \setminus Z)}}^{\mathrm T} & {\partial_1^{(X \setminus Z)}}^{\mathrm T} & & & & \\
        {f_1^{(Y|_X)}}^{\mathrm T} & {\partial_1^{(Y|_X)}}^{\mathrm T} & & {f_1^{(Y|_Z)}}^{\mathrm T} & & {\partial_1^{(Y|_Z)}}^{\mathrm T} \\
        & & & H_Z & f_0^{(X)} & \\
        & & & {f_1^{(Z \setminus X)}}^{\mathrm T} & & {\partial_1^{(Z \setminus X)}}^{\mathrm T}
    \end{pNiceArray}
\end{equation}

This code does include $L$ in its stabilizer group, but it also contains new logical operators that may be low weight, which is the same issue that arises in \cref{sec:single-log} requiring $\partial_0$. To fix this, we return to the topic of finding a cyclic basis as in \cref{sec:single-log}. The new undesired logical operators are given by the cycles of a particular graph. This graph is given by the union of the graphs defined by $\partial_1^{(X)}$ and $\partial_1^{(Z)}$ with vertices between the two identified if their respective stabilizers were merged, i.e., the incidence matrix of the graph is
\begin{equation}\label{eq:merged-incidence-mat}
    \partial_1 = \begin{pmatrix}
        \partial_1^{(X \setminus Z)} & \partial_1^{(Y|_X)} & \\
        & \partial_1^{(Y|_Z)} & \partial_1^{(Z \setminus X)}
    \end{pmatrix}.
\end{equation}
Find a cycle basis for the graph using \cref{alg:low-wt-p0} which we notate
\begin{equation}
    \partial_0 = \begin{pmatrix}
        \partial_0^{(X)} & \partial_0^{(Z)}
    \end{pmatrix},
\end{equation}
where edges from the top (\emph{resp.}, bottom) block row in \cref{eq:merged-incidence-mat} are represented in $\partial_0^{(X)}$ (\emph{resp.}, $\partial_0^{(Z)}$). If the row weights of $\partial_0$ are higher than desired, cellulation can be applied where now one may choose for each added edge whether to included it in the graph for $X$ or the graph for $Z$. The resulting code is given by the following stabilizer matrix in symplectic form:
\begin{equation}\label{eq:joint_final}
    \begin{pNiceArray}{ccc|ccc}
        H_X & & f_0^{(Z)} & & & \\
        {f_1^{(X \setminus Z)}}^{\mathrm T} & {\partial_1^{(X \setminus Z)}}^{\mathrm T} & & & & \\
        {f_1^{(Y|_X)}}^{\mathrm T} & {\partial_1^{(Y|_X)}}^{\mathrm T} & & {f_1^{(Y|_Z)}}^{\mathrm T} & & {\partial_1^{(Y|_Z)}}^{\mathrm T} \\
        & & & H_Z & f_0^{(X)} & \\
        & & & {f_1^{(Z \setminus X)}}^{\mathrm T} & & {\partial_1^{(Z \setminus X)}}^{\mathrm T} \\
        & & \partial_0^{(Z)} & & \partial_0^{(X)} & 
    \end{pNiceArray}
\end{equation}

The dimension of the resulting code is one less than that of the input code by construction. To see that the resulting distance is always preserved, consider an error $\vec e$ written in symplectic form and decompose it into $\vec e = \vec e_X + \vec e_Z$ for the $X$ and $Z$ parts of the error. Note that $\wt{\vec e} \geq \min\{\wt{\vec e_X}, \wt{\vec e_Z}\}$ and apply the arguments in the proof of \cref{thm:z-dist} to $\vec e_X$ and $\vec e_Z$ independently. The distance will often increase because the $X$-type measurement frequently increases the $Z$-distance, and vice versa.

\section{Examples and Numerical simulations} \label{sec:numerics}

In this Section, we first provide the parameters of the codes obtained after the logical measurements have been performed on a few examples. We next benchmark our fault-tolerant logical operator measurement protocol for a measurement-based quantum computer that uses GKP qubits concatenated in various quantum LDPC code.

\subsection{Examples} \label{subsec:examples}

We consider two classes of quantum LDPC codes: (a) lifted product (LP) codes \cite{Panteleev2022, Raveendran2022finiterateqldpcgkp} and (b) hypergraph product (HGP) codes \cite{tillich2014quantumLDPC, Roffe_QLPDC_decoding}. We focus on these particular code families as their distance scale with $n$, have constant encoding rate \cite{Panteleev2022, tillich2014quantumLDPC, Bravyi2024_IBM_GeneralizedBicycle}, and perform well with the photonic architecture we consider here, though we emphasize that the procedure presented in~\cref{sec:hom-meas} can be used with any code. To demonstrate homological measurement through our examples, for each code, we consider an $X$-logical operator for which the weight is equal to the distance of its code, but they have been otherwise chosen randomly.

\vspace*{0.1in}
\noindent\underline{Lifted product codes}
\vspace*{0.05in}

The first class of codes that we consider for our example are the lifted product codes which were first proposed by Panteleev and Kalachev \cite{Panteleev2022}. The lifted product code construction is the lifted version of the hypergraph product code construction.

Given a choice of \emph{lift size} $\ell$, consider the polynomial quotient ring $R_\ell = \mathbb{F}_2[x]/(x^\ell-1)$. For a polynomial $g(x) = g_0 + g_1x + \cdots + g_{\ell-1}x^{\ell-1} \in R_\ell$, define $g^{\mathrm T}(x) = g_0 + g_{\ell-1}x + \cdots + g_{1}x^{\ell-1} \in R_\ell$. The lift of $g(x)$ is an $\ell\times\ell$ circulant matrix $\mathbb{B}(g(x))$ where the first column is given by the coefficients of $g(x)$ and each subsequent column is obtained by a cyclic shifting down one index, i.e., the $i$-th column is given by the coefficients of $x^{i - 1}g(x)$. Similarly, the lift $\mathbb{B}(A)$ of a matrix $A$ over $R_\ell$ is obtained by replacing each element of $A$ by its lift.

 Let $A_1$ and $A_2$ be two matrices of size $m_1\times n_1$ and $m_2\times n_2$ respectively with entries in $R_\ell$. The quasi-cyclic lifted product code $\mathrm{LP}(A_1, A_2)$ is the CSS code of length $\ell(n_1m_2+n_2m_1)$ with stabilizer matrices
 \begin{align}
     H_X &= \mathbb{B}\left(\begin{bmatrix}A_1 \otimes I & I \otimes A_2\end{bmatrix}\right),\\
     H_Z &= \mathbb{B}\left(\begin{bmatrix}I \otimes A_2^T & A_1^T \otimes I\end{bmatrix}\right). 
 \end{align}

In these simulations, we choose lifted product codes given by matrices $A_1 = A_2^T = A$ whose entries are monomials in $R_\ell$. The matrices are explicitly shown in~\cref{tab:LP_code_matrices}.

\begin{table}[hbt!]
    \centering
    \begin{tabular}{|c|c|c|}
    \hline
         LP Code & $\ell$ & Matrix over $R_\ell$ \\ \hline\hline
         $\mathrm{LP}_1$ & 7 & $\begin{pmatrix}
             1 & 1 & 1 & 1\\1 & x & x^2 & x^5\\1 & x^6 & x^3 & x
         \end{pmatrix}$ \\
         $\mathrm{LP}_2$ & 9 & $\begin{pmatrix}
             1 & 1 & 1 & 1\\1 & x & x^6 & x^7\\1 & x^4 & x^5 & x^2
         \end{pmatrix}$ \\
         \hline
    \end{tabular}
    \caption{Matrices defining the lifted product codes used for simulations.}
    \label{tab:LP_code_matrices}
\end{table}

\vspace*{0.1in}
\noindent\underline{Hypergraph product codes}
\vspace*{0.05in}

The second class of codes that we consider for our simulations are the hypergraph product codes, first proposed by 
Tillich and Zémor \cite{tillich2014quantumLDPC}. The Tanner graph of a hypergraph product code is constructed by taking the graph product of the Tanner graphs of two classical codes.

Consider classical codes $C_1$ and $C_2$, and their respective parity-check matrices $H_1$ and $H_2$. Let $C_i$ have parameters $[n_i, k_i, d_i]$. Their transposed codes of $C_1$ is denoted by $C_i^T$ and has the parity-check matrix $H_i^T$ and parameters $[m_i, k_i^T, d_i^T]$. When $H_i^T$ is a full rank matrix, $C_i^T$ encodes no logical information and has infinite distance, i.e., $k_i^T=0$ and $d_i^T = \infty$. The hypergraph product code of $C_1$ and $C_2$ is denoted $\mathrm{HGP}(C_1, C_2)$ and has stabilizer matrices
\begin{align}
    H_X &= \begin{pmatrix} H_1 \otimes I & I \otimes H_2^{\mathrm T} \end{pmatrix}, \\
    H_Z &= \begin{pmatrix} I \otimes H_2 & H_1^{\mathrm T} \otimes I \end{pmatrix},
\end{align}
and parameters
\begin{align}
[\![n_1n_2+m_1m_2, k_1k_2 + k_1^Tk_2^T, \mathrm{min}(d_1,d_2,d_1^T,d_2^T)]\!].
\end{align}

The specific HGP codes we consider are built from $H = H_1 = H_2$, where $H$ is a parity-check matrix of a classical LDPC code with column weights of four and row weights of three. These classical LDPC codes obtained using the MacKay-Neal method \cite{MacKay_Neal1996,Roffe_QLPDC_decoding}. We consider two hypergraph product codes, denoted $\mathrm{HGP}_1$ and $\mathrm{HGP}_2$, given by the following parity-check matrices:
{\renewcommand{\arraystretch}{0.6}
\begin{align}
    H_{\mathrm{HGP}_1} = \left(\begin{smallmatrix}
0 & 0 & 0 & 1 & 1 & 0 & 0 & 0 & 0 & 0 & 0 & 0 & 0 & 0 & 0 & 0 & 0 & 1 & 0 & 1 \\
0 & 1 & 0 & 0 & 0 & 0 & 1 & 0 & 0 & 0 & 0 & 0 & 0 & 0 & 1 & 0 & 0 & 0 & 0 & 1 \\
0 & 0 & 0 & 0 & 0 & 1 & 0 & 0 & 0 & 1 & 0 & 0 & 0 & 0 & 0 & 1 & 1 & 0 & 0 & 0 \\
0 & 0 & 1 & 0 & 0 & 0 & 1 & 0 & 0 & 0 & 0 & 0 & 0 & 1 & 0 & 0 & 1 & 0 & 0 & 0 \\
0 & 0 & 0 & 0 & 0 & 0 & 0 & 0 & 1 & 1 & 0 & 1 & 0 & 0 & 0 & 0 & 0 & 0 & 1 & 0 \\
0 & 0 & 0 & 0 & 1 & 0 & 0 & 0 & 0 & 0 & 1 & 0 & 0 & 0 & 1 & 1 & 0 & 0 & 0 & 0 \\
0 & 0 & 0 & 0 & 0 & 0 & 0 & 1 & 0 & 1 & 0 & 0 & 1 & 0 & 0 & 0 & 0 & 1 & 0 & 0 \\
0 & 1 & 1 & 1 & 0 & 0 & 0 & 0 & 0 & 0 & 0 & 0 & 1 & 0 & 0 & 0 & 0 & 0 & 0 & 0 \\
0 & 0 & 0 & 0 & 0 & 1 & 0 & 0 & 1 & 0 & 0 & 0 & 1 & 0 & 1 & 0 & 0 & 0 & 0 & 0 \\
0 & 0 & 0 & 0 & 0 & 0 & 0 & 1 & 0 & 0 & 1 & 1 & 0 & 0 & 0 & 0 & 1 & 0 & 0 & 0 \\
0 & 0 & 0 & 1 & 0 & 1 & 0 & 1 & 0 & 0 & 0 & 0 & 0 & 0 & 0 & 0 & 0 & 0 & 1 & 0 \\
1 & 0 & 0 & 0 & 0 & 0 & 0 & 0 & 0 & 0 & 0 & 1 & 0 & 1 & 0 & 0 & 0 & 0 & 0 & 1 \\
1 & 0 & 0 & 0 & 1 & 0 & 1 & 0 & 0 & 0 & 0 & 0 & 0 & 0 & 0 & 0 & 0 & 0 & 1 & 0 \\
1 & 0 & 1 & 0 & 0 & 0 & 0 & 0 & 1 & 0 & 0 & 0 & 0 & 0 & 0 & 1 & 0 & 0 & 0 & 0 \\
0 & 1 & 0 & 0 & 0 & 0 & 0 & 0 & 0 & 0 & 1 & 0 & 0 & 1 & 0 & 0 & 0 & 1 & 0 & 0 \\
\end{smallmatrix}\right),
\end{align}
and
\begin{align}
    H_{\mathrm{HGP}_2} = \left(\begin{smallmatrix}
0 & 1 & 0 & 0 & 0 & 0 & 0 & 0 & 0 & 0 & 0 & 0 & 0 & 0 & 1 & 1 & 0 & 0 & 1 & 0 & 0 & 0 & 0 & 0 \\
0 & 0 & 0 & 0 & 0 & 0 & 1 & 1 & 0 & 0 & 0 & 0 & 0 & 0 & 0 & 1 & 0 & 0 & 0 & 0 & 0 & 0 & 1 & 0 \\
0 & 0 & 0 & 1 & 0 & 0 & 0 & 0 & 1 & 0 & 0 & 0 & 0 & 0 & 0 & 0 & 0 & 0 & 0 & 1 & 0 & 0 & 0 & 1 \\
0 & 0 & 1 & 0 & 0 & 0 & 0 & 0 & 0 & 0 & 0 & 1 & 0 & 1 & 0 & 0 & 0 & 0 & 0 & 0 & 0 & 0 & 0 & 1 \\
0 & 0 & 0 & 0 & 0 & 0 & 0 & 0 & 0 & 1 & 0 & 1 & 0 & 0 & 0 & 0 & 1 & 0 & 0 & 1 & 0 & 0 & 0 & 0 \\
0 & 0 & 0 & 0 & 0 & 1 & 0 & 0 & 1 & 0 & 0 & 0 & 0 & 0 & 0 & 0 & 0 & 1 & 0 & 0 & 0 & 0 & 1 & 0 \\
1 & 1 & 0 & 0 & 0 & 0 & 0 & 1 & 0 & 0 & 0 & 0 & 0 & 0 & 0 & 0 & 1 & 0 & 0 & 0 & 0 & 0 & 0 & 0 \\
0 & 0 & 0 & 0 & 0 & 1 & 0 & 0 & 0 & 0 & 0 & 0 & 1 & 1 & 1 & 0 & 0 & 0 & 0 & 0 & 0 & 0 & 0 & 0 \\
1 & 0 & 0 & 0 & 0 & 0 & 0 & 0 & 0 & 0 & 1 & 0 & 0 & 0 & 0 & 0 & 0 & 0 & 1 & 0 & 0 & 0 & 0 & 1 \\
0 & 0 & 0 & 0 & 1 & 0 & 0 & 0 & 0 & 0 & 0 & 0 & 0 & 0 & 0 & 0 & 1 & 0 & 0 & 0 & 0 & 1 & 1 & 0 \\
0 & 1 & 0 & 0 & 1 & 0 & 1 & 0 & 0 & 0 & 0 & 1 & 0 & 0 & 0 & 0 & 0 & 0 & 0 & 0 & 0 & 0 & 0 & 0 \\
0 & 0 & 0 & 1 & 0 & 0 & 0 & 1 & 0 & 0 & 0 & 0 & 0 & 0 & 1 & 0 & 0 & 1 & 0 & 0 & 0 & 0 & 0 & 0 \\
0 & 0 & 0 & 0 & 0 & 0 & 0 & 0 & 1 & 0 & 1 & 0 & 0 & 0 & 0 & 0 & 0 & 0 & 0 & 0 & 1 & 1 & 0 & 0 \\
1 & 0 & 0 & 0 & 1 & 0 & 0 & 0 & 0 & 0 & 0 & 0 & 0 & 0 & 0 & 0 & 0 & 0 & 0 & 1 & 1 & 0 & 0 & 0 \\
0 & 0 & 1 & 0 & 0 & 1 & 1 & 0 & 0 & 0 & 1 & 0 & 0 & 0 & 0 & 0 & 0 & 0 & 0 & 0 & 0 & 0 & 0 & 0 \\
0 & 0 & 0 & 0 & 0 & 0 & 0 & 0 & 0 & 1 & 0 & 0 & 0 & 1 & 0 & 0 & 0 & 0 & 1 & 0 & 0 & 1 & 0 & 0 \\
0 & 0 & 1 & 0 & 0 & 0 & 0 & 0 & 0 & 0 & 0 & 0 & 1 & 0 & 0 & 1 & 0 & 1 & 0 & 0 & 0 & 0 & 0 & 0 \\
0 & 0 & 0 & 1 & 0 & 0 & 0 & 0 & 0 & 1 & 0 & 0 & 1 & 0 & 0 & 0 & 0 & 0 & 0 & 0 & 1 & 0 & 0 & 0 \\
\end{smallmatrix} \right).
\end{align}}

The logical operators measured using our new technique are shown in Table \ref{tab:logical_operators}. The code parameters and row and column weights of the original codes are presented in~\cref{tab:original_code_parameters}, and the code parameters and weights obtained post-measurement are presented in~\cref{tab:new_code_parameters}.

\begin{table*}
    \centering     
    \begin{tabular}{|c|c|}
        \hline\hline
        Code & Logical operator\\ \hline 
        $\mathrm{LP}_1$ & $X_{30}X_{37}X_{38}X_{39}X_{44}X_{45}X_{47}X_{51}X_{53}X_{54}$ \\
        $\mathrm{LP}_2$ &  $X_{41}X_{42}X_{53}X_{55}X_{56}X_{60}X_{62}X_{63}X_{64}X_{65}X_{68}X_{72}$ \\
        $\mathrm{HGP}_1$ &  $X_{282}X_{286}X_{288}X_{290}X_{292}X_{293}X_{298}X_{300}$ \\
        $\mathrm{HGP}_2$ & $X_{217}X_{219}X_{220}X_{222}X_{224}X_{225}X_{230}X_{232}X_{235}X_{237}$ \\\hline
\hline
    \end{tabular}
     \caption{Logical operators used for fault-tolerant logical measurement simulations.}
    \label{tab:logical_operators}
\end{table*}

\renewcommand{\tabcolsep}{2pt}
\begin{table*}
\centering
\begin{tabular}{|c|c|c|c|c|}
\hline
Code & ($q_X$,$w_X$,$q_Z$,$w_Z$) & ($q$, $w$) & $(q_{X_{avg}}, w_{X_{avg}},q_{Z_{avg}}, w_{Z_{avg}})$ &$(q_{avg}, w_{avg})$\\
\hline\hline
$\mathrm{LP}_1$, $[\![175, 19, 10]\!]$ & (4, 7, 4, 7) & (8, 7) & (3.36, 7, 3.36, 7) & (6.72, 7)\\
$\mathrm{LP}_2$, $[\![225, 21, 12]\!]$ & (4, 7, 4, 7) & (8, 7) & (3.36, 7, 3.36, 7) & (6.72, 7)\\
$\mathrm{HGP}_1$, $[\![625,  25,  8]\!]$ & (4, 7, 4, 7) & (8,  7) & (3.36, 7, 3.36, 7) & (6.72,  7)\\
$\mathrm{HGP}_2$, $[\![900,  36,  10]\!]$ & (4, 7, 4, 7) & (8,  7) & (3.36, 7, 3.36, 7) & (6.72,  7)\\
\hline
\end{tabular}
\caption{Properties of the simulated codes}
\label{tab:original_code_parameters}
\end{table*}

\begin{table*}
\centering
\begin{tabular}{ |c|c|c|c|c|c|}
\hline
Original code &
     Measured code 
 & ($\tilde{q}_X$,$\tilde{w}_X$,$\tilde{q}_Z$,$\tilde{w}_Z$) & ($q$, $w$) & $(\tilde{q}_{X_{avg}}, \tilde{w}_{X_{avg}},\tilde{q}_{Z_{avg}}, \tilde{w}_{Z_{avg}})$ & $(\tilde{q}_{avg}, \tilde{w}_{avg})$\\
\hline\hline
$\mathrm{LP}_1$, $[\![175, 19, 10]\!]$ &  $[\![191, 18, 10]\!]$ & (4, 7, 6, 9) &  (10, 9) &  (3.31, 6.72, 3.39, 7.19) & (6.70, 6.95) \\
$\mathrm{LP}_2$, $[\![225, 21, 12]\!]$ & $[\![245, 20, 12]\!]$ & (4, 7, 7, 12) & (9, 12) & (3.30, 6.73, 3.47, 7.34) & (6.77, 7.03) \\
$\mathrm{HGP}_1$, $[\![625, 25, 8]\!]$ & $[\![638, 24, 8]\!]$ &  (4, 7, 5, 8) & (9, 8) & (3.35, 6.94, 3.37, 7.03) & (6.72, 6.98) \\
$\mathrm{HGP}_2$, $[\![900, 36, 10]\!]$ & $[\![917, 35, 10]\!]$ & (4, 7, 7, 11) & (9, 11) & (3.35, 6.94, 3.39, 7.06) & (6.73, 7.00) \\
\hline
\end{tabular}
\caption{Properties of the codes obtained post fault-tolerant logical measurements.}
\label{tab:new_code_parameters}
\end{table*}

\subsection{Numerical simulations} \label{subsec:numerics}

We benchmark our homological measurement protocol for the various codes presented in~\cref{subsec:examples} via simulation of a photonic architecture based on optical GKP states. The quantum code considered is first foliated (using a number of layers equal to the code distance) to give a cluster state as described in \cite{Bolt2016}. The cluster is generated by sending optical GKP qubits through a passive linear optics circuit comprising of beamsplitters and phase shifters, and homodyne measurements performed on a subset of the modes. The detailed description of the procedure used can be found in \cite{tzitrin2021staticlinear,walshe2024}. We provide a brief discussion of the photonic quantum computing architecture used for the simulations below:
\begin{itemize}
    \item[1)] \textbf{Preparing GKP Bell pair:} A GKP Bell pair is constructed using two GKP sensor states, a 50:50 beamsplitter, and a $\pi/2$ phase shifter. We model the generation of the GKP Bell pair by considering two modes in the ideal GKP sensor state, applying a Gaussian blurring channel of variance $\sigma^2$ that replaces each delta function of the GKP state by a Gaussian distribution of variance $\sigma^2$, and applying the optical circuit comprising of a 50:50 beamsplitter and a $\pi/2$ phase shifter. Throughout this paper, we express the variance of the Gaussian in terms of decibels, $\frac{\sigma^2}{\sigma_{\mathrm{vac}}^2}[\mathrm{dB}]$ $=$ $-10\mathrm{log}_{10}(\sigma^2/\sigma_{\mathrm{vac}}^2)$, where $\sigma_{\mathrm{vac}}^2$ is the variance of the vacuum, which is $1/2$ in units where $\hbar = 1$. This model is equivalent to uniform photon loss experienced throughout the cluster state generation and measurement process.
    \item[2)] \textbf{Preparing the optical resource state:} The optical resource state is obtained by considering the cluster state based on the foliated quantum code, replacing each edge of the cluster state by a GKP Bell pair, and applying a beamsplitter circuit on all the modes that correspond to the same node in the cluster state to obtain the resource state. The set of modes in various GKP Bell pairs corresponding to the same node in the cluster state are called a macronode. The beamsplitter circuit corresponds to performing a continuous-variable GHZ measurement.
    \item[3)] \textbf{Performing measurement-based quantum computation:} The GKP cluster state is obtained from the optical resource state by performing a homodyne measurement on all but one mode in each macronode in the q-quadrature. Homodyne measurements in the p-quadrature are performed on the remaining mode in each macronode as the quantum computer is considered to be operating in quantum error correction mode. The quantum error correction mode involves performing X basis measurements on all qubits in the cluster state. To perform a Z basis measurement when the quantum computer is in another mode, the remaining mode in the macronode should be measured in the q-quadrature. We note that the qubit Pauli error rate of a qubit increases monotonically with variance $\sigma^2$ (decreases monotonically with $\sigma^2/\sigma_{\mathrm{vac}}^2[\mathrm{dB}]$) and the degree of the node in the cluster state.
    \item[4)] \textbf{Applying feedforward corrections:} The homodyne measurements are processed to obtain the GKP qubit measurement value in the GKP cluster state. This involves first processing measurement outcomes in the same mode in the macronode. Each processed measurement outcome is decoded using a soft-in soft-out GKP decoder, which is the binning function described in \cite{tzitrin2021staticlinear} to obtain a bit value and a probability of error. Feedforward corrective displacements are needed to obtain the bit value and probability of error corresponding to the GKP qubits in the GKP cluster state. These are obtained by combining the bit values and probability of error of a particular mode of a macronode with a few modes of the macronode corresponding to its neighbors in the cluster state to obtain the qubit value and the probability of qubit error. 
    \item[5)] \textbf{Syndrome computation and qubit decoding:} The syndrome is computed from these bit values and the qubit weight $w_q$ is computed from the probability of qubit error $p_q$ as $w_q = \mathrm{log}((1-p_q)/p_q)$. The syndrome and the qubit weights are used to perform soft-in hard-out qubit decoding using the foliated quantum code. We use the belief propagation (BP)-ordered statistics decoder (OSD) to perform qubit decoding.
\end{itemize}

We use the min-sum algorithm for BP decoding with $N$ iterations and a flooding schedule, where $N$ is the number of qubits in the foliated cluster state. When the BP decoding does not converge, we use an OSD with combination sweep strategy with search depth parameter $\lambda = 60$. As our homological measurement technique can be used across the class of quantum LDPC codes, we use the BP-OSD decoder, which can be used with all quantum LDPC codes.

We use logical block error rate to quantify the performance of our architecture encoded with the example codes before and after the homological measurements. We perform Monte Carlo simulations to obtain estimates of the logical error rate. 

For the logical block error rate computation, we consider a logical error to have occurred in a trial when any logical qubit has an error after decoding. Denoting the number of times at least a logical error occurs by $n_{\mathrm{fail}}$ and $n_{\mathrm{tot}}$ the total number of trials, using the Agresti-Coull interval \cite{agresti1998approximate,Brown2001}, the logical block error rate is estimated by
\begin{align}
    p_{\text{fail}} = \frac{n_{\mathrm{fail}} + \kappa^2/2}{n_{\mathrm{tot}}+\kappa^2},
\end{align}
and the uncertainties are estimated by
\begin{align}
    \kappa \sqrt{\frac{p_{\mathrm{fail}}(1-p_{\mathrm{fail}})}{n_{\mathrm{tot}} + \kappa^2}},
\end{align}
with $p_{\mathrm{fail}}$ denoting the estimated probability of sustaining a logical block error, and where $\kappa$, the desired quantile of a standard normal distribution, is set to $\kappa = 1.96$ corresponding to a $95\%$ confidence interval.

It is worth emphasizing that the Agresti-Coull interval is used rather than the traditional Wald interval because the Wald interval is simply a bad choice. For a large number of trials, such as the present case, the Agresti-Coull interval is highly performant, while in other situations, both the Jeffreys and Wilson intervals should be considered \cite{Brown2001}.

\begin{figure*}[ht]
    \centering
    \subfloat[{$[\![175, 19, 10]\!]$ lifted product code, $\mathrm{LP}_1$.} \label{subfig:lp1}]{\includegraphics[width=.47\textwidth]{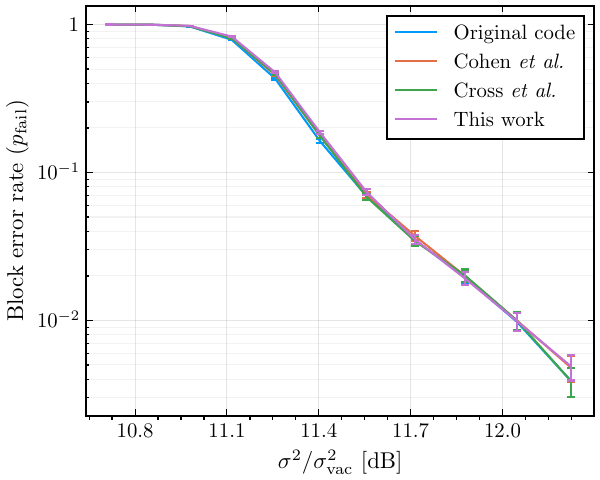}\quad}
    \subfloat[{$[\![225, 21, 12]\!]$ lifted product code, $\mathrm{LP}_2$.} \label{subfig:lp2}]{\quad\includegraphics[width=.47\textwidth]{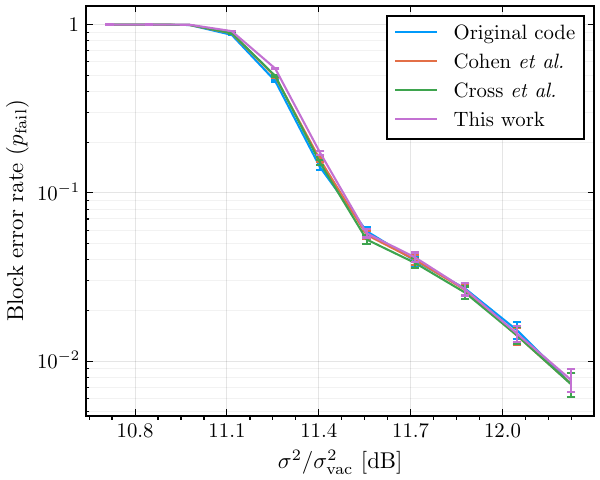}}
    \caption{Block error rates of lifted product codes before and after a logical X measurement.}\label{plot:LP_X}
\end{figure*}

\Cref{plot:LP_X} shows the logical block error rate for the lifted product codes $\mathrm{LP}_1$ and $\mathrm{LP}_2$. Our protocol performs similarly to those presented in~\cite{Cohen2021,cross2024linearsize} while using a smaller number of additional ancilla qubits. Note that the block error rates are also similar to that of the original codes.

\begin{figure*}[ht]
    \centering
    \subfloat[{$[\![625, 25, 8]\!]$ hypergraph product code, $\mathrm{HGP}_1$.} \label{subfig:hgp2}]{\includegraphics[width=0.47\textwidth]{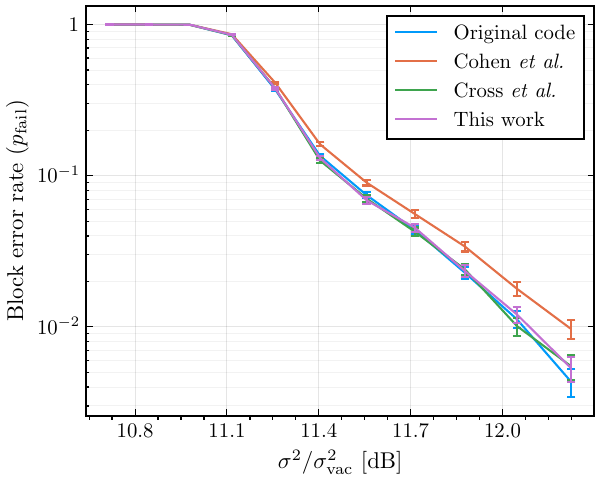}\quad}
    \subfloat[{$[\![900, 36, 10]\!]$ hypergraph product code, $\mathrm{HGP}_2$.} \label{subfig:hgp3}]{\quad\includegraphics[width=0.47\textwidth]{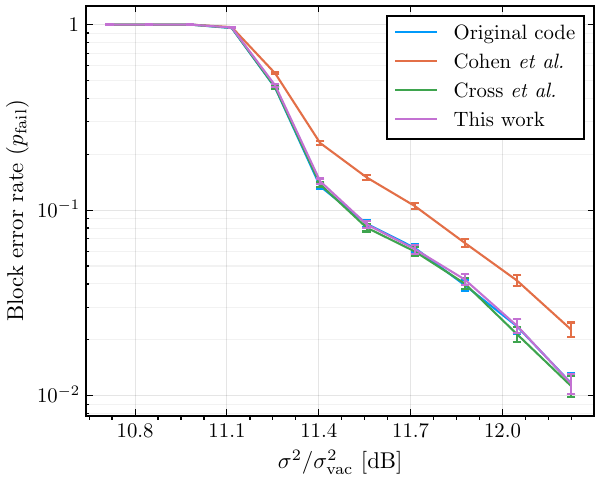}}
    \caption{Performance plot analyzing the block error rates of hypergraph product codes before and after a logical X measurement.}\label{plot:HGP_X}
\end{figure*}

\Cref{plot:HGP_X} shows the logical block error rate of the hypergraph product codes $\mathrm{HGP}_1$ and $\mathrm{HGP}_2$. For these codes, our scheme performs equivalently to the scheme presented in~\cite{cross2024linearsize}, while outperforming generalized lattice surgery~\cite{Cohen2021}. This may be due to the fact that the ancilla code used by generalized lattice surgery repeats a potentially bad structure $d$ times that depends on the specific logical representative being measured, in combination with the fact that the ancilla code in this scheme is significantly larger than for the other protocols considered.

\begin{table*}[ht]
    \centering
    \begin{tabular}{|c|c|c|c|}
    \hline
    Code & $n_{\text{anc}}$, Cohen \emph{et al.} \cite{Cohen2021} & $n_{\text{anc}}$, Cross \emph{et al.} \cite{cross2024linearsize} & $n_{\text{anc}}$, this work \\
    \hline\hline
    $\mathrm{LP}_1$, $[\![175, 19, 10]\!]$ & 230 & 38 & 16 \\
    $\mathrm{LP}_2$, $[\![225, 21, 12]\!]$ & 348 & 48 & 20 \\
    $\mathrm{HGP}_1$, $[\![625,  25,  8]\!]$ & 144 & 30 & 13 \\
    $\mathrm{HGP}_2$, $[\![900,  36,  10]\!]$ & 240 & 40 & 17 \\
    \hline
    \end{tabular}
    \caption{Comparison of the number of ancillas required for different measurement techniques. The logical operators that were measured are given in \cref{tab:logical_operators}.}
    \label{tab:ancilla_comparison}
\end{table*}

\section{Discussion} \label{sec:discussion}

We introduced the notion of homological measurements based on the chain complex description of CSS codes. Several protocols such as lattice surgery, generalized lattice surgery and some optimized variants thereof fall within the framework of homological measurements. The homological measurement framework provides a natural mathematical language that allowed us to devise a new protocol, the edge expanded homological measurement, which provides fault-tolerant measurement of logical Pauli operators with fewer ancilla qubits than in previously known methods such as generalized lattice surgery~\cite{Cohen2021}, its recent refinement~\cite{cross2024linearsize}, and homomorphic measurements~\cite{Huang2023}. Note that our scheme is general in the sense that it can be used for any CSS code. As it doesn't require any specific structure, it is likely that our method can be optimized further for combinations of specific codes and logical operators to be measured, similar to the treatment that was given to the gross code in \cite{cross2024linearsize}.

The recent work~\cite{williamson2024low} uses many related ideas but they take a larger number of ancilla qubits in general in exchange for a provably qLDPC construction. We quickly examine the case where their construction (for a \emph{CSS} code) can contain all cellulations within a single layer of their ancilla system for comparison. Let $\partial_1'$ be the incidence matrix of the graph obtained in \cref{alg:main-construction} before cellulation is applied, and $\partial_1$ be the incidence matrix including the cellulation. Let $\partial_0$, $f_1$, and $f_0$ be obtained as in \cref{alg:main-construction}. While our work results in a merged code with stabilizer matrices
\begin{equation}
    \widetilde H_X = \begin{pmatrix}
        H_X & \\
        f_1^{\mathrm T} & \partial_1^{\mathrm T}
    \end{pmatrix}, \quad \widetilde H_Z = \begin{pmatrix}
        H_Z & f_0 \\
         & \partial_0
    \end{pmatrix},
\end{equation}
this related work \cite{williamson2024low} results in
\begin{align}
    \widetilde H_X &= \begin{pmatrix}
        H_X &  &  & \\
        f_1^{\mathrm T} & {\partial_1'}^{\mathrm T} &  & I \\
         &  & \partial_1^{\mathrm T} & I
    \end{pmatrix}, \\
    \widetilde H_Z &= \begin{pmatrix}
        H_Z & f_0 &  & \\
         & I & \left(I\ 0\right) & \partial_1' \\
         &  & \partial_0 &
    \end{pmatrix},
\end{align}
where the zero-columns of $\begin{pmatrix}I & 0\end{pmatrix}$ correspond to the edges added for cellulation. We note that viewing this within our framework, a minor modification of \cref{thm:x-dist} shows that the mutligraph with all edges from both $\partial_1$ and $\partial_1'$ needs to have a Cheeger constant of 1 to maintain distance in the merged code. In this particular case, it implies that a Cheeger constant of only $0.5$ is necessary for the graph defined by $\partial_1'$, resulting in a useful optimization. When their construction requires more layers, the required Cheeger constant reduces further. Also note that for CSS codes this technique fits nicely within the homological measurement framework.

We performed numerical simulations of a photonic architecture based on the use of GKP qubits which demonstrated the practicality of the protocol. Since increasing the weight and connectivity of stabilizers measurements strongly negatively affects the effective qubit noise in the architecture on which our simulations are based~\cite{tzitrin2021staticlinear,sabo2024weight,walshe2024}, and the simulation results show that our protocol is very robust to noise, this demonstrates that the slight increase in connectivity of a limited number of stabilizer generators is not impactful in terms of logical performance. The strong resilience to noise on ancilla qubits used to perform lattice surgery in the case surface code has been previously observed~\cite{ramette2023faulttolerant}, and our results indicate that this may also be the case for more general classes of homological measurements. We have not analyzed the logical performance of edge expanded homological measurements beyond the results presented in~\cref{sec:numerics}, and in particular, we have made no attempt to optimize with regards of logical operator representative or decoding strategies to take into account the specific structures introduced by our protocol. A natural open question is how to optimize for specific logical operator structures leading to better logical performance. We leave such analysis and optimization for future work.

We have focused on designing protocols minimizing the number of ancilla qubits required, but we have not attempted to reduce the time overhead. Our protocol requires measuring the new stabilizers roughly $d$ times to be able to infer reliably the outcome of the logical measurement performed. However, if the initial codes allows for single-shot quantum error correction~\cite{PhysRevX.5.031043, Campbell_2019, PRXQuantum.2.020340}, one can wonder whether a scheme could be devised such that only a constant number of measurement rounds are needed. Since meta-stabilizer CSS codes can readily be described by longer chain complexes, we conjecture that the homological measurement framework will allow to find protocols reducing the number of measurements round required.

Note that the protocol we propose does not guarantee that an input qLDPC code will remain so, similar to the situation in \cite{cross2024linearsize} where no general bound can be given on the weight of the gauges that are promoted to stabilizers. Our focus was on practical constructions for near term devices, and the connectivity obtained for all specific constructions considered during the development of this work were low in practice. Using the decongestion lemma~\cite{freedman2021building}, we know that the asymptotic bound on connectivity increases polylogarithmically with distance in the worst case, though this worst case was never observed. While we speculate that a modification of \cref{alg:increase-cheeger} can give a provably qLDPC output code when the input is also qLDPC with only minimal extra structural requirements, we leave such exploration for future work. Alternatively, one could apply ideas of quantum weight reduction~\cite{sabo2024weight, hastings2016, hastings2023quantum} on the obtained code to reduce the connectivity to constant values. While such techniques typically incur somewhat large overheads, the number of stabilizers requiring weight reduction might be small enough so that it is still favorable compared to other measurement schemes.

As we presented in the main text, it is straightforward to measure different commuting logical Pauli operators in parallel, though measuring more operators will increase the connectivity of the resulting code in general, which in practice will limit the amount of achievable parallelization. This depends highly on the specific codes as well as logical operator to be measured and must be considered on a case-by-case basis.

\begin{acknowledgments}
We became aware of work by Dominic Williamson and Ted Yoder \cite{williamson2024low} on a similar construction, but we had no knowledge of any details while preparing this manuscript. The first versions of our manuscript and theirs were released simultaneously.

We thank Eric Sabo and Michael Vasmer for valuable discussions, as well as Rafael Alexander for carefully reviewing the manuscript. We thank Alexander Cowtan for noticing an error in a previous version of the manuscript and for his careful review. We also express our gratitude towards Professor John Sipe, for facilitating the access to numerical resources. This research was enabled in part by support provided by Scinet and the Digital Research Alliance of Canada. Computations were performed on the Niagara~\cite{niagara} supercomputer at the SciNet HPC Consortium, as well as on Graham system. SciNet is funded by Innovation, Science and Economic Development Canada; the Digital Research Alliance of Canada; the Ontario Research Fund: Research Excellence; and the University of Toronto.
\end{acknowledgments}

\appendix

\section{The mapping cylinder}\label{app:cylinder}

In addition to the mapping cone which is used in the main text, we also considered the mapping cylinder. This did not result in a useful measurement scheme. We present the construction here, along with the way it which it fails for implementing logical measurements, in the hope that it might be useful in future work.

The \emph{mapping cylinder} of a chain map \(f: \mathcal A \to \mathcal C\) is the chain complex with spaces $\mathrm{cyl}(f)_i = C_i \oplus A_{i-1} \oplus A_i$ and maps $\partial_i^{\mathrm{cyl}(f)}: \mathrm{cyl}(f)_i \to \mathrm{cyl}(f)_{i - 1}$ defined by
\begin{equation}\label{eq:cylinder}
    \begin{gathered}
    \begin{tikzcd}
        {A_{i}} & {A_{i-1}} \\
        {A_{i-1}} & {A_{i-2}} \\
        {C_i} & {C_{i-1}} \\
        {\mathrm{cyl}(f)_i} & {\mathrm{cyl}(f)_{i-1}}
        \arrow["{f_{i-1}}", from=2-1, to=3-2]
        \arrow["{I}", from=2-1, to=1-2]
        \arrow["{\partial^{\mathcal C}_i}"', from=3-1, to=3-2]
        \arrow["{\partial^{\mathcal A}_{i-1}}", from=2-1, to=2-2]
        \arrow["{\partial^{\mathcal A}_{i}}", from=1-1, to=1-2]
        \arrow["{\partial_i^{\mathrm{cyl}(f)}}", from=4-1, to=4-2]
        \arrow["\oplus", phantom, from=2-1, to=3-1]
        \arrow["\oplus", phantom, from=2-2, to=3-2]
        \arrow["\oplus", phantom, from=1-1, to=2-1]
        \arrow["\oplus", phantom, from=1-2, to=2-2]
    \end{tikzcd} \\
    \partial_i^{\mathrm{cyl}(f)} = \begin{pNiceArray}{ccc}[first-row, first-col]
        & C_i & A_{i-1} & A_i \\
        C_{i-1} & \partial^{\mathcal C}_i & -f_{i-1} & 0 \\
        A_{i-2} & 0 & -\partial^{\mathcal A}_{i-1} & 0 \\
        A_{i-1} & 0 & I & \partial_i^{\mathcal A}
    \end{pNiceArray}.
    \end{gathered}
\end{equation}
The minus signs in \(\partial_i^{\mathrm{cyl}(f)}\) can be dropped when working in \(\mathbb F_2\).

Note that we can interpret the mapping cylinder as the mapping cone of a different chain map. Define the chain complex $\mathcal B$ with spaces $B_i = A_i \oplus A_{i + 1}$ and maps
\begin{equation}
    \partial_i^{\mathcal B} = \begin{pmatrix} \partial_i^{\mathcal A} & \\ -I & -\partial_{i + 1}^{\mathcal A}\end{pmatrix}
\end{equation}
and the chain map $g: \mathcal B \to \mathcal C$ given by
\begin{equation}
    g_i = \begin{pmatrix}f_i & 0\end{pmatrix}.
\end{equation}
It can be seen that $\mathrm{cyl}(f) = \mathrm{cone}(g)$.

Suppose that we have a CSS code \(\mathrm{CSS}(H_X, H_Z)\) viewed as a chain,
\begin{equation}
    \mathcal{C}:\ C_2 \xrightarrow{H_X^{\mathrm T}} C_1 \xrightarrow{H_Z} C_0,
\end{equation}
an ancilla chain
\begin{equation}
    \mathcal{A}:\ {A_2 \xrightarrow{\partial_2} A_1 \xrightarrow{\partial_1} A_0 \xrightarrow{\partial_0} A_{-1}},
\end{equation}
and a chain map \(f: \mathcal A \to \mathcal C\). Define the \emph{cylinder code} by the diagram
\begin{equation}\label{cylindercode}
    \begin{tikzcd}
    & {A_2} & {A_1} & {A_0} \\
    & {A_1} & {A_0} & {A_{-1}} \\
	& {C_2} & {C_1} & {C_0} \\
	\arrow["{\partial_2}", from=1-2, to=1-3]
	\arrow["{\partial_1}", from=1-3, to=1-4]
	\arrow["{\partial_1}", from=2-2, to=2-3]
	\arrow["{\partial_0}", from=2-3, to=2-4]
    \arrow["{I}", from=2-2, to=1-3]
    \arrow["{I}", from=2-3, to=1-4]
	\arrow["{H_X^{\mathrm T}}"', from=3-2, to=3-3]
	\arrow["{H_Z}"', from=3-3, to=3-4]
	\arrow["{f_1}", from=2-2, to=3-3]
	\arrow["{f_0}", from=2-3, to=3-4]
    \arrow["\oplus", phantom, from=1-2, to=2-2]
    \arrow["\oplus", phantom, from=1-3, to=2-3]
    \arrow["\oplus", phantom, from=1-4, to=2-4]
    \arrow["\oplus", phantom, from=2-2, to=3-2]
    \arrow["\oplus", phantom, from=2-3, to=3-3]
    \arrow["\oplus", phantom, from=2-4, to=3-4]
    \end{tikzcd}
\end{equation}
which has stabilizer matrices
\begin{equation}\label{cylindercode-mats}
    \widetilde{H}_X = \begin{pmatrix}
        H_X & & \\
        f_1^{\mathrm T} & \partial_1^{\mathrm T} & I \\
        & & \partial_2^{\mathrm T}
    \end{pmatrix}, \quad
    \widetilde{H}_Z = \begin{pmatrix}
        H_Z & f_0 & \\
            & \partial_0 & \\
            & I & \partial_1
    \end{pmatrix}.
\end{equation}
Note that since \(\partial_0\partial_1 = 0\), all rows resulting from \(\partial_0\) in \(\widetilde H_Z\) are redundant, so we can instead choose
\begin{equation}
    \widetilde{H}_Z = \begin{pmatrix}
        H_Z & f_0 & \\
            & I & \partial_1
    \end{pmatrix}.
\end{equation}

Since this is a special case of the cone, we know that \(Z\)-distance is maintained by \cref{thm:z-dist}. If \(f_1\), \(\partial_1\), and $f_0$ are given as in \cref{eqn:f1_def,eqn:p1*_def,eqn:f0*_def} (without any improvements to the Cheeger constant) and \(\dim\ker\partial_1 = 1\), it can be seen that the \(X\)-distance is maintained without any other restrictions, and in particular, this holds regardless of the Cheeger constant of the graph defined by treating $\partial_1$ as an incidence matrix. Also, since \(\partial_1\partial_2 = 0\), there are only two choices for \(\partial_2\) in this case. It is either 0, or the all-ones vector. In the latter case, this has the same weight as the logical to measure and so it is not useful, and in the former case, the logical we wished to measure is not in the row space of \(\widetilde H_X\).

Note that by choosing \(\partial_2 = 0\) and promoting the gauge this creates to a \(Z\)-stabilizer, the \(L=2\) construction from \cite{cross2024linearsize} is obtained (though, because the \(X\)-distance is maintained regardless, there is no need to promote the gauge to a stabilizer). For larger even values of \(L\), there continues to be no way to get the logical into the row space of \(\widetilde H_X\). Note also that for even values of \(L\), there is no reason to go beyond \(L=2\) because the purpose of increasing \(L\) is to maintain distance. We do not believe that this construction allows for a stand-alone measurement protocol. While the mapping cylinder is not useful for directly measuring logical operators, the structure is related to two joint measurement schemes~\cite{zhang2024timeefficient,swaroop2024universaladaptersquantumldpc}.

It is our hope, inspired by the utility of the mapping cone in the main text, that formalizing the idea of the mapping cylinder may also be of use in future applications.

\bibliographystyle{apsrev4-2}
\bibliography{refs}

%apsrev4-2.bst 2019-01-14 (MD) hand-edited version of apsrev4-1.bst
%Control: key (0)
%Control: author (72) initials jnrlst
%Control: editor formatted (1) identically to author
%Control: production of article title (-1) disabled
%Control: page (0) single
%Control: year (1) truncated
%Control: production of eprint (0) enabled
\begin{thebibliography}{74}%
\makeatletter
\providecommand \@ifxundefined [1]{%
 \@ifx{#1\undefined}
}%
\providecommand \@ifnum [1]{%
 \ifnum #1\expandafter \@firstoftwo
 \else \expandafter \@secondoftwo
 \fi
}%
\providecommand \@ifx [1]{%
 \ifx #1\expandafter \@firstoftwo
 \else \expandafter \@secondoftwo
 \fi
}%
\providecommand \natexlab [1]{#1}%
\providecommand \enquote  [1]{``#1''}%
\providecommand \bibnamefont  [1]{#1}%
\providecommand \bibfnamefont [1]{#1}%
\providecommand \citenamefont [1]{#1}%
\providecommand \href@noop [0]{\@secondoftwo}%
\providecommand \href [0]{\begingroup \@sanitize@url \@href}%
\providecommand \@href[1]{\@@startlink{#1}\@@href}%
\providecommand \@@href[1]{\endgroup#1\@@endlink}%
\providecommand \@sanitize@url [0]{\catcode `\\12\catcode `\$12\catcode
  `\&12\catcode `\#12\catcode `\^12\catcode `\_12\catcode `\%12\relax}%
\providecommand \@@startlink[1]{}%
\providecommand \@@endlink[0]{}%
\providecommand \url  [0]{\begingroup\@sanitize@url \@url }%
\providecommand \@url [1]{\endgroup\@href {#1}{\urlprefix }}%
\providecommand \urlprefix  [0]{URL }%
\providecommand \Eprint [0]{\href }%
\providecommand \doibase [0]{https://doi.org/}%
\providecommand \selectlanguage [0]{\@gobble}%
\providecommand \bibinfo  [0]{\@secondoftwo}%
\providecommand \bibfield  [0]{\@secondoftwo}%
\providecommand \translation [1]{[#1]}%
\providecommand \BibitemOpen [0]{}%
\providecommand \bibitemStop [0]{}%
\providecommand \bibitemNoStop [0]{.\EOS\space}%
\providecommand \EOS [0]{\spacefactor3000\relax}%
\providecommand \BibitemShut  [1]{\csname bibitem#1\endcsname}%
\let\auto@bib@innerbib\@empty
%</preamble>
\bibitem [{\citenamefont {Shor}(1995)}]{Shor1995}%
  \BibitemOpen
  \bibfield  {author} {\bibinfo {author} {\bibfnamefont {P.~W.}\ \bibnamefont
  {Shor}},\ }\href {https://doi.org/10.1103/PhysRevA.52.R2493} {\bibfield
  {journal} {\bibinfo  {journal} {Phys. Rev. A}\ }\textbf {\bibinfo {volume}
  {52}},\ \bibinfo {pages} {R2493} (\bibinfo {year} {1995})}\BibitemShut
  {NoStop}%
\bibitem [{\citenamefont {Gottesman}(1997)}]{gottesman1997stabilizer}%
  \BibitemOpen
  \bibfield  {author} {\bibinfo {author} {\bibfnamefont {D.}~\bibnamefont
  {Gottesman}},\ }\emph {\bibinfo {title} {Stabilizer Codes and Quantum Error
  Correction}},\ \href
  {https://doi.org/https://doi.org/10.48550/arXiv.quant-ph/9705052} {Ph.D.
  thesis},\ \bibinfo  {school} {California Institute of Technology, division of
  Physics, Mathematics and Astronomy} (\bibinfo {year} {1997})\BibitemShut
  {NoStop}%
\bibitem [{\citenamefont {Kitaev}(2003)}]{KITAEV2003-2}%
  \BibitemOpen
  \bibfield  {author} {\bibinfo {author} {\bibfnamefont {A.}~\bibnamefont
  {Kitaev}},\ }\href
  {https://doi.org/https://doi.org/10.1016/S0003-4916(02)00018-0} {\bibfield
  {journal} {\bibinfo  {journal} {Annals of Physics}\ }\textbf {\bibinfo
  {volume} {303}},\ \bibinfo {pages} {2} (\bibinfo {year} {2003})}\BibitemShut
  {NoStop}%
\bibitem [{\citenamefont {Bombin}\ and\ \citenamefont
  {Martin-Delgado}(2006)}]{Bombin2006}%
  \BibitemOpen
  \bibfield  {author} {\bibinfo {author} {\bibfnamefont {H.}~\bibnamefont
  {Bombin}}\ and\ \bibinfo {author} {\bibfnamefont {M.~A.}\ \bibnamefont
  {Martin-Delgado}},\ }\href {https://doi.org/10.1103/PhysRevLett.97.180501}
  {\bibfield  {journal} {\bibinfo  {journal} {Phys. Rev. Lett.}\ }\textbf
  {\bibinfo {volume} {97}},\ \bibinfo {pages} {180501} (\bibinfo {year}
  {2006})}\BibitemShut {NoStop}%
\bibitem [{\citenamefont {Dennis}\ \emph {et~al.}(2002)\citenamefont {Dennis},
  \citenamefont {Kitaev}, \citenamefont {Landahl},\ and\ \citenamefont
  {Preskill}}]{Dennis2002}%
  \BibitemOpen
  \bibfield  {author} {\bibinfo {author} {\bibfnamefont {E.}~\bibnamefont
  {Dennis}}, \bibinfo {author} {\bibfnamefont {A.}~\bibnamefont {Kitaev}},
  \bibinfo {author} {\bibfnamefont {A.}~\bibnamefont {Landahl}},\ and\ \bibinfo
  {author} {\bibfnamefont {J.}~\bibnamefont {Preskill}},\ }\href
  {https://doi.org/10.1063/1.1499754} {\bibfield  {journal} {\bibinfo
  {journal} {Journal of Mathematical Physics}\ }\textbf {\bibinfo {volume}
  {43}},\ \bibinfo {pages} {4452} (\bibinfo {year} {2002})},\ \Eprint
  {https://arxiv.org/abs/https://pubs.aip.org/aip/jmp/article-pdf/43/9/4452/19183135/4452\_1\_online.pdf}
  {https://pubs.aip.org/aip/jmp/article-pdf/43/9/4452/19183135/4452\_1\_online.pdf}
  \BibitemShut {NoStop}%
\bibitem [{\citenamefont {Bombin}\ \emph {et~al.}(2012)\citenamefont {Bombin},
  \citenamefont {Andrist}, \citenamefont {Ohzeki}, \citenamefont {Katzgraber},\
  and\ \citenamefont {Martin-Delgado}}]{Bombin2012}%
  \BibitemOpen
  \bibfield  {author} {\bibinfo {author} {\bibfnamefont {H.}~\bibnamefont
  {Bombin}}, \bibinfo {author} {\bibfnamefont {R.~S.}\ \bibnamefont {Andrist}},
  \bibinfo {author} {\bibfnamefont {M.}~\bibnamefont {Ohzeki}}, \bibinfo
  {author} {\bibfnamefont {H.~G.}\ \bibnamefont {Katzgraber}},\ and\ \bibinfo
  {author} {\bibfnamefont {M.~A.}\ \bibnamefont {Martin-Delgado}},\ }\href
  {https://doi.org/10.1103/PhysRevX.2.021004} {\bibfield  {journal} {\bibinfo
  {journal} {Phys. Rev. X}\ }\textbf {\bibinfo {volume} {2}},\ \bibinfo {pages}
  {021004} (\bibinfo {year} {2012})}\BibitemShut {NoStop}%
\bibitem [{\citenamefont {Bravyi}\ \emph {et~al.}(2010)\citenamefont {Bravyi},
  \citenamefont {Poulin},\ and\ \citenamefont {Terhal}}]{BPT-2010}%
  \BibitemOpen
  \bibfield  {author} {\bibinfo {author} {\bibfnamefont {S.}~\bibnamefont
  {Bravyi}}, \bibinfo {author} {\bibfnamefont {D.}~\bibnamefont {Poulin}},\
  and\ \bibinfo {author} {\bibfnamefont {B.}~\bibnamefont {Terhal}},\ }\href
  {https://doi.org/10.1103/PhysRevLett.104.050503} {\bibfield  {journal}
  {\bibinfo  {journal} {Phys. Rev. Lett.}\ }\textbf {\bibinfo {volume} {104}},\
  \bibinfo {pages} {050503} (\bibinfo {year} {2010})}\BibitemShut {NoStop}%
\bibitem [{\citenamefont {Gidney}\ and\ \citenamefont
  {Eker{\aa{}}}(2021)}]{Gidney2021}%
  \BibitemOpen
  \bibfield  {author} {\bibinfo {author} {\bibfnamefont {C.}~\bibnamefont
  {Gidney}}\ and\ \bibinfo {author} {\bibfnamefont {M.}~\bibnamefont
  {Eker{\aa{}}}},\ }\href {https://doi.org/10.22331/q-2021-04-15-433}
  {\bibfield  {journal} {\bibinfo  {journal} {{Quantum}}\ }\textbf {\bibinfo
  {volume} {5}},\ \bibinfo {pages} {433} (\bibinfo {year} {2021})}\BibitemShut
  {NoStop}%
\bibitem [{\citenamefont {Kim}\ \emph {et~al.}(2022)\citenamefont {Kim},
  \citenamefont {Liu}, \citenamefont {Pallister}, \citenamefont {Pol},
  \citenamefont {Roberts},\ and\ \citenamefont
  {Lee}}]{PhysRevResearch.4.023019}%
  \BibitemOpen
  \bibfield  {author} {\bibinfo {author} {\bibfnamefont {I.~H.}\ \bibnamefont
  {Kim}}, \bibinfo {author} {\bibfnamefont {Y.-H.}\ \bibnamefont {Liu}},
  \bibinfo {author} {\bibfnamefont {S.}~\bibnamefont {Pallister}}, \bibinfo
  {author} {\bibfnamefont {W.}~\bibnamefont {Pol}}, \bibinfo {author}
  {\bibfnamefont {S.}~\bibnamefont {Roberts}},\ and\ \bibinfo {author}
  {\bibfnamefont {E.}~\bibnamefont {Lee}},\ }\href
  {https://doi.org/10.1103/PhysRevResearch.4.023019} {\bibfield  {journal}
  {\bibinfo  {journal} {Phys. Rev. Res.}\ }\textbf {\bibinfo {volume} {4}},\
  \bibinfo {pages} {023019} (\bibinfo {year} {2022})}\BibitemShut {NoStop}%
\bibitem [{\citenamefont {Kovalev}\ and\ \citenamefont
  {Pryadko}(2012)}]{Kovalev2012}%
  \BibitemOpen
  \bibfield  {author} {\bibinfo {author} {\bibfnamefont {A.~A.}\ \bibnamefont
  {Kovalev}}\ and\ \bibinfo {author} {\bibfnamefont {L.~P.}\ \bibnamefont
  {Pryadko}},\ }in\ \href {https://doi.org/10.1109/ISIT.2012.6284206} {\emph
  {\bibinfo {booktitle} {2012 IEEE International Symposium on Information
  Theory Proceedings}}}\ (\bibinfo {year} {2012})\ pp.\ \bibinfo {pages}
  {348--352}\BibitemShut {NoStop}%
\bibitem [{\citenamefont {Tillich}\ and\ \citenamefont
  {Zémor}(2014{\natexlab{a}})}]{Tillich2014}%
  \BibitemOpen
  \bibfield  {author} {\bibinfo {author} {\bibfnamefont {J.-P.}\ \bibnamefont
  {Tillich}}\ and\ \bibinfo {author} {\bibfnamefont {G.}~\bibnamefont
  {Zémor}},\ }\href {https://doi.org/10.1109/TIT.2013.2292061} {\bibfield
  {journal} {\bibinfo  {journal} {IEEE Transactions on Information Theory}\
  }\textbf {\bibinfo {volume} {60}},\ \bibinfo {pages} {1193} (\bibinfo {year}
  {2014}{\natexlab{a}})}\BibitemShut {NoStop}%
\bibitem [{\citenamefont {Gottesman}(2014)}]{Gottesman2014}%
  \BibitemOpen
  \bibfield  {author} {\bibinfo {author} {\bibfnamefont {D.}~\bibnamefont
  {Gottesman}},\ }\href@noop {} {\bibfield  {journal} {\bibinfo  {journal}
  {Quantum Info. Comput.}\ }\textbf {\bibinfo {volume} {14}},\ \bibinfo {pages}
  {1338–1372} (\bibinfo {year} {2014})}\BibitemShut {NoStop}%
\bibitem [{\citenamefont {Bravyi}\ and\ \citenamefont
  {Hastings}(2014)}]{Bravyi2014}%
  \BibitemOpen
  \bibfield  {author} {\bibinfo {author} {\bibfnamefont {S.}~\bibnamefont
  {Bravyi}}\ and\ \bibinfo {author} {\bibfnamefont {M.~B.}\ \bibnamefont
  {Hastings}},\ }in\ \href {https://doi.org/10.1145/2591796.2591870} {\emph
  {\bibinfo {booktitle} {Proceedings of the Forty-Sixth Annual ACM Symposium on
  Theory of Computing}}},\ \bibinfo {series and number} {STOC '14}\ (\bibinfo
  {publisher} {Association for Computing Machinery},\ \bibinfo {address} {New
  York, NY, USA},\ \bibinfo {year} {2014})\ p.\ \bibinfo {pages}
  {273–282}\BibitemShut {NoStop}%
\bibitem [{\citenamefont {Leverrier}\ \emph {et~al.}(2015)\citenamefont
  {Leverrier}, \citenamefont {Tillich},\ and\ \citenamefont
  {Zémor}}]{Leverrier2015}%
  \BibitemOpen
  \bibfield  {author} {\bibinfo {author} {\bibfnamefont {A.}~\bibnamefont
  {Leverrier}}, \bibinfo {author} {\bibfnamefont {J.-P.}\ \bibnamefont
  {Tillich}},\ and\ \bibinfo {author} {\bibfnamefont {G.}~\bibnamefont
  {Zémor}},\ }in\ \href {https://doi.org/10.1109/FOCS.2015.55} {\emph
  {\bibinfo {booktitle} {2015 IEEE 56th Annual Symposium on Foundations of
  Computer Science}}}\ (\bibinfo {address} {New York, New York},\ \bibinfo
  {year} {2015})\ pp.\ \bibinfo {pages} {810--824}\BibitemShut {NoStop}%
\bibitem [{\citenamefont {Breuckmann}\ and\ \citenamefont
  {Terhal}(2016)}]{Breuckmann2016}%
  \BibitemOpen
  \bibfield  {author} {\bibinfo {author} {\bibfnamefont {N.~P.}\ \bibnamefont
  {Breuckmann}}\ and\ \bibinfo {author} {\bibfnamefont {B.~M.}\ \bibnamefont
  {Terhal}},\ }\href {https://doi.org/10.1109/TIT.2016.2555700} {\bibfield
  {journal} {\bibinfo  {journal} {IEEE Transactions on Information Theory}\
  }\textbf {\bibinfo {volume} {62}},\ \bibinfo {pages} {3731} (\bibinfo {year}
  {2016})}\BibitemShut {NoStop}%
\bibitem [{\citenamefont {Breuckmann}\ \emph {et~al.}(2017)\citenamefont
  {Breuckmann}, \citenamefont {Vuillot}, \citenamefont {Campbell},
  \citenamefont {Krishna},\ and\ \citenamefont {Terhal}}]{Breuckmann2017}%
  \BibitemOpen
  \bibfield  {author} {\bibinfo {author} {\bibfnamefont {N.~P.}\ \bibnamefont
  {Breuckmann}}, \bibinfo {author} {\bibfnamefont {C.}~\bibnamefont {Vuillot}},
  \bibinfo {author} {\bibfnamefont {E.}~\bibnamefont {Campbell}}, \bibinfo
  {author} {\bibfnamefont {A.}~\bibnamefont {Krishna}},\ and\ \bibinfo {author}
  {\bibfnamefont {B.~M.}\ \bibnamefont {Terhal}},\ }\href
  {https://doi.org/10.1088/2058-9565/aa7d3b} {\bibfield  {journal} {\bibinfo
  {journal} {Quantum Science and Technology}\ }\textbf {\bibinfo {volume}
  {2}},\ \bibinfo {pages} {035007} (\bibinfo {year} {2017})}\BibitemShut
  {NoStop}%
\bibitem [{\citenamefont {Fawzi}\ \emph {et~al.}(2018)\citenamefont {Fawzi},
  \citenamefont {Grospellier},\ and\ \citenamefont {Leverrier}}]{Fawzi2018}%
  \BibitemOpen
  \bibfield  {author} {\bibinfo {author} {\bibfnamefont {O.}~\bibnamefont
  {Fawzi}}, \bibinfo {author} {\bibfnamefont {A.}~\bibnamefont {Grospellier}},\
  and\ \bibinfo {author} {\bibfnamefont {A.}~\bibnamefont {Leverrier}},\ }in\
  \href {https://doi.org/10.1109/FOCS.2018.00076} {\emph {\bibinfo {booktitle}
  {2018 IEEE 59th Annual Symposium on Foundations of Computer Science
  (FOCS)}}}\ (\bibinfo {address} {New York, New York},\ \bibinfo {year}
  {2018})\ pp.\ \bibinfo {pages} {743--754}\BibitemShut {NoStop}%
\bibitem [{\citenamefont {Soares}\ and\ \citenamefont
  {Da~Silva}(2018)}]{dasilva2018}%
  \BibitemOpen
  \bibfield  {author} {\bibinfo {author} {\bibfnamefont {W.~S.}\ \bibnamefont
  {Soares}}\ and\ \bibinfo {author} {\bibfnamefont {E.~B.}\ \bibnamefont
  {Da~Silva}},\ }\href@noop {} {\bibfield  {journal} {\bibinfo  {journal}
  {Quantum Info. Comput.}\ }\textbf {\bibinfo {volume} {18}},\ \bibinfo {pages}
  {306–318} (\bibinfo {year} {2018})}\BibitemShut {NoStop}%
\bibitem [{\citenamefont {Hastings}\ \emph {et~al.}(2021)\citenamefont
  {Hastings}, \citenamefont {Haah},\ and\ \citenamefont
  {O'Donnell}}]{Hastings2020}%
  \BibitemOpen
  \bibfield  {author} {\bibinfo {author} {\bibfnamefont {M.~B.}\ \bibnamefont
  {Hastings}}, \bibinfo {author} {\bibfnamefont {J.}~\bibnamefont {Haah}},\
  and\ \bibinfo {author} {\bibfnamefont {R.}~\bibnamefont {O'Donnell}},\ }in\
  \href {https://doi.org/10.1145/3406325.3451005} {\emph {\bibinfo {booktitle}
  {Proceedings of the 53rd Annual ACM SIGACT Symposium on Theory of
  Computing}}},\ \bibinfo {series and number} {STOC 2021}\ (\bibinfo
  {publisher} {Association for Computing Machinery},\ \bibinfo {address} {New
  York, NY, USA},\ \bibinfo {year} {2021})\ p.\ \bibinfo {pages}
  {1276–1288}\BibitemShut {NoStop}%
\bibitem [{\citenamefont {Evra}\ \emph {et~al.}(2022)\citenamefont {Evra},
  \citenamefont {Kaufman},\ and\ \citenamefont {Z{\'e}mor}}]{Evra2020}%
  \BibitemOpen
  \bibfield  {author} {\bibinfo {author} {\bibfnamefont {S.}~\bibnamefont
  {Evra}}, \bibinfo {author} {\bibfnamefont {T.}~\bibnamefont {Kaufman}},\ and\
  \bibinfo {author} {\bibfnamefont {G.}~\bibnamefont {Z{\'e}mor}},\ }\href
  {https://doi.org/10.1137/20M1383689} {\bibfield  {journal} {\bibinfo
  {journal} {SIAM Journal on Computing}\ ,\ \bibinfo {pages} {FOCS20}}
  (\bibinfo {year} {2022})}\BibitemShut {NoStop}%
\bibitem [{\citenamefont {Panteleev}\ and\ \citenamefont
  {Kalachev}(2021)}]{Panteleev2021}%
  \BibitemOpen
  \bibfield  {author} {\bibinfo {author} {\bibfnamefont {P.}~\bibnamefont
  {Panteleev}}\ and\ \bibinfo {author} {\bibfnamefont {G.}~\bibnamefont
  {Kalachev}},\ }\href {https://doi.org/10.22331/q-2021-11-22-585} {\bibfield
  {journal} {\bibinfo  {journal} {{Quantum}}\ }\textbf {\bibinfo {volume}
  {5}},\ \bibinfo {pages} {585} (\bibinfo {year} {2021})}\BibitemShut {NoStop}%
\bibitem [{\citenamefont {Breuckmann}\ and\ \citenamefont
  {Eberhardt}(2021)}]{Breuckmann2021Balanced}%
  \BibitemOpen
  \bibfield  {author} {\bibinfo {author} {\bibfnamefont {N.~P.}\ \bibnamefont
  {Breuckmann}}\ and\ \bibinfo {author} {\bibfnamefont {J.~N.}\ \bibnamefont
  {Eberhardt}},\ }\href {https://doi.org/10.1109/TIT.2021.3097347} {\bibfield
  {journal} {\bibinfo  {journal} {IEEE Transactions on Information Theory}\
  }\textbf {\bibinfo {volume} {67}},\ \bibinfo {pages} {6653} (\bibinfo {year}
  {2021})}\BibitemShut {NoStop}%
\bibitem [{\citenamefont {Panteleev}\ and\ \citenamefont
  {Kalachev}(2022)}]{Panteleev2022}%
  \BibitemOpen
  \bibfield  {author} {\bibinfo {author} {\bibfnamefont {P.}~\bibnamefont
  {Panteleev}}\ and\ \bibinfo {author} {\bibfnamefont {G.}~\bibnamefont
  {Kalachev}},\ }\href {https://doi.org/10.1109/TIT.2021.3119384} {\bibfield
  {journal} {\bibinfo  {journal} {IEEE Transactions on Information Theory}\
  }\textbf {\bibinfo {volume} {68}},\ \bibinfo {pages} {213} (\bibinfo {year}
  {2022})}\BibitemShut {NoStop}%
\bibitem [{\citenamefont {Vuillot}\ and\ \citenamefont
  {Breuckmann}(2022)}]{Vuillot2022}%
  \BibitemOpen
  \bibfield  {author} {\bibinfo {author} {\bibfnamefont {C.}~\bibnamefont
  {Vuillot}}\ and\ \bibinfo {author} {\bibfnamefont {N.~P.}\ \bibnamefont
  {Breuckmann}},\ }\href {https://doi.org/10.1109/TIT.2022.3170846} {\bibfield
  {journal} {\bibinfo  {journal} {IEEE Transactions on Information Theory}\
  }\textbf {\bibinfo {volume} {68}},\ \bibinfo {pages} {5955} (\bibinfo {year}
  {2022})}\BibitemShut {NoStop}%
\bibitem [{\citenamefont {Dinur}\ \emph {et~al.}(2023)\citenamefont {Dinur},
  \citenamefont {Hsieh}, \citenamefont {Lin},\ and\ \citenamefont
  {Vidick}}]{dinur2022good}%
  \BibitemOpen
  \bibfield  {author} {\bibinfo {author} {\bibfnamefont {I.}~\bibnamefont
  {Dinur}}, \bibinfo {author} {\bibfnamefont {M.-H.}\ \bibnamefont {Hsieh}},
  \bibinfo {author} {\bibfnamefont {T.-C.}\ \bibnamefont {Lin}},\ and\ \bibinfo
  {author} {\bibfnamefont {T.}~\bibnamefont {Vidick}},\ }in\ \href@noop {}
  {\emph {\bibinfo {booktitle} {Proceedings of the 55th annual ACM symposium on
  theory of computing}}}\ (\bibinfo {address} {New York, New York},\ \bibinfo
  {year} {2023})\ pp.\ \bibinfo {pages} {905--918}\BibitemShut {NoStop}%
\bibitem [{\citenamefont {Strikis}\ and\ \citenamefont
  {Berent}(2023)}]{Strikis2023}%
  \BibitemOpen
  \bibfield  {author} {\bibinfo {author} {\bibfnamefont {A.}~\bibnamefont
  {Strikis}}\ and\ \bibinfo {author} {\bibfnamefont {L.}~\bibnamefont
  {Berent}},\ }\href {https://doi.org/10.1103/PRXQuantum.4.020321} {\bibfield
  {journal} {\bibinfo  {journal} {PRX Quantum}\ }\textbf {\bibinfo {volume}
  {4}},\ \bibinfo {pages} {020321} (\bibinfo {year} {2023})}\BibitemShut
  {NoStop}%
\bibitem [{\citenamefont {Roffe}\ \emph {et~al.}(2023)\citenamefont {Roffe},
  \citenamefont {Cohen}, \citenamefont {Quintavalle}, \citenamefont {Chandra},\
  and\ \citenamefont {Campbell}}]{Roffe2023biastailoredquantum}%
  \BibitemOpen
  \bibfield  {author} {\bibinfo {author} {\bibfnamefont {J.}~\bibnamefont
  {Roffe}}, \bibinfo {author} {\bibfnamefont {L.~Z.}\ \bibnamefont {Cohen}},
  \bibinfo {author} {\bibfnamefont {A.~O.}\ \bibnamefont {Quintavalle}},
  \bibinfo {author} {\bibfnamefont {D.}~\bibnamefont {Chandra}},\ and\ \bibinfo
  {author} {\bibfnamefont {E.~T.}\ \bibnamefont {Campbell}},\ }\href
  {https://doi.org/10.22331/q-2023-05-15-1005} {\bibfield  {journal} {\bibinfo
  {journal} {{Quantum}}\ }\textbf {\bibinfo {volume} {7}},\ \bibinfo {pages}
  {1005} (\bibinfo {year} {2023})}\BibitemShut {NoStop}%
\bibitem [{\citenamefont {Yang}\ and\ \citenamefont
  {Calderbank}(2023)}]{yang2023spatiallycoupled}%
  \BibitemOpen
  \bibfield  {author} {\bibinfo {author} {\bibfnamefont {S.}~\bibnamefont
  {Yang}}\ and\ \bibinfo {author} {\bibfnamefont {R.}~\bibnamefont
  {Calderbank}},\ }\href@noop {} {\bibinfo {title} {Spatially-coupled qdlpc
  codes}} (\bibinfo {year} {2023}),\ \Eprint {https://arxiv.org/abs/2305.00137}
  {arXiv:2305.00137 [quant-ph]} \BibitemShut {NoStop}%
\bibitem [{\citenamefont {Feng}\ \emph {et~al.}(2024)\citenamefont {Feng},
  \citenamefont {Tang},\ and\ \citenamefont {Bai}}]{feng2023new}%
  \BibitemOpen
  \bibfield  {author} {\bibinfo {author} {\bibfnamefont {Y.}~\bibnamefont
  {Feng}}, \bibinfo {author} {\bibfnamefont {C.}~\bibnamefont {Tang}},\ and\
  \bibinfo {author} {\bibfnamefont {C.}~\bibnamefont {Bai}},\ }\href
  {https://doi.org/10.1088/1555-6611/ad3aed} {\bibfield  {journal} {\bibinfo
  {journal} {Laser Physics}\ }\textbf {\bibinfo {volume} {34}},\ \bibinfo
  {pages} {065201} (\bibinfo {year} {2024})}\BibitemShut {NoStop}%
\bibitem [{\citenamefont {Miao}\ \emph {et~al.}(2024)\citenamefont {Miao},
  \citenamefont {Mandelbaum}, \citenamefont {Jäkel},\ and\ \citenamefont
  {Schmalen}}]{miao2024joint}%
  \BibitemOpen
  \bibfield  {author} {\bibinfo {author} {\bibfnamefont {S.}~\bibnamefont
  {Miao}}, \bibinfo {author} {\bibfnamefont {J.}~\bibnamefont {Mandelbaum}},
  \bibinfo {author} {\bibfnamefont {H.}~\bibnamefont {Jäkel}},\ and\ \bibinfo
  {author} {\bibfnamefont {L.}~\bibnamefont {Schmalen}},\ }\href@noop {}
  {\bibinfo {title} {A joint code and belief propagation decoder design for
  quantum ldpc codes}} (\bibinfo {year} {2024}),\ \Eprint
  {https://arxiv.org/abs/2401.06874} {arXiv:2401.06874 [cs.IT]} \BibitemShut
  {NoStop}%
\bibitem [{\citenamefont {Sabo}\ \emph {et~al.}(2024)\citenamefont {Sabo},
  \citenamefont {Gunderman}, \citenamefont {Ide}, \citenamefont {Vasmer},\ and\
  \citenamefont {Dauphinais}}]{sabo2024weight}%
  \BibitemOpen
  \bibfield  {author} {\bibinfo {author} {\bibfnamefont {E.}~\bibnamefont
  {Sabo}}, \bibinfo {author} {\bibfnamefont {L.~G.}\ \bibnamefont {Gunderman}},
  \bibinfo {author} {\bibfnamefont {B.}~\bibnamefont {Ide}}, \bibinfo {author}
  {\bibfnamefont {M.}~\bibnamefont {Vasmer}},\ and\ \bibinfo {author}
  {\bibfnamefont {G.}~\bibnamefont {Dauphinais}},\ }\href@noop {} {\bibfield
  {journal} {\bibinfo  {journal} {PRX Quantum}\ }\textbf {\bibinfo {volume}
  {5}},\ \bibinfo {pages} {040302} (\bibinfo {year} {2024})}\BibitemShut
  {NoStop}%
\bibitem [{\citenamefont {Bourassa}\ \emph {et~al.}(2021)\citenamefont
  {Bourassa}, \citenamefont {Alexander}, \citenamefont {Vasmer}, \citenamefont
  {Patil}, \citenamefont {Tzitrin}, \citenamefont {Matsuura}, \citenamefont
  {Su}, \citenamefont {Baragiola}, \citenamefont {Guha}, \citenamefont
  {Dauphinais}, \citenamefont {Sabapathy}, \citenamefont {Menicucci},\ and\
  \citenamefont {Dhand}}]{Bourassa2021blueprintscalable}%
  \BibitemOpen
  \bibfield  {author} {\bibinfo {author} {\bibfnamefont {J.~E.}\ \bibnamefont
  {Bourassa}}, \bibinfo {author} {\bibfnamefont {R.~N.}\ \bibnamefont
  {Alexander}}, \bibinfo {author} {\bibfnamefont {M.}~\bibnamefont {Vasmer}},
  \bibinfo {author} {\bibfnamefont {A.}~\bibnamefont {Patil}}, \bibinfo
  {author} {\bibfnamefont {I.}~\bibnamefont {Tzitrin}}, \bibinfo {author}
  {\bibfnamefont {T.}~\bibnamefont {Matsuura}}, \bibinfo {author}
  {\bibfnamefont {D.}~\bibnamefont {Su}}, \bibinfo {author} {\bibfnamefont
  {B.~Q.}\ \bibnamefont {Baragiola}}, \bibinfo {author} {\bibfnamefont
  {S.}~\bibnamefont {Guha}}, \bibinfo {author} {\bibfnamefont {G.}~\bibnamefont
  {Dauphinais}}, \bibinfo {author} {\bibfnamefont {K.~K.}\ \bibnamefont
  {Sabapathy}}, \bibinfo {author} {\bibfnamefont {N.~C.}\ \bibnamefont
  {Menicucci}},\ and\ \bibinfo {author} {\bibfnamefont {I.}~\bibnamefont
  {Dhand}},\ }\href {https://doi.org/10.22331/q-2021-02-04-392} {\bibfield
  {journal} {\bibinfo  {journal} {{Quantum}}\ }\textbf {\bibinfo {volume}
  {5}},\ \bibinfo {pages} {392} (\bibinfo {year} {2021})}\BibitemShut {NoStop}%
\bibitem [{\citenamefont {Tzitrin}\ \emph {et~al.}(2021)\citenamefont
  {Tzitrin}, \citenamefont {Matsuura}, \citenamefont {Alexander}, \citenamefont
  {Dauphinais}, \citenamefont {Bourassa}, \citenamefont {Sabapathy},
  \citenamefont {Menicucci},\ and\ \citenamefont
  {Dhand}}]{tzitrin2021staticlinear}%
  \BibitemOpen
  \bibfield  {author} {\bibinfo {author} {\bibfnamefont {I.}~\bibnamefont
  {Tzitrin}}, \bibinfo {author} {\bibfnamefont {T.}~\bibnamefont {Matsuura}},
  \bibinfo {author} {\bibfnamefont {R.~N.}\ \bibnamefont {Alexander}}, \bibinfo
  {author} {\bibfnamefont {G.}~\bibnamefont {Dauphinais}}, \bibinfo {author}
  {\bibfnamefont {J.~E.}\ \bibnamefont {Bourassa}}, \bibinfo {author}
  {\bibfnamefont {K.~K.}\ \bibnamefont {Sabapathy}}, \bibinfo {author}
  {\bibfnamefont {N.~C.}\ \bibnamefont {Menicucci}},\ and\ \bibinfo {author}
  {\bibfnamefont {I.}~\bibnamefont {Dhand}},\ }\href
  {https://doi.org/10.1103/PRXQuantum.2.040353} {\bibfield  {journal} {\bibinfo
   {journal} {PRX Quantum}\ }\textbf {\bibinfo {volume} {2}},\ \bibinfo {pages}
  {040353} (\bibinfo {year} {2021})}\BibitemShut {NoStop}%
\bibitem [{\citenamefont {Bartolucci}\ \emph {et~al.}(2023)\citenamefont
  {Bartolucci}, \citenamefont {Birchall}, \citenamefont {Bombin}, \citenamefont
  {Cable}, \citenamefont {Dawson}, \citenamefont {Gimeno-Segovia},
  \citenamefont {Johnston}, \citenamefont {Kieling}, \citenamefont {Nickerson},
  \citenamefont {Pant} \emph {et~al.}}]{bartolucci2021fusionbased}%
  \BibitemOpen
  \bibfield  {author} {\bibinfo {author} {\bibfnamefont {S.}~\bibnamefont
  {Bartolucci}}, \bibinfo {author} {\bibfnamefont {P.}~\bibnamefont
  {Birchall}}, \bibinfo {author} {\bibfnamefont {H.}~\bibnamefont {Bombin}},
  \bibinfo {author} {\bibfnamefont {H.}~\bibnamefont {Cable}}, \bibinfo
  {author} {\bibfnamefont {C.}~\bibnamefont {Dawson}}, \bibinfo {author}
  {\bibfnamefont {M.}~\bibnamefont {Gimeno-Segovia}}, \bibinfo {author}
  {\bibfnamefont {E.}~\bibnamefont {Johnston}}, \bibinfo {author}
  {\bibfnamefont {K.}~\bibnamefont {Kieling}}, \bibinfo {author} {\bibfnamefont
  {N.}~\bibnamefont {Nickerson}}, \bibinfo {author} {\bibfnamefont
  {M.}~\bibnamefont {Pant}}, \emph {et~al.},\ }\href@noop {} {\bibfield
  {journal} {\bibinfo  {journal} {Nature Communications}\ }\textbf {\bibinfo
  {volume} {14}},\ \bibinfo {pages} {912} (\bibinfo {year} {2023})}\BibitemShut
  {NoStop}%
\bibitem [{\citenamefont {Pankovich}\ \emph {et~al.}(2024)\citenamefont
  {Pankovich}, \citenamefont {Kan}, \citenamefont {Wan}, \citenamefont
  {Ostmann}, \citenamefont {Neville}, \citenamefont {Omkar}, \citenamefont
  {Sohbi},\ and\ \citenamefont {Br{\'a}dler}}]{pankovich2023high}%
  \BibitemOpen
  \bibfield  {author} {\bibinfo {author} {\bibfnamefont {B.}~\bibnamefont
  {Pankovich}}, \bibinfo {author} {\bibfnamefont {A.}~\bibnamefont {Kan}},
  \bibinfo {author} {\bibfnamefont {K.~H.}\ \bibnamefont {Wan}}, \bibinfo
  {author} {\bibfnamefont {M.}~\bibnamefont {Ostmann}}, \bibinfo {author}
  {\bibfnamefont {A.}~\bibnamefont {Neville}}, \bibinfo {author} {\bibfnamefont
  {S.}~\bibnamefont {Omkar}}, \bibinfo {author} {\bibfnamefont
  {A.}~\bibnamefont {Sohbi}},\ and\ \bibinfo {author} {\bibfnamefont
  {K.}~\bibnamefont {Br{\'a}dler}},\ }\href@noop {} {\bibfield  {journal}
  {\bibinfo  {journal} {Physical Review Letters}\ }\textbf {\bibinfo {volume}
  {133}},\ \bibinfo {pages} {050604} (\bibinfo {year} {2024})}\BibitemShut
  {NoStop}%
\bibitem [{\citenamefont {Raussendorf}\ \emph {et~al.}(2005)\citenamefont
  {Raussendorf}, \citenamefont {Bravyi},\ and\ \citenamefont
  {Harrington}}]{Raussendorf2005}%
  \BibitemOpen
  \bibfield  {author} {\bibinfo {author} {\bibfnamefont {R.}~\bibnamefont
  {Raussendorf}}, \bibinfo {author} {\bibfnamefont {S.}~\bibnamefont
  {Bravyi}},\ and\ \bibinfo {author} {\bibfnamefont {J.}~\bibnamefont
  {Harrington}},\ }\href {https://doi.org/10.1103/PhysRevA.71.062313}
  {\bibfield  {journal} {\bibinfo  {journal} {Phys. Rev. A}\ }\textbf {\bibinfo
  {volume} {71}},\ \bibinfo {pages} {062313} (\bibinfo {year}
  {2005})}\BibitemShut {NoStop}%
\bibitem [{\citenamefont {Raussendorf}\ \emph {et~al.}(2006)\citenamefont
  {Raussendorf}, \citenamefont {Harrington},\ and\ \citenamefont
  {Goyal}}]{RAUSSENDORF2006}%
  \BibitemOpen
  \bibfield  {author} {\bibinfo {author} {\bibfnamefont {R.}~\bibnamefont
  {Raussendorf}}, \bibinfo {author} {\bibfnamefont {J.}~\bibnamefont
  {Harrington}},\ and\ \bibinfo {author} {\bibfnamefont {K.}~\bibnamefont
  {Goyal}},\ }\href {https://doi.org/https://doi.org/10.1016/j.aop.2006.01.012}
  {\bibfield  {journal} {\bibinfo  {journal} {Annals of Physics}\ }\textbf
  {\bibinfo {volume} {321}},\ \bibinfo {pages} {2242} (\bibinfo {year}
  {2006})}\BibitemShut {NoStop}%
\bibitem [{\citenamefont {Raussendorf}\ and\ \citenamefont
  {Harrington}(2007)}]{Raussendorf2007-1}%
  \BibitemOpen
  \bibfield  {author} {\bibinfo {author} {\bibfnamefont {R.}~\bibnamefont
  {Raussendorf}}\ and\ \bibinfo {author} {\bibfnamefont {J.}~\bibnamefont
  {Harrington}},\ }\href {https://doi.org/10.1103/PhysRevLett.98.190504}
  {\bibfield  {journal} {\bibinfo  {journal} {Phys. Rev. Lett.}\ }\textbf
  {\bibinfo {volume} {98}},\ \bibinfo {pages} {190504} (\bibinfo {year}
  {2007})}\BibitemShut {NoStop}%
\bibitem [{\citenamefont {Raussendorf}\ \emph {et~al.}(2007)\citenamefont
  {Raussendorf}, \citenamefont {Harrington},\ and\ \citenamefont
  {Goyal}}]{Raussendorf2007-2}%
  \BibitemOpen
  \bibfield  {author} {\bibinfo {author} {\bibfnamefont {R.}~\bibnamefont
  {Raussendorf}}, \bibinfo {author} {\bibfnamefont {J.}~\bibnamefont
  {Harrington}},\ and\ \bibinfo {author} {\bibfnamefont {K.}~\bibnamefont
  {Goyal}},\ }\href {https://doi.org/10.1088/1367-2630/9/6/199} {\bibfield
  {journal} {\bibinfo  {journal} {New Journal of Physics}\ }\textbf {\bibinfo
  {volume} {9}},\ \bibinfo {pages} {199} (\bibinfo {year} {2007})}\BibitemShut
  {NoStop}%
\bibitem [{\citenamefont {Bolt}\ \emph {et~al.}(2016)\citenamefont {Bolt},
  \citenamefont {Duclos-Cianci}, \citenamefont {Poulin},\ and\ \citenamefont
  {Stace}}]{Bolt2016}%
  \BibitemOpen
  \bibfield  {author} {\bibinfo {author} {\bibfnamefont {A.}~\bibnamefont
  {Bolt}}, \bibinfo {author} {\bibfnamefont {G.}~\bibnamefont {Duclos-Cianci}},
  \bibinfo {author} {\bibfnamefont {D.}~\bibnamefont {Poulin}},\ and\ \bibinfo
  {author} {\bibfnamefont {T.~M.}\ \bibnamefont {Stace}},\ }\href
  {https://doi.org/10.1103/PhysRevLett.117.070501} {\bibfield  {journal}
  {\bibinfo  {journal} {Phys. Rev. Lett.}\ }\textbf {\bibinfo {volume} {117}},\
  \bibinfo {pages} {070501} (\bibinfo {year} {2016})}\BibitemShut {NoStop}%
\bibitem [{\citenamefont {Walshe}\ \emph {et~al.}(2024)\citenamefont {Walshe},
  \citenamefont {Baragiola}, \citenamefont {Ferretti}, \citenamefont {Gefaell},
  \citenamefont {Vasmer}, \citenamefont {Weil}, \citenamefont {Matsuura},
  \citenamefont {Jaeken}, \citenamefont {Pantaleoni}, \citenamefont {Han},
  \citenamefont {Menicucci}, \citenamefont {Tzitrin},\ and\ \citenamefont
  {Alexander}}]{walshe2024}%
  \BibitemOpen
  \bibfield  {author} {\bibinfo {author} {\bibfnamefont {B.~W.}\ \bibnamefont
  {Walshe}}, \bibinfo {author} {\bibfnamefont {B.~Q.}\ \bibnamefont
  {Baragiola}}, \bibinfo {author} {\bibfnamefont {H.}~\bibnamefont {Ferretti}},
  \bibinfo {author} {\bibfnamefont {J.}~\bibnamefont {Gefaell}}, \bibinfo
  {author} {\bibfnamefont {M.}~\bibnamefont {Vasmer}}, \bibinfo {author}
  {\bibfnamefont {R.}~\bibnamefont {Weil}}, \bibinfo {author} {\bibfnamefont
  {T.}~\bibnamefont {Matsuura}}, \bibinfo {author} {\bibfnamefont
  {T.}~\bibnamefont {Jaeken}}, \bibinfo {author} {\bibfnamefont
  {G.}~\bibnamefont {Pantaleoni}}, \bibinfo {author} {\bibfnamefont
  {Z.}~\bibnamefont {Han}}, \bibinfo {author} {\bibfnamefont {N.~C.}\
  \bibnamefont {Menicucci}}, \bibinfo {author} {\bibfnamefont {I.}~\bibnamefont
  {Tzitrin}},\ and\ \bibinfo {author} {\bibfnamefont {R.~N.}\ \bibnamefont
  {Alexander}},\ }\href {https://arxiv.org/abs/2408.04126} {\bibinfo {title}
  {Linear-optical quantum computation with arbitrary error-correcting codes}}
  (\bibinfo {year} {2024}),\ \Eprint {https://arxiv.org/abs/2408.04126}
  {arXiv:2408.04126 [quant-ph]} \BibitemShut {NoStop}%
\bibitem [{\citenamefont {Litinski}(2019)}]{Litinski2019}%
  \BibitemOpen
  \bibfield  {author} {\bibinfo {author} {\bibfnamefont {D.}~\bibnamefont
  {Litinski}},\ }\href {https://doi.org/10.22331/q-2019-03-05-128} {\bibfield
  {journal} {\bibinfo  {journal} {{Quantum}}\ }\textbf {\bibinfo {volume}
  {3}},\ \bibinfo {pages} {128} (\bibinfo {year} {2019})}\BibitemShut {NoStop}%
\bibitem [{\citenamefont {Bravyi}\ and\ \citenamefont
  {Kitaev}(2005)}]{Bravyi2005}%
  \BibitemOpen
  \bibfield  {author} {\bibinfo {author} {\bibfnamefont {S.}~\bibnamefont
  {Bravyi}}\ and\ \bibinfo {author} {\bibfnamefont {A.}~\bibnamefont
  {Kitaev}},\ }\href {https://doi.org/10.1103/PhysRevA.71.022316} {\bibfield
  {journal} {\bibinfo  {journal} {Phys. Rev. A}\ }\textbf {\bibinfo {volume}
  {71}},\ \bibinfo {pages} {022316} (\bibinfo {year} {2005})}\BibitemShut
  {NoStop}%
\bibitem [{\citenamefont {Horsman}\ \emph {et~al.}(2012)\citenamefont
  {Horsman}, \citenamefont {Fowler}, \citenamefont {Devitt},\ and\
  \citenamefont {Meter}}]{Horsman_2012}%
  \BibitemOpen
  \bibfield  {author} {\bibinfo {author} {\bibfnamefont {C.}~\bibnamefont
  {Horsman}}, \bibinfo {author} {\bibfnamefont {A.~G.}\ \bibnamefont {Fowler}},
  \bibinfo {author} {\bibfnamefont {S.}~\bibnamefont {Devitt}},\ and\ \bibinfo
  {author} {\bibfnamefont {R.~V.}\ \bibnamefont {Meter}},\ }\href
  {https://doi.org/10.1088/1367-2630/14/12/123011} {\bibfield  {journal}
  {\bibinfo  {journal} {New Journal of Physics}\ }\textbf {\bibinfo {volume}
  {14}},\ \bibinfo {pages} {123011} (\bibinfo {year} {2012})}\BibitemShut
  {NoStop}%
\bibitem [{\citenamefont {Landahl}\ and\ \citenamefont
  {Ryan-Anderson}(2014)}]{landahl2014quantum}%
  \BibitemOpen
  \bibfield  {author} {\bibinfo {author} {\bibfnamefont {A.~J.}\ \bibnamefont
  {Landahl}}\ and\ \bibinfo {author} {\bibfnamefont {C.}~\bibnamefont
  {Ryan-Anderson}},\ }\href@noop {} {\bibinfo {title} {Quantum computing by
  color-code lattice surgery}} (\bibinfo {year} {2014}),\ \Eprint
  {https://arxiv.org/abs/1407.5103} {arXiv:1407.5103 [quant-ph]} \BibitemShut
  {NoStop}%
\bibitem [{\citenamefont {Litinski}\ and\ \citenamefont
  {Oppen}(2018)}]{Litinski2018latticesurgery}%
  \BibitemOpen
  \bibfield  {author} {\bibinfo {author} {\bibfnamefont {D.}~\bibnamefont
  {Litinski}}\ and\ \bibinfo {author} {\bibfnamefont {F.~v.}\ \bibnamefont
  {Oppen}},\ }\href {https://doi.org/10.22331/q-2018-05-04-62} {\bibfield
  {journal} {\bibinfo  {journal} {{Quantum}}\ }\textbf {\bibinfo {volume}
  {2}},\ \bibinfo {pages} {62} (\bibinfo {year} {2018})}\BibitemShut {NoStop}%
\bibitem [{\citenamefont {Poulsen~Nautrup}\ \emph {et~al.}(2017)\citenamefont
  {Poulsen~Nautrup}, \citenamefont {Friis},\ and\ \citenamefont
  {Briegel}}]{PoulsenNautrup2017}%
  \BibitemOpen
  \bibfield  {author} {\bibinfo {author} {\bibfnamefont {H.}~\bibnamefont
  {Poulsen~Nautrup}}, \bibinfo {author} {\bibfnamefont {N.}~\bibnamefont
  {Friis}},\ and\ \bibinfo {author} {\bibfnamefont {H.~J.}\ \bibnamefont
  {Briegel}},\ }\href {https://doi.org/10.1038/s41467-017-01418-2} {\bibfield
  {journal} {\bibinfo  {journal} {Nature Communications}\ }\textbf {\bibinfo
  {volume} {8}},\ \bibinfo {pages} {1321} (\bibinfo {year} {2017})}\BibitemShut
  {NoStop}%
\bibitem [{\citenamefont {Huang}\ \emph {et~al.}(2023)\citenamefont {Huang},
  \citenamefont {Jochym-O'Connor},\ and\ \citenamefont {Yoder}}]{Huang2023}%
  \BibitemOpen
  \bibfield  {author} {\bibinfo {author} {\bibfnamefont {S.}~\bibnamefont
  {Huang}}, \bibinfo {author} {\bibfnamefont {T.}~\bibnamefont
  {Jochym-O'Connor}},\ and\ \bibinfo {author} {\bibfnamefont {T.~J.}\
  \bibnamefont {Yoder}},\ }\href {https://doi.org/10.1103/PRXQuantum.4.030301}
  {\bibfield  {journal} {\bibinfo  {journal} {PRX Quantum}\ }\textbf {\bibinfo
  {volume} {4}},\ \bibinfo {pages} {030301} (\bibinfo {year}
  {2023})}\BibitemShut {NoStop}%
\bibitem [{\citenamefont {Cohen}\ \emph {et~al.}(2022)\citenamefont {Cohen},
  \citenamefont {Kim}, \citenamefont {Bartlett},\ and\ \citenamefont
  {Brown}}]{Cohen2021}%
  \BibitemOpen
  \bibfield  {author} {\bibinfo {author} {\bibfnamefont {L.~Z.}\ \bibnamefont
  {Cohen}}, \bibinfo {author} {\bibfnamefont {I.~H.}\ \bibnamefont {Kim}},
  \bibinfo {author} {\bibfnamefont {S.~D.}\ \bibnamefont {Bartlett}},\ and\
  \bibinfo {author} {\bibfnamefont {B.~J.}\ \bibnamefont {Brown}},\ }\href
  {https://doi.org/10.1126/sciadv.abn1717} {\bibfield  {journal} {\bibinfo
  {journal} {Science Advances}\ }\textbf {\bibinfo {volume} {8}},\ \bibinfo
  {pages} {eabn1717} (\bibinfo {year} {2022})},\ \Eprint
  {https://arxiv.org/abs/https://www.science.org/doi/pdf/10.1126/sciadv.abn1717}
  {https://www.science.org/doi/pdf/10.1126/sciadv.abn1717} \BibitemShut
  {NoStop}%
\bibitem [{\citenamefont {Cowtan}\ and\ \citenamefont
  {Burton}(2024)}]{cowtan2023css}%
  \BibitemOpen
  \bibfield  {author} {\bibinfo {author} {\bibfnamefont {A.}~\bibnamefont
  {Cowtan}}\ and\ \bibinfo {author} {\bibfnamefont {S.}~\bibnamefont
  {Burton}},\ }\href@noop {} {\bibfield  {journal} {\bibinfo  {journal}
  {Quantum}\ }\textbf {\bibinfo {volume} {8}},\ \bibinfo {pages} {1344}
  (\bibinfo {year} {2024})}\BibitemShut {NoStop}%
\bibitem [{\citenamefont {Cross}\ \emph {et~al.}(2024)\citenamefont {Cross},
  \citenamefont {He}, \citenamefont {Rall},\ and\ \citenamefont
  {Yoder}}]{cross2024linearsize}%
  \BibitemOpen
  \bibfield  {author} {\bibinfo {author} {\bibfnamefont {A.}~\bibnamefont
  {Cross}}, \bibinfo {author} {\bibfnamefont {Z.}~\bibnamefont {He}}, \bibinfo
  {author} {\bibfnamefont {P.}~\bibnamefont {Rall}},\ and\ \bibinfo {author}
  {\bibfnamefont {T.}~\bibnamefont {Yoder}},\ }\href
  {https://doi.org/10.48550/arXiv.2407.18393} {\bibinfo {title} {Linear-{{Size
  Ancilla Systems}} for {{Logical Measurements}} in {{QLDPC Codes}}}} (\bibinfo
  {year} {2024}),\ \Eprint {https://arxiv.org/abs/2407.18393} {arXiv:2407.18393
  [quant-ph]} \BibitemShut {NoStop}%
\bibitem [{\citenamefont {Weibel}(1994)}]{weibel1994introduction}%
  \BibitemOpen
  \bibfield  {author} {\bibinfo {author} {\bibfnamefont {C.~A.}\ \bibnamefont
  {Weibel}},\ }\href@noop {} {\emph {\bibinfo {title} {An introduction to
  homological algebra}}},\ \bibinfo {series} {Cambridge Studies in Advanced
  Mathematics}\ No.~\bibinfo {number} {38}\ (\bibinfo  {publisher} {Cambridge
  university press},\ \bibinfo {address} {Cambridge},\ \bibinfo {year}
  {1994})\BibitemShut {NoStop}%
\bibitem [{\citenamefont {Hastings}(2016)}]{hastings2016}%
  \BibitemOpen
  \bibfield  {author} {\bibinfo {author} {\bibfnamefont {M.~B.}\ \bibnamefont
  {Hastings}},\ }\href@noop {} {\bibinfo {title} {Weight {{Reduction}} for
  {{Quantum Codes}}}} (\bibinfo {year} {2016}),\ \Eprint
  {https://arxiv.org/abs/1611.03790} {arxiv:1611.03790} \BibitemShut {NoStop}%
\bibitem [{\citenamefont {Hastings}(2023)}]{hastings2023quantum}%
  \BibitemOpen
  \bibfield  {author} {\bibinfo {author} {\bibfnamefont {M.~B.}\ \bibnamefont
  {Hastings}},\ }\href@noop {} {\bibinfo {title} {On quantum weight reduction}}
  (\bibinfo {year} {2023}),\ \Eprint {https://arxiv.org/abs/2102.10030}
  {arxiv:2102.10030 [quant-ph]} \BibitemShut {NoStop}%
\bibitem [{\citenamefont {Wills}\ \emph {et~al.}(2023)\citenamefont {Wills},
  \citenamefont {Lin},\ and\ \citenamefont {Hsieh}}]{wills2023tradeoff}%
  \BibitemOpen
  \bibfield  {author} {\bibinfo {author} {\bibfnamefont {A.}~\bibnamefont
  {Wills}}, \bibinfo {author} {\bibfnamefont {T.-C.}\ \bibnamefont {Lin}},\
  and\ \bibinfo {author} {\bibfnamefont {M.-H.}\ \bibnamefont {Hsieh}},\
  }\href@noop {} {\bibinfo {title} {Tradeoff constructions for quantum locally
  testable codes}} (\bibinfo {year} {2023}),\ \Eprint
  {https://arxiv.org/abs/2309.05541} {arXiv:2309.05541} \BibitemShut {NoStop}%
\bibitem [{\citenamefont {Gottesman}\ \emph {et~al.}(2001)\citenamefont
  {Gottesman}, \citenamefont {Kitaev},\ and\ \citenamefont {Preskill}}]{GKP}%
  \BibitemOpen
  \bibfield  {author} {\bibinfo {author} {\bibfnamefont {D.}~\bibnamefont
  {Gottesman}}, \bibinfo {author} {\bibfnamefont {A.}~\bibnamefont {Kitaev}},\
  and\ \bibinfo {author} {\bibfnamefont {J.}~\bibnamefont {Preskill}},\ }\href
  {https://doi.org/10.1103/PhysRevA.64.012310} {\bibfield  {journal} {\bibinfo
  {journal} {Phys. Rev. A}\ }\textbf {\bibinfo {volume} {64}},\ \bibinfo
  {pages} {012310} (\bibinfo {year} {2001})}\BibitemShut {NoStop}%
\bibitem [{\citenamefont {Rotman}(2009)}]{rotman2009introduction}%
  \BibitemOpen
  \bibfield  {author} {\bibinfo {author} {\bibfnamefont {J.}~\bibnamefont
  {Rotman}},\ }\href@noop {} {\emph {\bibinfo {title} {An introduction to
  homological algebra}}},\ Vol.~\bibinfo {volume} {2}\ (\bibinfo  {publisher}
  {Springer},\ \bibinfo {address} {New York},\ \bibinfo {year}
  {2009})\BibitemShut {NoStop}%
\bibitem [{\citenamefont {Tillich}\ and\ \citenamefont
  {Zémor}(2014{\natexlab{b}})}]{tillich2014quantumLDPC}%
  \BibitemOpen
  \bibfield  {author} {\bibinfo {author} {\bibfnamefont {J.-P.}\ \bibnamefont
  {Tillich}}\ and\ \bibinfo {author} {\bibfnamefont {G.}~\bibnamefont
  {Zémor}},\ }\href {https://doi.org/10.1109/TIT.2013.2292061} {\bibfield
  {journal} {\bibinfo  {journal} {IEEE Transactions on Information Theory}\
  }\textbf {\bibinfo {volume} {60}},\ \bibinfo {pages} {1193} (\bibinfo {year}
  {2014}{\natexlab{b}})}\BibitemShut {NoStop}%
\bibitem [{\citenamefont {Zhang}\ and\ \citenamefont
  {Li}(2024)}]{zhang2024timeefficient}%
  \BibitemOpen
  \bibfield  {author} {\bibinfo {author} {\bibfnamefont {G.}~\bibnamefont
  {Zhang}}\ and\ \bibinfo {author} {\bibfnamefont {Y.}~\bibnamefont {Li}},\
  }\href@noop {} {\bibinfo {title} {Time-efficient logical operations on
  quantum {{LDPC}} codes}} (\bibinfo {year} {2024}),\ \Eprint
  {https://arxiv.org/abs/2408.01339} {arXiv:2408.01339 [quant-ph]} \BibitemShut
  {NoStop}%
\bibitem [{\citenamefont {Randall}(1993)}]{randall1993efficient}%
  \BibitemOpen
  \bibfield  {author} {\bibinfo {author} {\bibfnamefont {D.}~\bibnamefont
  {Randall}},\ }\href {https://doi.org/10.1002/rsa.3240040108} {\bibfield
  {journal} {\bibinfo  {journal} {Random Structures \& Algorithms}\ }\textbf
  {\bibinfo {volume} {4}},\ \bibinfo {pages} {111} (\bibinfo {year}
  {1993})}\BibitemShut {NoStop}%
\bibitem [{\citenamefont {Freedman}\ and\ \citenamefont
  {Hastings}(2021)}]{freedman2021building}%
  \BibitemOpen
  \bibfield  {author} {\bibinfo {author} {\bibfnamefont {M.}~\bibnamefont
  {Freedman}}\ and\ \bibinfo {author} {\bibfnamefont {M.}~\bibnamefont
  {Hastings}},\ }\href@noop {} {\bibfield  {journal} {\bibinfo  {journal}
  {Geometric and Functional Analysis}\ }\textbf {\bibinfo {volume} {31}},\
  \bibinfo {pages} {855} (\bibinfo {year} {2021})}\BibitemShut {NoStop}%
\bibitem [{\citenamefont {Raveendran}\ \emph {et~al.}(2022)\citenamefont
  {Raveendran}, \citenamefont {Rengaswamy}, \citenamefont {Rozpedek},
  \citenamefont {Raina}, \citenamefont {Jiang},\ and\ \citenamefont
  {Vasi{\'{c}}}}]{Raveendran2022finiterateqldpcgkp}%
  \BibitemOpen
  \bibfield  {author} {\bibinfo {author} {\bibfnamefont {N.}~\bibnamefont
  {Raveendran}}, \bibinfo {author} {\bibfnamefont {N.}~\bibnamefont
  {Rengaswamy}}, \bibinfo {author} {\bibfnamefont {F.}~\bibnamefont
  {Rozpedek}}, \bibinfo {author} {\bibfnamefont {A.}~\bibnamefont {Raina}},
  \bibinfo {author} {\bibfnamefont {L.}~\bibnamefont {Jiang}},\ and\ \bibinfo
  {author} {\bibfnamefont {B.}~\bibnamefont {Vasi{\'{c}}}},\ }\href
  {https://doi.org/10.22331/q-2022-07-20-767} {\bibfield  {journal} {\bibinfo
  {journal} {{Quantum}}\ }\textbf {\bibinfo {volume} {6}},\ \bibinfo {pages}
  {767} (\bibinfo {year} {2022})}\BibitemShut {NoStop}%
\bibitem [{\citenamefont {Roffe}\ \emph {et~al.}(2020)\citenamefont {Roffe},
  \citenamefont {White}, \citenamefont {Burton},\ and\ \citenamefont
  {Campbell}}]{Roffe_QLPDC_decoding}%
  \BibitemOpen
  \bibfield  {author} {\bibinfo {author} {\bibfnamefont {J.}~\bibnamefont
  {Roffe}}, \bibinfo {author} {\bibfnamefont {D.~R.}\ \bibnamefont {White}},
  \bibinfo {author} {\bibfnamefont {S.}~\bibnamefont {Burton}},\ and\ \bibinfo
  {author} {\bibfnamefont {E.}~\bibnamefont {Campbell}},\ }\href
  {https://doi.org/10.1103/PhysRevResearch.2.043423} {\bibfield  {journal}
  {\bibinfo  {journal} {Phys. Rev. Res.}\ }\textbf {\bibinfo {volume} {2}},\
  \bibinfo {pages} {043423} (\bibinfo {year} {2020})}\BibitemShut {NoStop}%
\bibitem [{\citenamefont {Bravyi}\ \emph {et~al.}(2024)\citenamefont {Bravyi},
  \citenamefont {Cross}, \citenamefont {Gambetta}, \citenamefont {Maslov},
  \citenamefont {Rall},\ and\ \citenamefont
  {Yoder}}]{Bravyi2024_IBM_GeneralizedBicycle}%
  \BibitemOpen
  \bibfield  {author} {\bibinfo {author} {\bibfnamefont {S.}~\bibnamefont
  {Bravyi}}, \bibinfo {author} {\bibfnamefont {A.~W.}\ \bibnamefont {Cross}},
  \bibinfo {author} {\bibfnamefont {J.~M.}\ \bibnamefont {Gambetta}}, \bibinfo
  {author} {\bibfnamefont {D.}~\bibnamefont {Maslov}}, \bibinfo {author}
  {\bibfnamefont {P.}~\bibnamefont {Rall}},\ and\ \bibinfo {author}
  {\bibfnamefont {T.~J.}\ \bibnamefont {Yoder}},\ }\href
  {https://doi.org/10.1038/s41586-024-07107-7} {\bibfield  {journal} {\bibinfo
  {journal} {Nature}\ }\textbf {\bibinfo {volume} {627}},\ \bibinfo {pages}
  {778} (\bibinfo {year} {2024})}\BibitemShut {NoStop}%
\bibitem [{\citenamefont {MacKay}\ and\ \citenamefont
  {Neal}(1996)}]{MacKay_Neal1996}%
  \BibitemOpen
  \bibfield  {author} {\bibinfo {author} {\bibfnamefont {D.}~\bibnamefont
  {MacKay}}\ and\ \bibinfo {author} {\bibfnamefont {R.}~\bibnamefont {Neal}},\
  }\href
  {https://digital-library.theiet.org/content/journals/10.1049/el_19961141}
  {\bibfield  {journal} {\bibinfo  {journal} {Electronics Letters}\ }\textbf
  {\bibinfo {volume} {32}},\ \bibinfo {pages} {1645} (\bibinfo {year}
  {1996})}\BibitemShut {NoStop}%
\bibitem [{\citenamefont {Agresti}\ and\ \citenamefont
  {Coull}(1998)}]{agresti1998approximate}%
  \BibitemOpen
  \bibfield  {author} {\bibinfo {author} {\bibfnamefont {A.}~\bibnamefont
  {Agresti}}\ and\ \bibinfo {author} {\bibfnamefont {B.~A.}\ \bibnamefont
  {Coull}},\ }\href@noop {} {\bibfield  {journal} {\bibinfo  {journal} {The
  American Statistician}\ }\textbf {\bibinfo {volume} {52}},\ \bibinfo {pages}
  {119} (\bibinfo {year} {1998})}\BibitemShut {NoStop}%
\bibitem [{\citenamefont {Brown}\ \emph {et~al.}(2001)\citenamefont {Brown},
  \citenamefont {Cai},\ and\ \citenamefont {DasGupta}}]{Brown2001}%
  \BibitemOpen
  \bibfield  {author} {\bibinfo {author} {\bibfnamefont {L.~D.}\ \bibnamefont
  {Brown}}, \bibinfo {author} {\bibfnamefont {T.~T.}\ \bibnamefont {Cai}},\
  and\ \bibinfo {author} {\bibfnamefont {A.}~\bibnamefont {DasGupta}},\ }\href
  {http://www.jstor.org/stable/2676784} {\bibfield  {journal} {\bibinfo
  {journal} {Statistical Science}\ }\textbf {\bibinfo {volume} {16}},\ \bibinfo
  {pages} {101} (\bibinfo {year} {2001})}\BibitemShut {NoStop}%
\bibitem [{\citenamefont {Williamson}\ and\ \citenamefont
  {Yoder}(2024)}]{williamson2024low}%
  \BibitemOpen
  \bibfield  {author} {\bibinfo {author} {\bibfnamefont {D.~J.}\ \bibnamefont
  {Williamson}}\ and\ \bibinfo {author} {\bibfnamefont {T.~J.}\ \bibnamefont
  {Yoder}},\ }\href {https://arxiv.org/abs/2410.02213} {\bibinfo {title}
  {Low-overhead fault-tolerant quantum computation by gauging logical
  operators}} (\bibinfo {year} {2024}),\ \Eprint
  {https://arxiv.org/abs/2410.02213} {arXiv:2410.02213 [quant-ph]} \BibitemShut
  {NoStop}%
\bibitem [{\citenamefont {Ramette}\ \emph {et~al.}(2024)\citenamefont
  {Ramette}, \citenamefont {Sinclair}, \citenamefont {Breuckmann},\ and\
  \citenamefont {Vuleti{\'c}}}]{ramette2023faulttolerant}%
  \BibitemOpen
  \bibfield  {author} {\bibinfo {author} {\bibfnamefont {J.}~\bibnamefont
  {Ramette}}, \bibinfo {author} {\bibfnamefont {J.}~\bibnamefont {Sinclair}},
  \bibinfo {author} {\bibfnamefont {N.~P.}\ \bibnamefont {Breuckmann}},\ and\
  \bibinfo {author} {\bibfnamefont {V.}~\bibnamefont {Vuleti{\'c}}},\
  }\href@noop {} {\bibfield  {journal} {\bibinfo  {journal} {npj Quantum
  Information}\ }\textbf {\bibinfo {volume} {10}},\ \bibinfo {pages} {58}
  (\bibinfo {year} {2024})}\BibitemShut {NoStop}%
\bibitem [{\citenamefont {Bomb\'{\i}n}(2015)}]{PhysRevX.5.031043}%
  \BibitemOpen
  \bibfield  {author} {\bibinfo {author} {\bibfnamefont {H.}~\bibnamefont
  {Bomb\'{\i}n}},\ }\href {https://doi.org/10.1103/PhysRevX.5.031043}
  {\bibfield  {journal} {\bibinfo  {journal} {Phys. Rev. X}\ }\textbf {\bibinfo
  {volume} {5}},\ \bibinfo {pages} {031043} (\bibinfo {year}
  {2015})}\BibitemShut {NoStop}%
\bibitem [{\citenamefont {Campbell}(2019)}]{Campbell_2019}%
  \BibitemOpen
  \bibfield  {author} {\bibinfo {author} {\bibfnamefont {E.~T.}\ \bibnamefont
  {Campbell}},\ }\href {https://doi.org/10.1088/2058-9565/aafc8f} {\bibfield
  {journal} {\bibinfo  {journal} {Quantum Science and Technology}\ }\textbf
  {\bibinfo {volume} {4}},\ \bibinfo {pages} {025006} (\bibinfo {year}
  {2019})}\BibitemShut {NoStop}%
\bibitem [{\citenamefont {Quintavalle}\ \emph {et~al.}(2021)\citenamefont
  {Quintavalle}, \citenamefont {Vasmer}, \citenamefont {Roffe},\ and\
  \citenamefont {Campbell}}]{PRXQuantum.2.020340}%
  \BibitemOpen
  \bibfield  {author} {\bibinfo {author} {\bibfnamefont {A.~O.}\ \bibnamefont
  {Quintavalle}}, \bibinfo {author} {\bibfnamefont {M.}~\bibnamefont {Vasmer}},
  \bibinfo {author} {\bibfnamefont {J.}~\bibnamefont {Roffe}},\ and\ \bibinfo
  {author} {\bibfnamefont {E.~T.}\ \bibnamefont {Campbell}},\ }\href
  {https://doi.org/10.1103/PRXQuantum.2.020340} {\bibfield  {journal} {\bibinfo
   {journal} {PRX Quantum}\ }\textbf {\bibinfo {volume} {2}},\ \bibinfo {pages}
  {020340} (\bibinfo {year} {2021})}\BibitemShut {NoStop}%
\bibitem [{\citenamefont {Ponce}\ \emph {et~al.}(2019)\citenamefont {Ponce},
  \citenamefont {van Zon}, \citenamefont {Northrup}, \citenamefont {Gruner},
  \citenamefont {Chen}, \citenamefont {Ertinaz}, \citenamefont {Fedoseev},
  \citenamefont {Groer}, \citenamefont {Mao}, \citenamefont {Mundim},
  \citenamefont {Nolta}, \citenamefont {Pinto}, \citenamefont {Saldarriaga},
  \citenamefont {Slavnic}, \citenamefont {Spence}, \citenamefont {Yu},\ and\
  \citenamefont {Peltier}}]{niagara}%
  \BibitemOpen
  \bibfield  {author} {\bibinfo {author} {\bibfnamefont {M.}~\bibnamefont
  {Ponce}}, \bibinfo {author} {\bibfnamefont {R.}~\bibnamefont {van Zon}},
  \bibinfo {author} {\bibfnamefont {S.}~\bibnamefont {Northrup}}, \bibinfo
  {author} {\bibfnamefont {D.}~\bibnamefont {Gruner}}, \bibinfo {author}
  {\bibfnamefont {J.}~\bibnamefont {Chen}}, \bibinfo {author} {\bibfnamefont
  {F.}~\bibnamefont {Ertinaz}}, \bibinfo {author} {\bibfnamefont
  {A.}~\bibnamefont {Fedoseev}}, \bibinfo {author} {\bibfnamefont
  {L.}~\bibnamefont {Groer}}, \bibinfo {author} {\bibfnamefont
  {F.}~\bibnamefont {Mao}}, \bibinfo {author} {\bibfnamefont {B.~C.}\
  \bibnamefont {Mundim}}, \bibinfo {author} {\bibfnamefont {M.}~\bibnamefont
  {Nolta}}, \bibinfo {author} {\bibfnamefont {J.}~\bibnamefont {Pinto}},
  \bibinfo {author} {\bibfnamefont {M.}~\bibnamefont {Saldarriaga}}, \bibinfo
  {author} {\bibfnamefont {V.}~\bibnamefont {Slavnic}}, \bibinfo {author}
  {\bibfnamefont {E.}~\bibnamefont {Spence}}, \bibinfo {author} {\bibfnamefont
  {C.-H.}\ \bibnamefont {Yu}},\ and\ \bibinfo {author} {\bibfnamefont {W.~R.}\
  \bibnamefont {Peltier}},\ }in\ \href
  {https://doi.org/10.1145/3332186.3332195} {\emph {\bibinfo {booktitle}
  {Proceedings of the Practice and Experience in Advanced Research Computing on
  Rise of the Machines (Learning)}}},\ \bibinfo {series and number} {PEARC
  '19}\ (\bibinfo  {publisher} {Association for Computing Machinery},\ \bibinfo
  {address} {New York, NY, USA},\ \bibinfo {year} {2019})\BibitemShut {NoStop}%
\bibitem [{\citenamefont {Swaroop}\ \emph {et~al.}(2024)\citenamefont
  {Swaroop}, \citenamefont {Jochym-O'Connor},\ and\ \citenamefont
  {Yoder}}]{swaroop2024universaladaptersquantumldpc}%
  \BibitemOpen
  \bibfield  {author} {\bibinfo {author} {\bibfnamefont {E.}~\bibnamefont
  {Swaroop}}, \bibinfo {author} {\bibfnamefont {T.}~\bibnamefont
  {Jochym-O'Connor}},\ and\ \bibinfo {author} {\bibfnamefont {T.~J.}\
  \bibnamefont {Yoder}},\ }\href {https://arxiv.org/abs/2410.03628} {\bibinfo
  {title} {Universal adapters between quantum ldpc codes}} (\bibinfo {year}
  {2024}),\ \Eprint {https://arxiv.org/abs/2410.03628} {arXiv:2410.03628
  [quant-ph]} \BibitemShut {NoStop}%
\end{thebibliography}%

\end{document}